\documentclass[twoside,11pt]{article}
\usepackage{jair,theapa,rawfonts}

\jairheading{62}{2018}{??-??}{06/17}{08/18}
\ShortHeadings{Querying Log Data with Metric Temporal Logic}
{Brandt, Kalayc{\i}, Ryzhikov, Xiao, \& Zakharyaschev}
\firstpageno{1}

\usepackage{times}
\usepackage{helvet}
\usepackage{courier}
\usepackage{placeins}

\usepackage{listings}
\makeatletter
\let\c@table\c@figure
\makeatother

\usepackage{multirow, array, diagbox}

\usepackage{xcolor}
\usepackage{tikz}
\usetikzlibrary{fit}
\usetikzlibrary{calc}
\usetikzlibrary{shapes,decorations,shadows,calc}
\usetikzlibrary{decorations.pathmorphing, decorations.markings,decorations.pathreplacing,decorations.shapes}

\usetikzlibrary{patterns}
\usepgflibrary{arrows}
\usetikzlibrary{trees}
\usetikzlibrary{backgrounds, arrows, spy}
\usetikzlibrary{circuits.logic.US}

\usepackage{latexsym}
\usepackage{amsmath,amsthm,amssymb,stmaryrd,bm}
\usepackage{fancyvrb}
\usepackage{url}
\usepackage{pifont} 

 \newtheorem{theorem}{Theorem}
 
 \newtheorem{example}[theorem]{Example}
 \newtheorem{lemma}[theorem]{Lemma}
 \newtheorem*{lemma*}{Lemma}
 \newtheorem{proposition}[theorem]{Proposition}

\usepackage{textcomp} 
\usepackage[T1]{fontenc} 
\usepackage{paralist}

\newcommand{\OWL}{\textsl{OWL\,2}}
\newcommand{\OWLQL}{\textsl{OWL\,2\,QL}}
\newcommand{\LTL}{\textsl{LTL}}

\newcommand{\nxt}{{\ensuremath\raisebox{0.25ex}{\text{\scriptsize$\bigcirc$}}}}
\newcommand{\ind}{\mathsf{ind}}
\newcommand{\Mmf}{\mathfrak{M}}
\newcommand{\avec}[1]{\boldsymbol{#1}}
\renewcommand{\int}{\mathsf{int}}
\newcommand{\q}{\boldsymbol q}
\newcommand{\rul}{{\boldsymbol r}}
\newcommand{\range}{\varrho}

\newcommand{\D}{\mathcal{D}}

\newcommand{\num}{\mathsf{num}}
\newcommand{\sect}{\mathsf{sec}}

\newcommand{\Si}{\mathbin{\mathcal{S}}}
\newcommand{\U}{\mathbin{\mathcal{U}}}

\newcommand{\diamondplus}{%
  \raisebox{-.15ex}{\begin{tikzpicture}
    \useasboundingbox (-0.7ex, -0.9ex) rectangle (0.7ex, 0.9ex);
    \node (w) at (-0.8ex,0) {};
    \node (e) at (+0.8ex,0) {};
    \node (s) at (0,-.9ex) {};
    \node (n) at (0,+.9ex) {};
    \draw (n.center) -- (e.center) -- (s.center) -- (w.center) -- (n.center);
    \draw (n.center) -- (s.center);
    \draw (e.center) -- (w.center);
  \end{tikzpicture}}}
\newcommand{\diamondminus}{%
  \raisebox{-.15ex}{\begin{tikzpicture}
    \useasboundingbox (-0.7ex, -0.9ex) rectangle (0.7ex, 0.9ex);
    \node (w) at (-0.8ex,0) {};
    \node (e) at (+0.8ex,0) {};
    \node (s) at (0,-.9ex) {};
    \node (n) at (0,+.9ex) {};
    \draw (n.center) -- (e.center) -- (s.center) -- (w.center) -- (n.center);
    \draw (e.center) -- (w.center);
  \end{tikzpicture}}}

\newcommand{\Rnext}{\nxt_{\!\scriptscriptstyle F}}
\newcommand{\Lnext}{\nxt_{\!\scriptscriptstyle P}}

\newcommand{\Rdiamond}{\Diamond_{\!\scriptscriptstyle F}}
\newcommand{\Ldiamond}{\Diamond_{\!\scriptscriptstyle P}}

\newcommand{\Rbox}{\rule{0pt}{1.4ex}\Box_{\!\scriptscriptstyle F}}
\newcommand{\Lbox}{\rule{0pt}{1.4ex}\Box_{\!\scriptscriptstyle P}}

\newcommand{\C}{\mathcal{C}}
\newcommand{\cl}{\mathsf{cl}}

\newcommand{\MTL}{\textsl{MTL}}
\newcommand{\hMTL}{\textsl{hornMTL}}
\newcommand{\dMTL}{\textsl{datalogMTL}}
\newcommand{\nrdMTL}{\textsl{datalog$_{\it nr}$MTL}}
\newcommand{\TB}{\ensuremath{\mathsf{TB}\_\mathsf{Sensor}}}

\newcommand{\Ci}{\mathop{\text{\ding{192}}}}
\newcommand{\Cii}{\mathop{\text{\ding{193}}}}

\newcommand{\pc}{\mathop{+^c}}
\newcommand{\po}{\mathop{+^o}}
\newcommand{\mc}{\mathop{-^c}}
\newcommand{\mo}{\mathop{-^o}}

\definecolor{myred}{HTML}{C00D2A}

\definecolor{myblue2}{RGB}{ 11, 144, 189 }
\colorlet{darkgreen}{green!50!black}
\definecolor{myblue}{HTML}{116699}
\definecolor{myviolet}{HTML}{3E1199}
\definecolor{myorange}{HTML}{cc6633}
\definecolor{mygray}{HTML}{cccccc}
\definecolor{mydarkgray}{HTML}{999999}
\colorlet{textgray}{black!80}

\begin{document}

\title{Querying Log Data with Metric Temporal Logic}

\author{\name Sebastian Brandt \email sebastian-philipp.brandt@siemens.com \\
\addr Siemens CT\\ M\"unchen, Germany
\AND
       \name Elem G\"uzel Kalayc{\i} \email kalayci@inf.unibz.it \\
       \addr KRDB Research Centre\\ Faculty of Computer Science\\ Free University of Bozen-Bolzano, Italy
\AND
       \name Vladislav Ryzhikov \email vlad@dcs.bbk.ac.uk\\
       \addr Department of Computer Science and Information Systems\\ Birkbeck, University of London, UK
\AND
       \name Guohui Xiao \email xiao@inf.unibz.it\\
       \addr KRDB Research Centre\\ Faculty of Computer Science\\ Free University of Bozen-Bolzano, Italy
       \AND
       \name Michael Zakharyaschev \email michael@dcs.bbk.ac.uk \\
       \addr Department of Computer Science and Information Systems\\ Birkbeck, University of London, UK}


\maketitle

\begin{abstract}
  We propose a novel framework for ontology-based access to temporal log data using a datalog extension $\dMTL$ of the Horn fragment of the metric temporal logic $\MTL$. We show that $\dMTL$ is {\sc ExpSpace}-complete  even with punctual intervals, in which case full $\MTL$ is known to be undecidable. We also prove that nonrecursive $\dMTL$ is {\sc PSpace}-complete for combined complexity and in AC$^0$ for data complexity.  We demonstrate by two real-world use cases that nonrecursive $\dMTL$ programs can express complex temporal concepts from typical user queries and thereby facilitate access to temporal log data. Our experiments with Siemens turbine data and MesoWest weather data show that $\dMTL$ ontology-mediated queries are efficient and scale on large datasets.
\end{abstract}


\section{Introduction}

In this paper, we present a new ontology-based framework for querying temporal log data. We begin by outlining this framework in the context of data gathering and analysis at Siemens, a leading manufacturer and supplier of systems for power generation, power transmission, medical diagnosis, and industry automation.

\subsection{Data gathering at Siemens}
For the Siemens equipment, analytics services are usually delivered by remote diagnostic centres that store data from the relevant industrial sites or individual equipment around the globe.
The analytics provided at these centres falls into three categories: descriptive, predictive, and prescriptive. Descriptive analytics describes or quantifies in
detail what has happened after an event. Predictive analytics aims to
anticipate events before they occur and provide a window of
opportunity for countermeasures. Prescriptive analytics aims to automate the process of suggesting underlying reasons for the predicted events and carrying out appropriate countermeasures.
All these types of analytics heavily rely on the ability to recognise interesting events using sensor measurements or other machine data such as the power output of a gas turbine, its maximum rotor speed, average exhaust temperature, etc.
For example, a service engineer at a Siemens remote diagnostic centre could be interested in active power trips of the turbine, that is, events when
%
\begin{description}
\item[($\mathsf{ActivePowerTrip}$)] the active power was above 1.5MW
  for a period of at least 10 seconds, maximum 3 seconds after which
  there was a period of at least one minute where the active power was
  below 0.15MW.
\end{description}
Under the standard workflow, when facing the task of finding the
active power trips of the turbine, the engineer would call an IT
expert who would then produce a specific script (in a proprietary
signal processing language developed by Siemens) such as
\begin{align*}
&\mathtt{message}(``\mathtt{active\ power\ trip}")  =\\
&\quad \$t1: \mathtt{eval}( >, \#\mathtt{activePower}, 1.5 ) :\\
&\quad\quad\quad \mathtt{for}( >= 10s)\\
&\quad\quad\quad\text{\&\&}\\
&\quad\quad\quad \mathtt{eval}( <, \#\mathtt{activePower}, 0.15 ) :\\
&\quad\quad\quad \mathtt{start}( \mathtt{after}[ 0s, 3s ] \$t1:\mathtt{end} ):\\
&\quad\quad\quad \mathtt{for}( >= 1m);
\end{align*}
for the turbine aggregated data stored in a table \TB, which looks as follows:
\begin{center}
\begin{tabular}{|c|c|c|c|c|c|}
\hline
turbineId & dateTime & activePower & rotorSpeed & mainFlame & \dots\\\hline
 & & \dots &  & & \\
tb0 & 2015-04-04\;12:20:48 & 2 & 1550 & 0 & \\
tb0 & 2015-04-04\;12:20:49 & 1.8 & 1400 & null & \\
tb0 & 2015-04-04\;12:20:52 & 1.7 & 1350 & 1 & \\
 & & \dots & & & \\\hline
\end{tabular}
\end{center}
The result of running the script is a log with records such as
\begin{center}
``2015-04-04\;12:22:17 $\mathtt{active\ power\ trip}$ tb0''
\end{center}
where information about all the events is accumulated.

When facing the same task but for a different turbine, the engineer may have to call the IT expert once again because different models of turbines and sensors may have different log/database formats. Moreover, the storage platform for the sensor data often changes (thus, currently Siemens are pondering over migrating certain data to a cloud-based storage).  Maintaining a set of scripts, one for each data source, does not provide an efficient solution since a query such as `find all the turbines that had an active power trip in May 2017' would require an intermediate database with integrated data of active power trips.
Another difficulty is that the definitions of events the engineer is interested in can also change. Some changes are minor, say the pressure threshold or the number of seconds in the active power trip definition, but some could be more substantial, such as `find the active power trips that were followed by a high pressure within 3 minutes that lasted for 30 seconds'\!. This modification would require rewriting the script above into a much longer one rather than using it as a module in the new definition.

The permanent involvement of an IT expert familiar with database technology incurs high costs for Siemens, and data gathering accounts for a major part of the time the service engineers spend at Siemens remote diagnostic centres, most of which due to the indirect access to data.

\subsection{Ontology-based data access}
Ontology-based data access (OBDA for short) offers a  different workflow that excludes the IT middleman from data gathering~\cite{PLCD*08}; consult also the recent survey by \citeA{IJCAI-18}. In a nutshell, the OBDA workflow in the Siemens context looks as follows. Domain experts develop and maintain an ontology that contains terms for the events the engineers may be interested in. IT experts develop and maintain  mappings that relate these terms to the database schemas. The engineer can now use familiar terms from the ontology and a graphical tool such as OptiqueVQS~\cite{VQSSoylu16} to construct and run queries such as $\mathsf{ActivePowerTrip}(\text{tb0})@x$. The task of the OBDA system such as Ontop~\cite{DBLP:conf/semweb/Rodriguez-MuroKZ13,DBLP:journals/semweb/CalvaneseCKKLRR17} will be, using the mappings, to rewrite the engineer's \emph{ontology-mediated  query} into an SQL query over the database and then execute it returning the time intervals $x$ where the turbine with the ID tb0 had active power trips.

Unfortunately, the ontology and query languages designed for OBDA and standardised by the W3C---the \OWLQL{} profile of \OWL{} and SPARQL---are not suitable for the Siemens case because they were not meant to deal with essentially \emph{temporal} data, concepts and properties. There have been several attempts to develop temporal OBDA.

One approach is to use the same \OWLQL{} as an ontology language, assuming that ontology axioms hold at all times, and extend the query language with various temporal operators~\cite{DBLP:conf/rr/Gutierrez-BasultoK12,BBL13,BLT13,OMNZK:13,DBLP:conf/rr/KlarmanM14,DBLP:conf/rweb/OzcepM14,DBLP:conf/sigmod/KharlamovBJKLMN16}. Unfortunately, \OWLQL{} is not able to define the temporal feature of `active power trip'\!, and so the engineer would have to capture it in a complex temporal query (or call an expert in temporal logic). Another known  approach is to allow the temporal operators of the linear-time temporal logic \LTL{} in both queries and ontologies~\cite{ArtaleKWZ13,DBLP:conf/ijcai/ArtaleKKRWZ15,GJK-IJCAI16}. For more details and further references, consult the recent survey by \citeA{DBLP:conf/time/ArtaleKKRWZ17}\footnote{Surveys of early developments in temporal deductive databases are given by~\citeA{DBLP:books/bc/tanselCGSS93/BaudinetCW93} and \citeA{DBLP:conf/dagstuhl/ChomickiT98}.}\!.

However, standard \LTL{} over a \emph{discrete} timeline such as $(\mathbb N,\le)$ or $(\mathbb Z,\le)$ is not able to adequately represent the temporal data and knowledge in the Siemens use case because measurements are taken and sent \emph{asynchronously} by \emph{multiple} sensors at \emph{irregular} time intervals, which can depend on the turbine model, sensor type, etc. To model measurements and events using discrete time, one could take a sufficiently small time unit (quantum), say 1 second, and encode `active power was below 0.15MW for a period of one minute' by an \LTL-formula of the form $\Lnext p \land \Lnext^2 p \land \dots \land \Lnext^{60} p$, where $\Lnext$ is the previous-time operator. One problem with this encoding is that it is clearly awkward, not succinct, and only works under the assumption that the active power is measured \emph{each and every second}. If, for some reason, a measurement is missing as in \TB, the formula becomes inadequate.
This problem can be solved by using the (more succinct) metric temporal logic $\MTL$ with  operators like $\boxminus_{[1,60]}$ interpreted as `at every time instant within the previous minute when a measurement was taken'\!. The satisfiability problem for the description logic $\mathcal{ALC}$ extended with such operators over $(\mathbb N,\le)$ was investigated by Guti\'errez-Basulto, Jung, and Ozaki~\citeyear{DBLP:conf/ecai/Gutierrez-Basulto16}.
%
A more fundamental issue with modelling turbine events using discrete time is that it only applies to data complying with the chosen quantum and requires amendments every time the  quantum has to be set to a different value because of a new equipment or because asynchronous sensor measurements start to happen more frequently.
Thus, a better way of modelling the temporal data and events under consideration is by means of a suitable fragment of $\MTL$ interpreted over \emph{dense} time such as the rationals $(\mathbb Q,\le)$ or reals $(\mathbb R,\le)$. This would allow us to capture, for example, that one event, say a sharp temperature rise, happened \emph{just before} (maybe a fraction of a quantum), and so possibly caused another event, say an emergency shutdown, which is a typical feature of an asynchronous behaviour of real-time systems where the actual time of event occurrences cannot be predicted at the modelling stage.

\subsection{Metric Temporal Logic}
The \emph{metric temporal logic} \MTL{}
was originally designed for modelling and reasoning about real-time systems \cite{DBLP:journals/rts/Koymans90,DBLP:journals/iandc/AlurH93}. \MTL{} is equipped  with two alternative semantics, \emph{pointwise} and \emph{continuous} (aka interval-based). In both semantics, the timestamps are taken from a dense timeline $(\mathbb T,\le)$ such as $(\mathbb Q,\le)$ or $(\mathbb R,\le)$. Under the pointwise semantics, an interpretation is a \emph{timed word}, that is, a finite or infinite sequence of pairs $(\Sigma_i, t_i)$, where $\Sigma_i$ is a subset of propositional variables that are assumed to hold at $t_i \in \mathbb T$ and $t_i < t_j$ for $i < j$. Under the continuous semantics, an interpretation is an assignment of a set  of propositional variables  to \emph{each} $t \in \mathbb T$. \MTL{} allows formulas such as $\boxplus_{[1.5,3]} \varphi$ (or $\diamondplus_{[1.5,3]} \varphi$) that holds at a moment $t$ if and only if $\varphi$ holds at every (respectively, some) moment in the interval $[t+1.5,t+3]$. However, under the pointwise semantics, $t$ must be a timestamp from the timed word and $\varphi$ must only hold at every (respectively, some) $t_i$ with $1.5 \le t_i - t \le 3$. Thus, $\boxplus_{[1,1]} \bot$ is satisfiable under the pointwise semantics, for example, by a timed word with $t_{i+1} - t_i > 1$, but not under the continuous semantics.

In the Siemens case, we assume that the real-time system is being continuously monitored, the result of the next measurement of a sensor is only recorded
when it exceeds the previous one by some fixed margin, and events such as active power trip can happen between measurements. This makes the continuous semantics a natural choice for temporal modelling.
%
%
The satisfiability problem for \MTL{} under this semantics turns out to be undecidable~\cite{DBLP:journals/iandc/AlurH93} and {\sc ExpSpace}-complete if the punctual operators such as $\diamondplus_{[1,1]}$ are disallowed~\cite{DBLP:journals/jacm/AlurFH96}; see also the work by \citeA{Ouaknine:2005:DMT:1078035.1079694,DBLP:conf/formats/OuaknineW08}. Note that, under the pointwise semantics, \MTL{} is decidable over finite timed words, though not primitive recursive~\cite{Ouaknine:2005:DMT:1078035.1079694}.



\subsection{Our contribution}
Having analysed two real-world scenarios of querying asynchronous real-time systems (to be discussed in Section~\ref{sec:uc}), we came to a conclusion that a basic ontology language for temporal OBDA should contain datalog rules with \MTL{} operators in their bodies. In this language, for example, the event of  active power trip can be defined by the rule
\begin{multline}\label{eq:power-trip}
\mathsf{ActivePowerTrip}(v) \leftarrow  \mathsf{Turbine}(v) \land{}
\boxminus_{[0,1m]}\mathsf{ActivePowerBelow0.15}(v) \land {}\\
\diamondminus_{[60s,63s]} \boxminus_{[0,10s]} \mathsf{ActivePowerAbove1.5}(v).
\end{multline}
The variables of the predicates in such rules range over a (non-temporal) object domain. Thus, the intended domain for $v$ in~\eqref{eq:power-trip} comprises turbines, their parts, sensors, etc. The underlying (dense) timeline is implicit: we understand~\eqref{eq:power-trip} as saying that $\mathsf{ActivePowerTrip}(v)$ holds at any given time instant $t$ if the pattern shown in the picture below has occurred before $t$:
\begin{center}
\begin{tikzpicture}[point/.style={draw, thick, circle, inner sep=1.5, outer
    sep=2}]
  \node[point,label=above:{$t$}] (t) at (0,0) {}  node[below right] {$\mathsf{ActivePowerTrip}$};


  \coordinate (t') at (-5,0);
  \draw[purple, ultra thick] (t) -- (t') node [above, midway] {$\mathsf{ActivePowerBelow0.15}$};

  \coordinate (t''') at (-5.4,0);
  \draw[purple, dotted, thick] (t') -- (t''') ;

  \coordinate (t'') at (-7.4,0);
  \draw[blue, ultra thick] (t''') -- (t'')  node [above left, pos=0.4] {$\mathsf{ActivePowerAbove1.5}$};

   \draw[thick, dotted] (t'') -- (-10,0) node [left, pos =1] {$v$};

   \draw [decorate,decoration={brace,amplitude=15pt, raise=2pt}]
(t) -- ($(t')+(0.15,0)$) node [midway,shift={(.4,-.7)}]
{$\textcolor{darkgreen}{1m}$};

   \draw [decorate,decoration={brace,amplitude=15pt, raise=7pt}]
($(t'')+(1.7,0)$) -- (t) node [midway,shift={(.4,.9)}]
{$\textcolor{darkgreen}{63s}$};


\draw [decorate,decoration={brace,amplitude=15pt, raise=2pt}]
(t''') -- ($(t'')+(.15,0)$) node [midway,shift={(.4,-.7)}]
{$\textcolor{darkgreen}{10s}$};
\end{tikzpicture}
\end{center}
Unlike model-checking liveness properties (that some events eventually
happen) in transition systems, our task is to query historical data
for events that have already happened and are actually implicitly
recorded in the data. As a consequence, we do not need ontology axioms
with eventuality operators in the head such as
$\diamondplus_{[0,3s]} \mathsf{ShutDown}(v) \leftarrow
\mathsf{ActivePowerTrip}(v)$ saying that an active power trip must be
followed by a shutdown within 3 seconds.  \OWLQL{} allows existential
quantification in the head of rules such as
$\exists u \, \mathsf{hasRotor}(v,u) \leftarrow \mathsf{Turbine}(v)$
stating that every turbine has a rotor. Although axioms of this sort are
present in the Siemens turbine configuration ontology~\cite{KharlamovMMNORS17}, we opted not to include $\exists$ in the head of rules in our language. On the one hand, we have not found  meaningful queries in the use cases for which such axioms would provide more answers. On the
other hand, it is known that existential axioms may considerably increase the
combined complexity of both
atemporal~\cite{DBLP:journals/ai/GottlobKKPSZ14,JACM} and temporal
ontology-mediated query
answering~\cite{DBLP:conf/ijcai/ArtaleKKRWZ15}. For these reasons, we
do not allow existential rules in our ontology language and leave
their investigation for future work.

The resulting temporal ontology language can be described as a datalog extension of the \emph{Horn fragment} of \MTL{} (without diamond operators in the head of rules). We denote this language by $\dMTL$ and prove in Section~\ref{sec:complexity} that answering ontology-mediated queries of the form $(\Pi, Q(\avec{v})@x)$ is {\sc ExpSpace}-complete for combined complexity, where $\Pi$ is a $\dMTL$ program, $Q(\avec{v})$ a goal with individual variables $\avec{v}$, and $x$ a variable over time intervals during which $Q(\avec{v})$ holds. On the other hand, we show that $\hMTL$ becomes undecidable if the diamond operators are allowed in the head of rules.
We also prove that answering \emph{propositional} $\dMTL$ queries is P-hard for data complexity. To compare, recall that answering ontology-mediated queries with propositional (not necessarily Horn) \LTL{} ontologies is \text{NC}$^1$-complete for data complexity~\cite{DBLP:conf/ijcai/ArtaleKKRWZ15}.

From the practical point of view, most interesting are \emph{nonrecursive} $\dMTL$ queries. We show in Section~\ref{nonrecursive} that answering such queries is in $\text{AC}^0$ for data complexity (assuming that data timestamps and the ranges of the temporal operators in $\dMTL$ programs are represented as finite binary fractions) and {\sc PSpace}-complete for combined complexity (even {\sc NP}-complete if the arity of predicates is bounded). In this case, we develop a query answering algorithm that can be implemented in standard SQL with window functions. We also present in Section~\ref{implementing} a framework for practical OBDA with nonrecursive $\dMTL$ queries and temporal log data stored in databases as shown above. Finally, in Section~\ref{sec:uc}, we evaluate our framework on two use cases. We develop a $\dMTL$ ontology for temporal concepts used in typical  queries at Siemens (e.g., $\mathsf{NormalStop}$ that takes place if events $\mathsf{ActivePowerOff}$, $\mathsf{MainFlameOff}$, $\mathsf{CoastDown6600to1500}$, and $\mathsf{CoastDown1500to200}$ happen in a certain temporal pattern). We also create a weather ontology defining standard meteorological concepts such as $\mathsf{Hurricane}$ ($\mathsf{HurricaneForceWind}$, wind with the speed above 118 km/h, lasting at least 1 hour).
Using Siemens sensor databases and MesoWest historical records of the weather stations across the US, we experimentally demonstrate  that our algorithm is efficient in practice and scales on large datasets of up to 8.3GB. We used two systems, PostgreSQL and Apache Spark, to evaluate our SQL programs. To our surprise, Apache Spark achieved tenfold better performance on the weather data than PostgreSQL. This effect can be attributed to the capacity of Spark to parallelise query execution as well as to the natural `modularity' of weather data by location.

An extended abstract of this paper was presented at AAAI-17~\cite{DBLP:conf/aaai/BrandtKKRXZ17}.


\section{Datalog\MTL}

In the standard metric temporal logic \MTL~\cite{DBLP:journals/jacm/AlurFH96}, the temporal domain is the real numbers $\mathbb{R}$, while the intervals $\range$ in the constrained temporal operators such as $\diamondplus_{\range}$ (sometime in the future within the interval $\range$ from now) have natural numbers or $\infty$ as their endpoints. In the context of the applications of $\MTL$ we deal with in this paper, it  is more natural to assume that the endpoints of $\range$ are non-negative \emph{dyadic rational numbers}---finite binary fractions\footnote{In other words, a dyadic rational is a  number of the form $n/2^m$, where $n \in \mathbb Z$ and $m \in \mathbb N$.} such as 101.011---or $\infty$. We denote the set of dyadic rationals by $\mathbb Q_2$ and remind the reader that $\mathbb Q_2$ is dense in $\mathbb R$ and, by Cantor's theorem, $(\mathbb Q_2,<)$ is isomorphic to $(\mathbb Q,<)$.
%
%
By an \emph{interval}, $\iota$, we mean any nonempty subset of $\mathbb Q_2$ of the form  $[t_1, t_2]$, $[t_1, t_2)$, $(t_1, t_2]$ or $(t_1, t_2)$, where $t_1,t_2\in \mathbb{Q}_2\cup \{-\infty, \infty\}$ and $t_1 \le t_2$. We identify $(t,\infty]$ with $(t,\infty$), $[-\infty,t]$ with $(-\infty,t]$, etc.
A \emph{range}, $\range$, is an interval with non-negative endpoints.  The temporal operators of \MTL{} take the form $\boxplus_{\range}$, $\diamondplus_{\range}$ and $\U_{\range}$, which refer to the future, and $\boxminus_{\range}$, $\diamondminus_{\range}$ and $\Si_{\range}$, which refer to the past.
The end-points of intervals and ranges are assumed to be represented in binary.

An \emph{individual term}, $\tau$, is an individual variable, $v$, or a  constant, $c$. As usual, we assume that there is a countably-infinite list of predicate symbols, $P$, with assigned arities. A $\dMTL$ \emph{program}, $\Pi$, is a finite set of \emph{rules} of the form
\begin{equation*}
A^+ \leftarrow A_1\land \dots \land A_k  \qquad \text{ or } \qquad \bot \leftarrow A_1\land \dots \land A_k,
\end{equation*}
where $k \ge 1$, each $A_i$ $(1 \le i \le k)$ is either an inequality $(\tau \ne \tau')$  or defined by the grammar
\begin{equation*}
A \ ::= \ P(\tau_1,\dots,\tau_m)\ \ \mid\ \ \top \ \mid \ \ \boxplus_{\range} A \ \ \mid\ \ \boxminus_{\range} A \ \ \mid \ \ \diamondplus_{\range} A \ \ \mid\ \ \diamondminus_{\range} A  \ \ \mid\ \ A \U_\range A' \ \ \mid\ \ A \Si_\range A'
\end{equation*}
and $A^+$ is given by the same grammar but \emph{without} any `non-deterministic' operators $\diamondplus_{\range}$, $\diamondminus_{\range}$, $\U_\range$, $\Si_\range$. The atoms $A_1,\dots,A_k$ constitute the \emph{body} of the rule, while $A^+$ or $\bot$ its \emph{head}. As usual, we assume that every variable in the head of a rule also occurs in its body.

A \emph{data instance}, $\D$, is a finite set of \emph{facts} of the form $P(\avec{c})@\iota$, where $P(\avec{c})$ is a ground atom (with a tuple $\avec{c}$ of individual constants) and $\iota$ an interval. The fact $P(\avec{c})@\iota$ states that $P(\avec{c})$ holds throughout the interval $\iota$. We denote by $\num(\D)$ the set of numbers (excluding $\pm\infty$) that occur in $\D$, and by $\num(\Pi, \D)$ the set of number occurring in $\Pi$ or $\D$.

An \emph{interpretation}, $\Mmf$, is based on a \emph{domain} $\Delta \ne \emptyset$ for the individual variables and constants. For any $m$-ary predicate $P$, $m$-tuple $\avec{a}$ from $\Delta$, and moment of time $t \in \mathbb R$, the interpretation $\mathfrak M$ specifies whether $P$ is \emph{true on~$\avec{a}$ at $t$}, in which case we write $\mathfrak M,t \models P(\avec{a})$. Let $\nu$ be an \emph{assignment} of elements of $\Delta$ to the individual terms. To simplify notation, we adopt the standard name assumption according to which $\nu(c) = c$, for every individual constant $c$.
We then set inductively:
\begin{align*}
& \Mmf, t \models^\nu P(\avec{\tau}) \quad\text{ iff }\quad \Mmf, t \models P(\nu(\avec{\tau})),\\
& \Mmf, t \models^\nu (\tau \ne \smash{\tau'}) \quad\text{iff }\quad \nu(\tau) \ne \nu(\smash{\tau'}),\\
& \mathfrak M,t \models^\nu \boxplus_{\range} A \quad \text{ iff } \quad  \mathfrak{M},s\models^\nu A \text{ for all } s \text{ with }s-t \in \range,\\
& \mathfrak M,t \models^\nu \boxminus_{\range} A \quad \text{ iff } \quad  \mathfrak{M},s\models^\nu A \text{ for all } s \text{ with } t-s \in \range,\\
& \mathfrak{M},t\models^\nu \diamondplus_{\range} A \quad \text{ iff } \quad  \mathfrak{M},s\models^\nu A \text{ for some } s \text{ with }s-t \in \range,
\end{align*}
\begin{align*}
& \mathfrak{M},t\models^\nu \diamondminus_{\range} A \quad \text{iff} \quad \mathfrak{M},s\models^\nu A \text{ for some } s \text{ with } t-s \in \range,\\
& \mathfrak{M},t\models^\nu A \U_\range A' \quad \text{iff} \quad  \mathfrak{M},t'\models^\nu A' \text{ for some } t' \text{ with }t'-t \in \range \text{ and } \mathfrak{M},s\models^\nu A \text{ for all } s \in (t, t'),\\
& \mathfrak{M},t\models^\nu A \Si_\range A' \quad \text{iff} \quad  \mathfrak{M},t'\models^\nu A' \text{ for some } t' \text{ with }t-t' \in \range \text{ and } \mathfrak{M},s\models^\nu A \text{ for all } s \in (t', t),\\
& \mathfrak{M},t\models^\nu \top,\\
& \mathfrak{M},t \not \models^\nu \bot.
%
\end{align*}
The picture below illustrates the semantics of the `future' operators for $\range = [d, e]$:

{\centering
\begin{tikzpicture}[point/.style={draw, thick, circle, inner sep=1.5, outer
    sep=2}]
  \node[point] (t) at (0,0) {} node[below left] {$t$} node[above] {$\textcolor{purple}{\boxplus_\range} \textcolor{blue}{A}\ \ \ \ \quad$};

  \draw[thick] (-1, 0) -- (t);

  \coordinate (t') at (3,0);
  \draw[purple, thick] (t) -- (t');

  \coordinate (t'') at (5,0);
  \draw[blue, ultra thick] (t') -- (t'') node [below, midway] {$\textcolor{black}{s}$} node [above, midway] {$A$};

   \draw[thick, -open triangle 60] (t'') -- (6,0);

   \draw [decorate,decoration={brace,amplitude=15pt, raise=2pt}]
(t) -- (t') node [midway,shift={(.3,.7)}]
{$\textcolor{darkgreen}{d}$};

   \draw [decorate,decoration={brace,amplitude=15pt, raise=7pt}]
(t'') -- (t) node [midway,shift={(.3,-.7)}]
{$\textcolor{darkgreen}{e}$};

   \draw[blue] ($(t')+(0,.1)$) -- ($(t')+(0,-.1)$);
   \draw[blue] ($(t'')+(0,.1)$) -- ($(t'')+(0,-.1)$);

\end{tikzpicture}

\begin{tikzpicture}[point/.style={draw, thick, circle, inner sep=1.5, outer
    sep=2}]
  \node[point] (t) at (0,0) {} node[below left] {$t$} node[above] {$\textcolor{purple}{\diamondplus_\range} \textcolor{blue}{A}\ \ \ \ \quad$};

  \draw[thick] (-1, 0) -- (t);

  \coordinate (t') at (3,0);
  \draw[purple, thick] (t) -- (t');

  \coordinate (t'') at (5,0);
  \draw[blue, thick, dotted] (t') -- (t'') node [below, pos = 0.7] {$\textcolor{black}{s}$} node [above, pos = 0.7] {$A$} node[pos = 0.7] {$\bullet$};

   \draw[thick, -open triangle 60] (t'') -- (6,0);

   \draw [decorate,decoration={brace,amplitude=15pt, raise=2pt}]
(t) -- (t') node [midway,shift={(.3,.7)}]
{$\textcolor{darkgreen}{d}$};

   \draw [decorate,decoration={brace,amplitude=15pt, raise=7pt}]
(t'') -- (t) node [midway,shift={(.3,-.7)}]
{$\textcolor{darkgreen}{e}$};

   \draw[blue] ($(t')+(0,.1)$) -- ($(t')+(0,-.1)$);
   \draw[blue] ($(t'')+(0,.1)$) -- ($(t'')+(0,-.1)$);

\end{tikzpicture}

\begin{tikzpicture}[point/.style={draw, thick, circle, inner sep=1.5, outer
    sep=2}]
  \node[point] (t) at (0,0) {} node[below left] {$t$} node[above] {$\textcolor{orange}{A}\,\textcolor{purple}{\U_\range} \textcolor{blue}{A'}\ \ \ \ \quad$};

  \draw[thick] (-1, 0) -- (t);

  \coordinate (t') at (3,0);
  \draw[purple, thick] (t) -- (t');

  \coordinate (t'') at (5,0);
  \draw[blue, thick, dotted] (t') -- (t'') node [below, pos = 0.7] {$\textcolor{black}{s}$} node [above, pos = 0.7] {$A'$} node[pos = 0.7] {$\bullet$};

   \draw[thick, -open triangle 60] (t'') -- (6,0);

   \draw [decorate,decoration={brace,amplitude=15pt, raise=2pt}]
(t) -- (t') node [midway,shift={(.3,.7)}]
{$\textcolor{darkgreen}{d}$};

   \draw [decorate,decoration={brace,amplitude=15pt, raise=7pt}]
(t'') -- (t) node [midway,shift={(.3,-.7)}]
{$\textcolor{darkgreen}{e}$};

   \draw[blue] ($(t')+(0,.1)$) -- ($(t')+(0,-.1)$);
   \draw[blue] ($(t'')+(0,.1)$) -- ($(t'')+(0,-.1)$);

\draw[orange, ultra thick] ($(t)+(0.15,-.1)$) -- ($(t')+(1.3,-.1)$) node[midway, below, shift={(0,.1)}] {$\textcolor{orange}{A}$};

\end{tikzpicture}

}

We say that $\Mmf$ \emph{satisfies} a $\dMTL$ program $\Pi$ under an assignment $\nu$ if, for  \emph{all} $t \in \mathbb R$ and all the rules $A \leftarrow A_1\land \dots \land A_k$ in $\Pi$, we have
\begin{equation*}
\Mmf, t \models^\nu A \ \ \text{ whenever }\ \ \Mmf, t \models^\nu A_i \ \text{ for }  1 \le i \le k.
\end{equation*}
We call $\Mmf$ a \emph{model} of
$\Pi$ and $\D$ and write $\Mmf \models (\Pi, \D)$ if $\Mmf$ satisfies $\Pi$ under every assignment, and
$\mathfrak M,t \models P(\avec{c})$ for any $P(\avec{c})@\iota$ in
$\D$ and any $t \in \iota$.  $\Pi$ and $\D$ are \emph{consistent} if
they have a model.

Note that ranges $\range$ in the temporal operators can be punctual $[r,r]$, in which case $\boxplus_{[r,r]} A$ is equivalent to $\diamondplus_{[r,r]} A$, and $\boxminus_{[r,r]} A$ to $\diamondminus_{[r,r]} A$. We also observe that $\top \Si_\range A$ is equivalent to $\diamondminus_\range A$ (that is, $\mathfrak{M},t\models^\nu \top \Si_\range A$ iff $\mathfrak{M},t\models^\nu \diamondminus_\range A$ for all $\mathfrak M$, $t$ and $\nu$), and $\top \U_\range A$ is equivalent to $\diamondplus_\range A$.

A $\dMTL$ \emph{query} takes the form $(\Pi, \q(\avec{v},x))$, where $\Pi$ is a $\dMTL$ program and $\q(\avec{v},x) = Q(\avec{\tau})@x$, for some predicate $Q$, $\avec{v}$ is a tuple of all individual variables occurring in the terms $\avec{\tau}$, and $x$ an \emph{interval variable}. A \emph{certain answer} to $(\Pi, \q(\avec{v},x))$ over a data instance $\D$ is a pair $(\avec{c},\iota)$ such that $\avec{c}$ is a tuple of constants from $\D$ of the same length as $\avec{v}$, $\iota$ an interval and, for any $t \in \iota$,  any model $\Mmf$ of $\Pi$ and $\D$, and any assignment $\nu$ mapping $\avec{v}$ to $\avec{c}$, we have $\Mmf,t \models^\nu Q(\avec{\tau})$. In this case, we write $\Mmf,t \models \q(\avec{c})$.
If the tuple $\avec{v}$ is empty (that is, $Q(\avec{\tau})$ does not have any individual variables), then we say that $\iota$ is a \emph{certain answer} to $(\Pi, \q(x))$ over $\D$.

\begin{example}\em
Suppose that $\Pi$ has one rule~\eqref{eq:power-trip} and $\D$
consists of the facts
\begin{align*}
& \mathsf{Turbine(tb0)}@(-\infty, \infty),\\
& \mathsf{ActivePowerAbove1.5(tb0)}@[13{:}00{:}00, 13{:}00{:}15),\\
& \mathsf{ActivePowerBelow0.15(tb0)}@[13{:}00{:}17, 13{:}01{:}25).
\end{align*}
Then any subinterval of the interval
$[13{:}01{:}17, 13{:}01{:}18)$ is a certain answer to the $\dMTL$ query
$(\Pi,\mathsf{ActivePowerTrip(tb0)}@x)$.
\end{example}
\begin{example}\label{ex:dance} \em
We illustrate the importance of the operators $\Si$ (since) and $\U$ (until) using an example inspired by the ballet moves ontology~\cite{DBLP:conf/esws/RahebMRPI17}. Suppose we want to say that $\mathsf{SupportBending}$ is a move spanning from the beginning to the end  of $\mathsf{RightAndLeftSupportLowPlace}$ provided that it is preceded by $\mathsf{RightAndLeftSupportMiddlePlace}$, which ends within $3s$ from the beginning of the $\mathsf{RightAndLeftSupportLowPlace}$, as shown below:

{\centering
\begin{tikzpicture}[point/.style={draw, thick, circle, inner sep=1.5, outer
    sep=2}]

  \draw[thick, blue] (0, 0) -- (6, 0) node[midway,above] {$\textcolor{blue}{\mathsf{RightAndLeftSupportMiddlePlace}}$};

  \draw[thick, orange] (8, 0) -- (15, 0) node[midway,above] {$\textcolor{orange}{\mathsf{RightAndLeftSupportLowPlace}}$};

  \draw[thick, darkgreen] (8, -1) -- (15, -1) node[midway,above] {$\textcolor{darkgreen}{\mathsf{SupportBending}}$};

  \draw [decorate,decoration={brace,amplitude=15pt, raise=2pt}]
(6,0) -- (8.5,0) node [midway,shift={(.3,.7)}]
{$\textcolor{darkgreen}{3s}$};

 \draw[thick, gray] (0, -2) -- (8.5, -2) node[midway,above] {$\textcolor{gray}{\diamondminus_{[0, 3s]}\mathsf{RightAndLeftSupportMiddlePlace}}$};

\end{tikzpicture}
}\\[5mm]
We can define the $\mathsf{SupportBending}$ move using the following rule:
\begin{multline*}
\mathsf{SupportBending} \leftarrow \\
\mathsf{RightAndLeftSupportLowPlace} \Si_{[0, \infty)} \big(\diamondminus_{[0, 3s]} \mathsf{RightAndLeftSupportMiddlePlace}\big).
\end{multline*}
(note that a definition of $\mathsf{SupportBending}$ in $\dMTL$ would be problematic if only the $\Box$ and $\Diamond$ operators were available).
\end{example}

By \emph{answering $\dMTL$ queries} we understand the problem of
checking whether a given pair $(\avec{c},\iota)$ is a certain answer
to a given $\dMTL$ query $(\Pi, \q(\avec{v},x))$ over a given data
instance~$\D$. The \emph{consistency} (or \emph{satisfiability}) \emph{problem} is to check whether a given $\dMTL$ program $\Pi$ is consistent with a given data
instance~$\D$. As usual in database theory~\cite{Vardi82} and ontology-mediated query answering, we distinguish between the \emph{combined complexity} and the \emph{data complexity} of these problems: the former regards all the ingredients---$\Pi$, $\q(\avec{c},\iota)$ and $\D$---as input, while the latter one assumes that $\Pi$ and $\q$ are fixed and only~$\D$ and $(\avec{c},\iota)$ are the input.
\begin{proposition}\label{prop:red}
Answering $\dMTL$ queries and consistency checking are polynomially reducible to the complement of each other.
\end{proposition}
\begin{proof}
Suppose first that we want to check whether $(\avec{c}, \iota)$ is a certain answer to $(\Pi, \q(\avec{v},x))$ over $\D$, where $\q(\avec{v},x) = Q(\avec{\tau})@x$ and $\iota = [-t_1, t_2)$, $t_1, t_2 \in \mathbb{Q}_2^{\geq 0}$; other types of $\iota$ are considered analogously. Consider the following program $\Pi'$ and data instance $\D'$:
\begin{align*}
& \Pi' = \Pi \cup \{ \bot \leftarrow P(\avec{v}) \land \boxminus_{[0,t_1]} Q(\avec{v}) \land \boxplus_{(0,t_2)} Q(\avec{v})\},\\
& \D' = \D \cup \{P(\avec{c})@[0,0] \},
\end{align*}
where $P$ is a fresh predicate. It is readily seen that $(\avec{c}, \iota)$ is a certain answer to $(\Pi, \q(\avec{v},x))$ over $\D$ iff $\Pi'$ is \emph{not} consistent with $\D'$. Conversely, $\Pi$ and $\D$ are consistent iff $[0,0]$ is \emph{not} a certain answer to $(\Pi, P@x)$ over $\D$, where $P$ is a fresh $0$-ary predicate, that is, a propositional variable.
\end{proof}

We conclude this section by reminding the reader that, over the integer numbers $(\mathbb Z,<)$, \MTL{} is as expressive as the \emph{linear temporal logic} \LTL{} with the operators $\Rnext$ (at the next moment), $\U$ (until), $\Rbox$ (always in the future), $\Rdiamond$ (some time in the future) and their past counterparts $\Lnext$, $\mathcal{S}$, $\Lbox$ and $\Ldiamond$. For example, the \LTL-formula
$\Rnext A$ is equivalent to $\diamondplus_{[1,1]}A$ and $A \U B$ under the irreflexive semantics to $A\U_{(0,\infty)}B$; conversely, $\diamondplus_{[2,3]}A$ is clearly equivalent to the \LTL-formula $\Rnext\Rnext A \lor \Rnext\Rnext\Rnext A$. However, \MTL{} operators are more succinct, which explains why \MTL-satisfiability over $(\mathbb Z,<)$ is \textsc{ExpSpace}-complete~\cite{DBLP:journals/iandc/AlurH93,DBLP:conf/ictac/FuriaS08} whereas \LTL-satisfiability is \textsc{PSpace}-complete~\cite{Sistla&Clarke85}.

In the next section, we show that consistency checking for $\dMTL$ programs is \textsc{ExpSpace}-complete for combined complexity. It follows from Proposition~\ref{prop:red} that answering $\dMTL$ que\-ries is \textsc{ExpSpace}-complete as well. On the other hand, we also prove that answering propositional $\dMTL$ queries is P-hard for data complexity, and that the extension of $\dMTL$ with $\diamondplus$ in the head of rules leads to undecidability.


\section{Complexity of answering $\dMTL$ queries}
\label{sec:complexity}

Observe first that every $\dMTL$ program $\Pi$ can be transformed (using polynomially-many fresh predicates) to a $\dMTL$ program in \emph{normal form} that only contains rules such as
\begin{align}
& P(\avec{\tau}) \leftarrow \bigwedge_{i\in I} P_i(\avec{\tau}_i), &&  \bot \leftarrow \bigwedge_{i\in I} P_i(\avec{\tau}_i),\label{eq:nf-horn}\\
&  P(\avec{\tau}) \leftarrow P_1 (\avec{\tau}_1)\Si_\range P_2(\avec{\tau}_2), &&  P(\avec{\tau}) \leftarrow P_1(\avec{\tau}_1) \U_\range P_2(\avec{\tau}_2),\label{eq:nf-to-box}\\
& P(\avec{\tau}) \leftarrow \boxminus_{\range} P_1(\avec{\tau}_1), && P(\avec{\tau}) \leftarrow \boxplus_{\range} P_1(\avec{\tau}_1), \label{eq:nf-from-box}
\end{align}
and gives the same certain answers as $\Pi$ over any data instance. (In particular, $\dMTL$ programs in normal form do not contain occurrences of the diamond operators.) For example, we can replace the rule $\boxplus_{\range'} P(\avec{\tau}) \leftarrow P_1(\avec{\tau}_1) \land
\boxminus_{\range} P_2(\avec{\tau}_2)$ in $\Pi$ with three rules
\begin{align*}
& P'(\avec{\tau}) \leftarrow P_1(\avec{\tau}_1) \land
 P'_2(\avec{\tau}_2),\\
& P'_2(\avec{\tau}_2) \leftarrow \boxminus_\range P_2(\avec{\tau}_2),\\
& P(\avec{\tau}) \leftarrow \top \Si_{\range'} P'(\avec{\tau}),
\end{align*}
where $P'$ is a fresh predicate of the same arity as $P$ and $P'_2$ a fresh predicate of the same arity as $P_2$.
Moreover, we can only consider those programs and data instances  where intervals take one of the following two forms:
\begin{itemize}
\item[--] $(t_1, t_2)$ with $t_1,t_2\in \mathbb{Q}_2\cup \{ -\infty,\infty\}$,

\item[--] $[t,t]$ with $t \in \mathbb{Q}_2$; such intervals are called \emph{punctual}.
\end{itemize}
For example, a data instance $\D = \D' \cup \{P(\avec{c})@(t_1, t_2] \}$ is equivalent to the data instance
$$\D = \D' \cup \{P(\avec{c})@(t_1, t_2),\ P(\avec{c})@[t_2, t_2] \}
$$
in the sense that is gives the same certain answers as $\D$, the rule $P(v) \leftarrow \boxminus_{(r_1, r_2]} P'(v)$ is equivalent to $P(v) \leftarrow \boxminus_{(r_1, r_2)} P'(v) \land \boxminus_{[r_2, r_2]} P'(v)$, whereas the rule $P(v) \leftarrow P_1(v) \U_{(r_1, r_2]} P_2(v)$ is equivalent to the pair of rules
$$
P(v) \leftarrow P_1(v) \U_{(r_1, r_2)} P_2(v),\quad P(v) \leftarrow P_1 (v) \U_{[r_2, r_2]} P_2(v).
$$
We use the following notations. We assume that $\langle$ is one of $($ and $[$, while $\rangle$ is one of $)$ and $]$. Given an interval $\iota = \langle \iota_b, \iota_e \rangle$ and a range $\range$, we set
\begin{align*}
  \iota \po \range = %
  \begin{cases}
    \langle \iota_b + r, \iota_e +r \rangle, &\text{ if }\range = [r,r],\\
    (\iota_b + r_1, \iota_e + r_2), & \text{ if } \range = (r_1, r_2),
      \end{cases} &\quad
  \iota \mo \range = %
  \begin{cases}
    \langle \iota_b - r, \iota_e -r \rangle, & \text{ if }\range = [r,r],\\
    (\iota_b - r_2, \iota_e - r_1), & \text{ if }\range = (r_1, r_2).
      \end{cases}
\end{align*}
In other words,
$
\iota \po \range = \{ t + k \mid t\in\iota \text{ and } k \in\range\}
$ and
$
\iota \mo \range = \{ t - k \mid t\in\iota \text{ and } k \in\range\}
$. We also set
\begin{align*}
       \iota \mc \range =& %
  \begin{cases}
    \langle \iota_b - r, \iota_e -r \rangle, & \text{ if }\range = [r,r],\\
    [\iota_b - r_1, \iota_e - r_2], & \text{ if }\range = (r_1, r_2),\quad r_2, \iota_e \in \mathbb{Q}_2,\\
    [\iota_b - r_1, \infty), & \text{ if }\range = (r_1, r_2),\quad r_2 = \infty \text{ or } \iota_e= \infty,
      \end{cases}\\
         \iota \pc \range =& %
  \begin{cases}
    \langle \iota_b + r, \iota_e +r \rangle, & \text{ if }\range = [r,r],\\
    [\iota_b + r_2, \iota_e + r_1], & \text{ if }\range = (r_1, r_2), \quad  r_2, \iota_b \in \mathbb{Q}_2,\\
    (-\infty, \iota_e + r_1], & \text{ if }\range = (r_1, r_2), \quad r_2= \infty \text{ or }\iota_b= -\infty.
      \end{cases}
\end{align*}
%
%
We assume that $\iota \mc \range$ and $\iota \pc \range$ are only defined if $r_2 - r_1 \leq \iota_e - \iota_b$, in which case we write $\range \sqsubseteq \iota$. Thus, $\iota \mc \range$ is defined if there is $t'$ such that $t' + k \in \iota$, for all $k\in\range$. Symmetrically, $\iota \pc \range$ is defined if there is $t'$ such that $t' - k \in \iota$. The picture below illustrates the intuition behind $\iota \po \range$ and $\iota \pc \range$, for non-punctual $\range$, and the difference between them:

\noindent
\begin{center}
	\begin{tikzpicture}[point/.style={draw, thick, circle, inner sep=1.5, outer
		sep=2}]
	\coordinate (f) at (-0.5,4);
	\draw[black](f) node  {$\textcolor{black}{\iota \po \range}$};
	
	\coordinate (t) at (4,4);
	\draw[gray, dashed, thick] (0.5, 4) -- (t);
	\draw[black](t) node  {$\textcolor{black}{(}$};
	
	\coordinate (t') at (12,4);
	\draw[black, thick] (t) -- (t');
	\draw[black](t') node  {$\textcolor{black}{)}$};
	
	\draw[gray, dashed, thick, -open triangle 60] (t') -- (13,4);
	
	\coordinate (f) at (-0.5,3);
	\draw[darkgreen](f) node  {$\textcolor{darkgreen}{\range}$};
	
	\coordinate (t) at (1,3);
	\draw[gray, dashed, thick] (0.5, 3) -- (t);
	\draw[darkgreen](t) node  {$\textcolor{darkgreen}{(}$};
	
	\coordinate (t') at (3,3);
	\draw[darkgreen, thick] (t) -- (t');
	\draw[darkgreen](t') node  {$\textcolor{darkgreen}{)}$};
	
	\coordinate (t'') at (4,3);
	\draw[gray, thick] (t') -- (t'');
	\draw[darkgreen](t') node  {$\textcolor{darkgreen}{)}$};
	
	\draw[gray] ($(t'')+(0,.1)$) -- ($(t'')+(0,-.1)$);
	
	\draw [darkgreen,decorate,decoration={brace,amplitude=5pt, raise=5pt}]
	(t) -- (t'') node [midway,shift={(.2,.5)}]
	{$\textcolor{darkgreen}{r_2}$};
	
	\draw [darkgreen,decorate,decoration={brace,amplitude=5pt, raise=5pt, mirror}]
	(t') -- (t'') node [midway,shift={(.2,-.5)}]
	{$\textcolor{darkgreen}{r_1}$};
	
	\coordinate (x) at (9,3);
	\draw[gray, dashed, thick] (t'') -- (x);
	\draw[blue](x) node  {$\textcolor{blue}{(}$};
	
	\coordinate (x') at (11,3);
	\draw[blue, thick] (x) -- (x');
	\draw[blue](x') node  {$\textcolor{blue}{)}$};
	
	\coordinate (x'') at (12,3);
	\draw[gray, thick] (x') -- (x'');
	\draw[blue](x') node  {$\textcolor{blue}{)}$};
	
	\draw[gray] ($(x'')+(0,.1)$) -- ($(x'')+(0,-.1)$);
	
	\draw [blue,decorate,decoration={brace,amplitude=5pt, raise=5pt}]
	(x) -- (x'') node [midway,shift={(.2,.5)}]
	{$\textcolor{blue}{r_2}$};
	
	\draw [blue,decorate,decoration={brace,amplitude=5pt, raise=5pt, mirror}]
	(x') -- (x'') node [midway,shift={(.2,-.5)}]
	{$\textcolor{blue}{r_1}$};
	
	\draw[gray, dashed, thick, -open triangle 60] (x'') -- (13,3);
	
	\coordinate (f) at (-0.5,2);
	\draw[myred](f) node  {$\textcolor{myred}{\iota}$};
	
	\coordinate (t) at (3,2);
	\draw[gray, dashed, thick] (0.5, 2) -- (t);
	\draw[myred](t) node  {$\textcolor{myred}{\langle}$};
	
	\coordinate (t') at (9,2);
	\draw[myred, thick] (t) -- (t');
	\draw[myred](t') node  {$\textcolor{myred}{\rangle}$};		
	
    \draw[gray, dashed, thick, -open triangle 60] (t') -- (13,2);
	
	\coordinate (f) at (-0.5,1);
	\draw[darkgreen](f) node  {$\textcolor{darkgreen}{\range}$};
	
	\coordinate (t) at (3,1);
	\draw[gray, dashed, thick] (0.5, 1) -- (t);
	\draw[darkgreen](t) node  {$\textcolor{darkgreen}{(}$};
	
	\coordinate (t') at (5,1);
	\draw[darkgreen, thick] (t) -- (t');
	\draw[darkgreen](t') node  {$\textcolor{darkgreen}{)}$};
	
	\coordinate (t'') at (6,1);
	\draw[gray, thick] (t') -- (t'');
	\draw[darkgreen](t') node  {$\textcolor{darkgreen}{)}$};
	
	\draw[gray] ($(t'')+(0,.1)$) -- ($(t'')+(0,-.1)$);
	
	\draw [darkgreen,decorate,decoration={brace,amplitude=5pt, raise=5pt}]
	(t) -- (t'') node [midway,shift={(.2,.5)}]
	{$\textcolor{darkgreen}{r_2}$};
	
	\draw [darkgreen,decorate,decoration={brace,amplitude=5pt, raise=5pt, mirror}]
	(t') -- (t'') node [midway,shift={(.2,-.5)}]
	{$\textcolor{darkgreen}{r_1}$};
	
	\coordinate (x) at (7,1);
	\draw[gray, dashed, thick] (t'') -- (x);
	\draw[blue](x) node  {$\textcolor{blue}{(}$};
	
	\coordinate (x') at (9,1);
	\draw[blue, thick] (x) -- (x');
	\draw[blue](x') node  {$\textcolor{blue}{)}$};
	
	\coordinate (x'') at (10,1);
	\draw[gray, thick] (x') -- (x'');
	\draw[blue](x') node  {$\textcolor{blue}{)}$};
	
	\draw[gray] ($(x'')+(0,.1)$) -- ($(x'')+(0,-.1)$);
	
	\draw [blue,decorate,decoration={brace,amplitude=5pt, raise=5pt}]
	(x) -- (x'') node [midway,shift={(.2,.5)}]
	{$\textcolor{blue}{r_2}$};
	
	\draw [blue,decorate,decoration={brace,amplitude=5pt, raise=5pt, mirror}]
	(x') -- (x'') node [midway,shift={(.2,-.5)}]
	{$\textcolor{blue}{r_1}$};
	
\draw[gray, dashed, thick, -open triangle 60] (x'') -- (13,1);
	
	\coordinate (f) at (-0.5,0);
	\draw[black](f) node  {$\textcolor{black}{\iota \pc \range}$};
	
	\coordinate (t) at (6,0);
	\draw[gray, dashed, thick] (0.5, 0) -- (t);
	\draw[black](t) node  {$\textcolor{black}{[}$};
	
	\coordinate (t') at (10,0);
	\draw[black, thick] (t) -- (t');
	\draw[black](t') node  {$\textcolor{black}{]}$};
	
\draw[gray, dashed, thick, -open triangle 60] (t') -- (13,0);
	
	\coordinate (x) at (3,2.7);
	\coordinate (x') at (3,2.2);
	\draw[gray, dotted, thick] (x) -- (x');	
	
	\coordinate (x) at (3,1.7);
	\coordinate (x') at (3,1.2);
	\draw[gray, dotted, thick] (x) -- (x');		
	
	\coordinate (x) at (9,2.7);
	\coordinate (x') at (9,2.2);
	\draw[gray, dotted, thick] (x) -- (x');	
	
	\coordinate (x) at (9,1.7);
	\coordinate (x') at (9,1.2);
	\draw[gray, dotted, thick] (x) -- (x');		
	
	\end{tikzpicture}
\end{center}

Furthermore, we write
\begin{itemize}
\item[--] $\bigcap_{i \in I}\iota_i \neq \emptyset$ to say that the  intersection of the intervals $\iota_i$, for $i \in I$, is non-empty;

\item[--] $\bigcap_{i \in I}\iota_i$ for the intersection of the intervals $\iota_i$ provided that $\bigcap_{i \in I}\iota_i \neq \emptyset$; otherwise $\bigcap_{i \in I}\iota_i$ is undefined;

\item[--] $\bigcup_{i \in I}\iota_i$ for the union of the intervals $\iota_i$ provided that $\bigcup_{i \in I}\iota_i$ is a single interval; otherwise $\bigcup_{i \in I}\iota_i$ is undefined;

\item[--] $\iota^c$ for the \emph{closure} of an interval $\iota$, that is  $\iota^c=[\iota_b, \iota_e]$ for any $\iota=\langle \iota_b, \iota_e \rangle$.
\end{itemize}

Suppose now that we are given a $\dMTL$ program $\Pi$ (in normal form)
and a data instance $\D$. We define a (possibly infinite) set $\mathfrak C_{\Pi,\D}$ of atoms of the form $P(\avec{c})@\iota$ or $\bot@\iota$ that contains all answers to $\dMTL$ queries with $\Pi$ over $\D$. The construction is essentially the standard chase procedure from database theory~\cite{Abitebouletal95} adapted to time intervals and the temporal operators by mimicking their semantics. The only new chase rule is coalescing  (coal) that merges---possibly infinitely-many---smaller intervals into the lager one they cover. Because of this rule, our chase construction requires transfinite recursion; see also the work by \citeA{DBLP:journals/tocl/BresolinKMRSZ17} and \citeA{DBLP:conf/ijcai/ArtaleKKRWZ15}.

Let $\C$ be some set of atoms of the form
$P(\avec{c})@\iota$ or $\bot @ \iota$ from $\Pi$ and $\D$.  Denote by
$\cl(\C)$ the result of applying exhaustively and non-recursively the following rules
to $\C$:
\begin{description}
\item[(coal)] if $P(\avec{c})@\iota_i \in \C$, for all $i \in I$ with a possibly infinite set $I$, and $\bigcup_{i\in I}\iota_i$ is defined, then we add $P(\avec{c})@\bigcup_{i\in I}\iota_i$ to $\C$;


\item[(horn)] if $P(\avec{c}) \leftarrow \bigwedge_{i \in I} P_i(\avec{c}_i)$ is an instance of a rule in $\Pi$ with all $P_i(\avec{c}_i)@\iota_i$ in $\C$ and $\bigcap_{i\in I}\iota_i \ne \emptyset$, then we add $P(\avec{c})@\bigcap_{i\in I}\iota_i$ to $\C$; if $\bot \leftarrow \bigwedge_{i \in I} P_i(\avec{c}_i)$ is an instance of a rule in $\Pi$, then we add $\bot @\bigcap_{i\in I}\iota_i$ to $\C$;

\item[$(\Si_{\range})$] if $P(\avec{c}) \leftarrow P_1(\avec{c}_1) \Si_\range P_2(\avec{c}_2)$ is an instance of a rule in $\Pi$ with $P_i(\avec{c}_i)@ \iota_i \in \C$ for $i \in \{1,2\}$, $\iota^c_1 \cap \iota_2 \neq \emptyset$, and $((\iota^c_1 \cap \iota_2) \po \range) \cap \iota^c_1 \neq \emptyset$, then we add $P(\avec{c})@((\iota^c_1 \cap \iota_2) \po \range) \cap \iota^c_1$ to $\C$; see the picture below, where $\range = (r_1, r_2)$;
\begin{center}
	\begin{tikzpicture}[point/.style={draw, thick, circle, inner sep=1.5, outer
		sep=2}]
	\coordinate (f) at (-1.7,4);
	\draw[blue](f) node  {$\textcolor{blue}{\iota_2}$};
	
	\coordinate (t) at (1,4);
	\draw[gray, dashed, thick] (0, 4) -- (t);
	\draw[blue](t) node  {$\textcolor{blue}{\langle}$};
	
	\coordinate (t') at (3,4);
	\draw[blue, thick] (t) -- (t') node [above, midway] {$\textcolor{blue}{P_2}$};
	\draw[blue](t') node  {$\textcolor{blue}{\rangle}$};
	
    \draw[gray, dashed, thick, -open triangle 60] (t') -- (7,4);
	
	\coordinate (f) at (-1.7,3);
	\draw[red](f) node  {$\textcolor{red}{\iota_1}$};
	
	\coordinate (t) at (2,3);
	\draw[gray, dashed, thick] (0, 3) -- (t);
	\draw[red](t) node  {$\textcolor{red}{\langle}$};
	
	\coordinate (t') at (5,3);
	\draw[red, thick] (t) -- (t') node [above, midway] {$\textcolor{red}{P_1}$};
	\draw[red](t') node  {$\textcolor{red}{\rangle}$};
	
	\draw[gray, dashed, thick, -open triangle 60] (t') -- (7,3);
	
	\coordinate (f) at (-1.7,2);
	\draw[black](f) node  {$\textcolor{black}{(\iota^c_1 \cap \iota_2)}$};
	
	\coordinate (t) at (2,2);
	\draw[gray, dashed, thick] (0, 2) -- (t);
	\draw[black](t) node  {$\textcolor{black}{[}$};
	
	\coordinate (t') at (3,2);
	\draw[black, thick] (t) -- (t');
	\draw[black](t') node  {$\textcolor{black}{\rangle}$};

	\draw[gray, dashed, thick, -open triangle 60] (t') -- (7,2);
	
%
%
%
%
%
%
%
	
	\coordinate (f) at (-1.7,1);
	\draw[myred](f) node  {$\textcolor{myred}{(\iota^c_1 \cap \iota_2)\po\range}$};
	
	\coordinate (t) at (2,1);
	\draw[gray, dashed, thick] (0, 1) -- (t);
	\draw[gray](t) node  {$\textcolor{gray}{[}$};
	
	\coordinate (t') at (3,1);
	\draw[gray, thick] (t) -- (t');
	\draw[gray](t') node  {$\textcolor{gray}{\rangle}$};
	
	\coordinate (x) at (3.5,1);
	\draw [gray,decorate,decoration={brace,amplitude=5pt, raise=7pt}]
	(t) -- (x) node [midway,shift={(.2,.5)}]
	{$\textcolor{gray}{r_1}$};
	
	\coordinate (x') at (5.5,1);
	\draw [gray,decorate,decoration={brace,amplitude=5pt, raise=5pt, mirror}]
	(t') -- (x') node [midway,shift={(.2,-.5)}]
	{$\textcolor{gray}{r_2}$};
	
	
	\coordinate (t) at (3.5,1);
	\draw[gray, dashed, thick] (0, 1) -- (t);
	\draw[myred](t) node  {$\textcolor{myred}{(}$};
	
	\coordinate (t') at (5.5,1);
	\draw[myred, thick] (t) -- (t');
	\draw[myred](t') node  {$\textcolor{myred}{)}$};
	
	\draw[gray, dashed, thick, -open triangle 60] (t') -- (7,1);
	
	\coordinate (f) at (-1.7,0);
	\draw[myblue](f) node  {$\textcolor{myblue}{((\iota^c_1 \cap \iota_2) \po \range) \cap \iota^c_1}$};
	
	\coordinate (t) at (3.5,0);
	\draw[gray, dashed, thick] (0, 0) -- (t);
	\draw[myblue](t) node  {$\textcolor{myblue}{(}$};
	
	\coordinate (t') at (5,0);
	\draw[myblue, thick] (t) -- (t');
	\draw[myblue](t') node  {$\textcolor{myblue}{]}$};
	
	\draw[gray, dashed, thick, -open triangle 60] (t') -- (7,0);
	
	\end{tikzpicture}
\end{center}

\item[$(\boxplus_{\range})$] if $P(\avec{c}) \leftarrow \boxplus_{\range} P_1(\avec{c}_1)$ is an instance of a rule in $\Pi$ with $P_1(\avec{c}_1)@\iota \in \C$ and $\range \sqsubseteq \iota$, then we add $P(\avec{c})@(\iota \mc \range)$ to $\C$;

\item[$(\U_{\range})$] if $P(\avec{c}) \leftarrow P_1(\avec{c}_1) \U_\range P_2(\avec{c}_2)$ is an instance of a rule in $\Pi$ with $P_i(\avec{c}_i)@ \iota_i \in \C$, $\iota^c_1 \cap \iota_2 \neq \emptyset$ and $((\iota^c_1 \cap \iota_2) \mo \range) \cap \iota^c_1 \neq \emptyset$, then we add $P(\avec{c})@((\iota^c_1 \cap \iota_2) \mo \range) \cap \iota^c_1$ to $\C$;

\item[$(\boxminus_{\range})$] if $P(\avec{c}) \leftarrow \boxminus_{\range} P_1(\avec{c}_1)$ is an instance of a rule in $\Pi$ with $P_1(\avec{c}_1)@\iota \in \C$ and $\range \sqsubseteq \iota$, then we add $P(\avec{c})@(\iota \pc \range)$ to $\C$.
\end{description}
We set $\cl^0(\D) = \D \cup \{\top(-\infty, \infty)\}$ and, for any successor ordinal $\xi +1$ and limit ordinal $\zeta$,
\begin{equation}
\cl^{\xi +1}(\D) = \cl(\cl^\xi(\D)),\quad  \cl^{\zeta} (\D) = \bigcup\nolimits_{\xi<\zeta} \cl^{\xi}(\D) \quad \text{ and }\quad \mathfrak C_{\Pi,\D} = \cl^{\omega_1}(\D), \label{closures}
\end{equation}
where $\omega_1$ is the first uncountable ordinal (as $\cl^{\omega_1}(\D)$ is countable, there is an ordinal $\alpha < \omega_1$ such that $\cl^\alpha(\D) = \cl^\beta (\D)$, for all $\beta \ge \alpha$). We regard $\mathfrak C_{\Pi,\D}$ as both a set of atoms of the form $P(\avec{c})@\iota$ or $\bot@\iota$ and an interpretation where, for any $t \in \mathbb{R}$, any $P$ (different from $\bot$), and any tuple $\avec{c}$ of individual constants, we have $\mathfrak C_{\Pi,\D},t \models P(\avec{c})$ iff $P(\avec{c})@\iota \in \mathfrak C_{\Pi,\D}$ and $t \in \iota$. The \emph{domain} of $\mathfrak C_{\Pi,\D}$ is the set $\ind(\D) \cup \ind(\Pi)$ that comprises the individual constants occurring in $\D$ and $\Pi$.

We illustrate the definition above by a simple example:
\begin{example}\em
Let $\Pi$ have two rules $P \leftarrow \boxminus_{[1,1]} P$ and $Q \leftarrow \boxplus_{(0, \infty)} P$, and let $\D = \{P(0,1] \}$. The first $\omega$ steps of the construction of $\mathfrak C_{\Pi,\D}$ will produce, using the rules $(\boxminus_{\range})$ and (coal), the atoms $P(n,n+1]$ and $P(0,n+1]$, for $n < \omega$. In the step $\omega + 1$, (coal) will give $P(0,\infty)$ and then $(\boxplus_{\range})$ will return $Q@[0, \infty)$.
%
\end{example}

\begin{lemma}\label{canonical}
Let $\Pi$ be a $\dMTL$ program and $\D$ a data instance. Then, for any predicate symbol $P$ from $\Pi$ and $\D$, any tuple $\avec{c}$ of constants from $\D$ and $\Pi$, and any interval $\iota$,
\begin{description}
\item[\rm (\emph{i})] $P(\avec{c})@\iota \in \mathfrak C_{\Pi,\D}$ implies $\Mmf,t \models P(\avec{c})$, for all $t \in \iota$ and all models $\Mmf$ of $\Pi$ and $\D$\textup{;}

\item[\rm (\emph{ii})] if $\bot@\iota \notin \mathfrak C_{\Pi,\D}$ for any $\iota$, then $\mathfrak C_{\Pi,\D}\models (\Pi, \D)$\textup{;} otherwise, $\Pi$ and $\D$ are inconsistent.
\end{description}
\end{lemma}
\begin{proof}
(\emph{i}) Suppose that $\Mmf$ is a model of $\Pi$ and $\D$, and that  $P(\avec{c})@\iota \in \mathfrak C_{\Pi,\D}$. Let $\xi$ be the smallest ordinal such that $P(\avec{c})@\iota \in \cl^\xi(\D)$. We show that $\Mmf,t \models P(\avec{c})$ for all $t \in \iota$ by induction of $\xi$. If $\xi = 0$, then $P(\avec{c})@\iota \in \D$, and since $\Mmf$ satisfies every assertion in $\D$, we are done. If  $\xi = \xi' + 1$ then $P(\avec{c})@\iota$ was obtained from $\cl^{\xi'}(\D)$ by applying one of the construction rules for $\mathfrak C_{\Pi,\D}$. Suppose $P(\avec{c})@\iota$ is $P(\avec{c})@\bigcup_{i\in I}\iota_i$ obtained by {\rm (coal)}. By the induction hypothesis, $\Mmf,t \models P(\avec{c})$ for all $t \in \iota_i$ and $i \in I$. Clearly, $\Mmf, t \models P(\avec{c})$ for all $t \in \bigcup_{i\in I}\iota_i$, and so for all $t \in \iota$. The case of {\rm (horn)} is similar (with intersection in place of union).

Suppose $P(\avec{c})@\iota$ is obtained by $(\Si_{\range})$ from $P_i(\avec{c}_i)@\iota_i$, $i \in \{1,2\}$. By the induction hypothesis, $\Mmf,t \models P_i(\avec{c}_i)$ for every $t \in \iota_i$. Take an arbitrary $t \in ((\iota^c_1 \cap \iota_2) \po \range) \cap \iota_1^c$. Then there exists $t' \in \iota^c_1 \cap \iota_2$ such that $t - t' \in \range$ and $\Mmf,t \models P_2(\avec{c}_2)$. Moreover, we have $\Mmf, s \models P_1(\avec{c}_1)$ for all $s \in (t',t)$. Therefore, $\Mmf, t \models P_1(\avec{c}_1) \Si P_2(\avec{c}_2)$. If $P(\avec{c})@\iota$ is obtained by $(\boxplus_{\range})$ from $P_1(\avec{c}_1)@\iota$, the proof is analogous by considering $t \in \iota \mc \range$. The remaining rules are treated similarly.

(\emph{ii}) Suppose $\bot@\iota \notin \mathfrak C_{\Pi,\D}$ for any $\iota$. By definition, $\D \subseteq \mathfrak C_{\Pi,\D}$, and so $\mathfrak C_{\Pi,\D} \models P(\avec{c})@\iota$ for every $P(\avec{c})@\iota \in \D$. To show that all the rules in $\Pi$ are satisfied by $\mathfrak C_{\Pi,\D}$, we take an assignment $\nu$, a rule $P(\avec{\tau}) \leftarrow \bigwedge_{i\in I} P_i(\avec{\tau}_i)$ from $\Pi$, and suppose that $\mathfrak C_{\Pi,\D}, t \models^\nu P_i(\avec{\tau}_i)$, for all $i\in I$. By the definition of $\mathfrak C_{\Pi,\D}$, it follows that $\mathfrak C_{\Pi,\D}, t \models P_i(\nu(\avec{\tau}_i))$ and $P_i(\nu(\avec{\tau}_i)) \in \mathfrak C_{\Pi,\D}$, for some $\iota_i \ni t$. Moreover, there are ordinals $\xi_i$, $i \in I$, such that $P_i(\nu(\avec{\tau}_i))@\iota_i \in \cl^{\xi_i}(\D)$. By the rule {\rm (horn)}, we then have $P(\nu(\avec{\tau}))@\bigcap_{i\in I} \iota_i \in \cl^{\max \{\xi_i \mid i \in I\}+1}(\D)$, from which   $P(\nu(\avec{\tau}))@\bigcap_{i\in I} \iota_i \in \mathfrak C_{\Pi,\D}$, and so $\mathfrak C_{\Pi,\D},t \models P(\nu(\avec{\tau}))$. Now, consider a rule $\bot \leftarrow \bigwedge_{i\in I} P_i(\avec{\tau}_i)$ and suppose that $\mathfrak C_{\Pi,\D}, t \models^\nu P_i(\avec{\tau}_i)$, for all $i\in I$. By the argument  above, we then should have $\bot @ \bigcap_{i\in I} \iota_i \in \mathfrak C_{\Pi,\D}$, which is a contradiction. For a rule  $P(\avec{\tau}) \leftarrow P_1(\avec{\tau}_1) \Si_\range P_2(\avec{\tau}_2)$, take an arbitrary $t$ and suppose that $\mathfrak C_{\Pi,\D}, t_2 \models^\nu P_2(\avec{\tau}_2)$ for some $t_2$ with $t - t_2 \in \range$ and $\mathfrak C_{\Pi,\D}, t_1 \models^\nu P_1(\avec{\tau}_1)$ for all $t_2 \in (t_2, t)$. By the construction of  $\mathfrak C_{\Pi,\D}$, it follows that $P_2(\nu(\avec{\tau}_2)) @ \iota_2 \in \mathfrak C_{\Pi,\D}$ for some $\iota_2 \ni t_2$. Moreover, there are finitely many intervals $\iota_i'$, $i \in I$, such that $(t_2, t) \subseteq \bigcup_{i \in I} \iota_i'$ and $P_1(\nu(\avec{\tau}_1)) @ \iota_i' \in \mathfrak C_{\Pi,\D}$. By the rule {\rm (coal)}, $P_1(\nu(\avec{\tau}_1)) @ \iota_1 \in \mathfrak C_{\Pi,\D}$ for $\iota_1 = \bigcup_{i \in I} \iota_i'$. It follows then that $t_2, t \in \iota_1^c$, and so $\iota_1^c \cap \iota_2 \neq \emptyset$ and $t \in ((\iota^c_1 \cap \iota_2) \po \range) \cap \iota^c_1$. Thus, by the rule $(\Si_{\range})$, we have $P(\nu(\avec{\tau}))@((\iota^c_1 \cap \iota_2) \po \range) \cap \iota^c_1 \in \mathfrak C_{\Pi,\D}$. Therefore, $\mathfrak C_{\Pi,\D}, t \models^\nu P(\avec{\tau})$. The remaining rules are considered in the same manner.

That $\bot@\iota \in \mathfrak C_{\Pi,\D}$, for some $\iota$, implies inconsistency of $\D$ and $\Pi$ follows from (\emph{i}).   \end{proof}

If $\bot @ \iota \notin \mathfrak C_{\Pi,\D}$, we call $\mathfrak C_{\Pi,\D}$ the \emph{canonical} (or \emph{minimal}) \emph{model of $\Pi$ and $\D$}. We now establish an important property of $\mathfrak C_{\Pi,\D}$ that will allow us to reduce consistency checking for $\dMTL$ programs and data to the satisfiability problem for formulas in the linear temporal logic $\LTL$ over $(\mathbb Z,<)$.

Recall that the \emph{greatest common divisor} of a finite set $N \subseteq \mathbb{Q}$ (at least one of which is not 0) is the largest number $\gcd(N) >0$ such that every $n \in N$ is divisible by $\gcd(N)$ (in the sense that $n/\gcd(N) \in \mathbb Z$). It is known that $\gcd(N)$ always exists and $\gcd(N) \leq \prod_{n \in N} |n|$. It is easy to see that, for any a finite set $N \subseteq \mathbb Q_2$ (at least one of which is not 0), we have $\gcd(N) = 2^m$, where $m$ is the maximal natural number such that $n/2^m \in  N$ is an irreducible fraction. Thus, $\gcd(N)$ can be computed and stored using space polynomial in $|N|$ (the size of the binary encoding of $N$). To make further definitions simpler, it will be convenient to assume that $\gcd(N)=1$ if $N = \{0\}$.

Given a $\dMTL$ program $\Pi$ and a data instance $\D$, we take $d = \gcd(\num(\Pi, \D))$. Denote by $\sect_{\Pi, \D}$ the set of all the intervals of the form $[k d, k d]$ and \mbox{$((k-1) d, k d)$}, for $k \in \mathbb{Z}$. Clearly, $\sect_{\Pi, \D}$ is a partition of $\mathbb{Q}_2$. We represent $\sect_{\Pi, \D}$ as
$$
\sect_{\Pi, \D} = \{ \dots ,\sigma_{-3},\sigma_{-2},\sigma_{-1}, \sigma_0, \sigma_1, \sigma_2, \sigma_3, \dots\},
$$
where $\sigma_0 = [0,0]$, $\sigma_1 = (0, d)$, $\sigma_2 = [d,d]$, $\sigma_3 = (d,2d)$, $\sigma_{-1} = (-d, 0)$, etc. Thus, $\sigma_i$ is punctual if $i$ is even and non-punctual if $i$ is odd. We refer to the $\sigma_i$ as \emph{sections} of $\sect_{\Pi, \D}$.

\begin{lemma}\label{l:zones}
For every atom $P(\avec{c})$ and every $\sigma \in \sect_{\Pi,\D}$, we either have $\mathfrak C_{\Pi, \D}, t \models P(\avec{c})$ for all $t \in \sigma$, or $\mathfrak C_{\Pi, \D}, t \not \models P(\avec{c})$ for all $t \in \sigma$.
\end{lemma}
\begin{proof}
It suffices to show that every interval $\iota$ such that  $P(\avec{c})@\iota \in \mathfrak C_{\Pi, \D}$ takes one of the following forms: $(- \infty, \infty)$, $\langle dk, \infty)$, $(- \infty, dk \rangle$, $\langle dk, dk' \rangle$, where $k,k' \in \mathbb Z$. This can readily be done by induction on the construction of $\mathfrak C_{\Pi, \D}$. Indeed, when applied to a set of atoms of this form, the operator $\cl$ also results in a set of such atoms.
%
\end{proof}

Our aim now is to encode the structure of $\mathfrak C_{\Pi, \D}$ given by Lemma~\ref{l:zones} by means of an \LTL-formula $\varphi_{\Pi, \D}$ that is satisfiable over $(\mathbb Z,<)$ iff $\Pi$ and $\D$ are consistent.
%
The \LTL-formula $\varphi_{\Pi, \D}$ contains \emph{propositional variables} of the form $P^{\avec{c}}$, where $P$ is a predicate symbol from $\Pi$ and $\D$ of arity $m$ and $\avec{c}$ an $m$-tuple of individual constants from $\D$ and $\Pi$, as well as two additional propositional variables $\mathsf{odd}$ and $\mathsf{even}$. We define $\varphi_{\Pi, \D}$ as a conjunction of the following clauses, where $\nu$ is any assignment of the individual constants from $\D$ and $\Pi$ to the terms in $\Pi$, and $\Box \psi$ is a shorthand for $\Lbox \varphi \land \varphi \land \Rbox \varphi$:
%
\begin{itemize}
\item[--] $\mathsf{even} \land \Box(\mathsf{even} \leftrightarrow \Rnext \mathsf{odd}) \land \Box(\mathsf{odd} \leftrightarrow \Rnext \mathsf{even})$;

\item[--] $\displaystyle\Box (P^{\nu(\avec{\tau})} \leftarrow \bigwedge_{i\in I} P_i^{\nu(\avec{\tau}_i)})$, for every rule $P(\avec{\tau}) \leftarrow \bigwedge_{i\in I} P_i(\avec{\tau}_i)$ in $\Pi$;

\item[--] $\displaystyle\Box (\bot \leftarrow \bigwedge_{i\in I} P_i^{\nu(\avec{\tau}_i)})$, for every rule $\bot \leftarrow \bigwedge_{i\in I} P_i(\avec{\tau}_i)$ in $\Pi$;

\item[--] for every rule $P(\avec{\tau}) \leftarrow P_1(\avec{\tau}_1) \Si_\range P_2(\avec{\tau}_2)$ in $\Pi$ with $\range = [r,r]$, we require two clauses:
\begin{align*}
\Box \bigl( P^{\nu(\avec{\tau})} \leftarrow \mathsf{even} \land \nxt^{-2r/d} P_2^{\nu(\avec{\tau}_2)} \land \bigwedge_{-2r/d < j < 0} \nxt^j P_1^{\nu(\avec{\tau}_1)}\bigr),\\
\Box \bigl( P^{\nu(\avec{\tau})} \leftarrow \mathsf{odd} \land \nxt^{-2r/d} P_2^{\nu(\avec{\tau}_2)} \land \bigwedge_{-2r/d \leq j \leq 0} \nxt^j P_1^{\nu(\avec{\tau}_1)}\bigr),
\end{align*}
where $\nxt^n \varphi= \underbrace{\Rnext \dots \Rnext}_n \varphi$  if $n > 0$, $\nxt^0 \varphi = \varphi$, and $\nxt^n \varphi = \underbrace{\Lnext \dots \Lnext}_{|n|} \varphi$ if $n < 0$;

\item[--] for every rule $P(\avec{\tau}) \leftarrow P_1(\avec{\tau}_1) \Si_\range P_2(\avec{\tau}_2)$ in $\Pi$ with $\range = (r_1,r_2)$, we require four clauses:
\begin{align*}
\Box \bigl( P^{\nu(\avec{\tau})} \leftarrow \mathsf{even} \land %
\bigvee_{-2r_2/d < k < - 2r_1/d} (\nxt^k P_2^{\nu(\avec{\tau}_2)} \land \nxt^k \mathsf{even} \land \bigwedge_{k < j < 0} \nxt^j P_1^{\nu(\avec{\tau}_1)}\bigr),\\
%
\Box \bigl( P^{\nu(\avec{\tau})} \leftarrow \mathsf{even} \land %
\bigvee_{-2r_2/d < k < - 2r_1/d} (\nxt^k P_2^{\nu(\avec{\tau}_2)} \land \nxt^k \mathsf{odd} \land \bigwedge_{k \leq j < 0} \nxt^j P_1^{\nu(\avec{\tau}_1)}\bigr),\\
\Box \bigl( P^{\nu(\avec{\tau})} \leftarrow \mathsf{odd} \land %
\bigvee_{-2r_2/d \leq k \leq - 2r_1/d} (\nxt^k P_2^{\nu(\avec{\tau}_2)} \land \nxt^k \mathsf{even} \land \bigwedge_{k < j \leq 0} \nxt^j P_1^{\nu(\avec{\tau}_1)}\bigr),\\
\Box \bigl( P^{\nu(\avec{\tau})} \leftarrow \mathsf{odd} \land %
\bigvee_{-2r_2/d \leq k \leq - 2r_1/d} (\nxt^k P_2^{\nu(\avec{\tau}_2)} \land \nxt^k \mathsf{odd} \land \bigwedge_{k \leq j \leq 0} \nxt^j P_1^{\nu(\avec{\tau}_1)}\bigr);
\end{align*}

\item[--] for every rule $P(\avec{\tau}) \leftarrow P_1(\avec{\tau}_1) \Si_\range P_2(\avec{\tau}_2)$ in $\Pi$ with $\range = (r_1, \infty)$,
\begin{align*}
\Box \bigl( P^{\nu(\avec{\tau})} \leftarrow \mathsf{even}\land \bigwedge_{-2r_1/d \leq j < 0} \nxt^j P_1^{\nu(\avec{\tau}_1)}\land \nxt^{-2r_1/d} (&
 P_1^{\nu(\avec{\tau}_1)} \Si (\mathsf{even} \land P_2^{\nu(\avec{\tau}_2)})\lor{}\\
 &P_1^{\nu(\avec{\tau}_1)} \Si (\mathsf{odd} \land P_1^{\nu(\avec{\tau}_1)} \land P_2^{\nu(\avec{\tau}_2)}))\bigr),\\
\Box \bigl( P^{\nu(\avec{\tau})} \leftarrow \mathsf{odd}\land \bigwedge_{-2r_1/d \leq j \leq 0} \nxt^j P_1^{\nu(\avec{\tau}_1)}\land \nxt^{-2r_1/d} (&
 P_2^{\nu(\avec{\tau}_2)} \lor P_1^{\nu(\avec{\tau}_1)} \Si (\mathsf{even} \land P_2^{\nu(\avec{\tau}_2)})\lor{}\\
 &P_1^{\nu(\avec{\tau}_1)} \Si (\mathsf{odd} \land P_1^{\nu(\avec{\tau}_1)} \land P_2^{\nu(\avec{\tau}_2)}))\bigr)
 \end{align*}
(recall that $P \Si Q$ holds at $i$ iff there exists $k < i$, such that $Q$ holds at $k$ and $P$ holds at all $j$ with $k < j < i$);

\item[--] similar clauses for the rules of the form $P(\avec{\tau}) \leftarrow P_1(\avec{\tau}_1) \U_\range P_2(\avec{\tau}_2)$ (here we need the `until' operator $\U$), $P(\avec{\tau}) \leftarrow \boxminus_\range P_1(\avec{\tau}_1)$ and $P(\avec{\tau}) \leftarrow \boxplus_\range P_1(\avec{\tau}_1)$ in $\Pi$;

\item[--] for every fact $P(\avec{c})@\iota$ in $\D$, we need the clauses:
\begin{align*}
& \nxt^{2r/d} P^{\avec{c}}, && \text{ if }\iota = [r,r],\\
& \bigwedge_{2r_1/d < i < 2r_2/d} \nxt^{i} P^{\avec{c}},&& \text{ if }\iota = (r_1, r_2) \text{ and } r_1, r_2 \in \mathbb{Q}_2,\\
& \nxt^{2r_1/d} \Rbox P^{\avec{c}},&& \text{ if }\iota = (r_1, r_2), r_1 \in \mathbb{Q}_2 \text{ and } r_2 = \infty,\\
& \nxt^{2r_2/d} \Lbox P^{\avec{c}},&& \text{ if }\iota = (r_1, r_2), r_1 = -\infty \text{ and } r_2 \in \mathbb{Q}_2,\\
& \Rbox \Lbox P^{\avec{c}}, && \text{ if }\iota = (r_1, r_2), r_1 = -\infty \text{ and } r_2 = \infty.
 \end{align*}
  \end{itemize}
\begin{lemma}
$(\Pi, D)$ is consistent iff $\varphi_{\Pi, \D}$ is satisfiable.
\end{lemma}
\begin{proof}
$(\Rightarrow)$ If $\mathfrak C_{\Pi, \D}$ is a model of $(\Pi, \D)$, we define an $\LTL$-interpretation $\Mmf$ by taking
\begin{itemize}
\item[--] $\Mmf, i \models P^{\avec{c}}$ iff $\mathfrak C_{\Pi, \D}, t \models P(\avec{c})$, for all $t \in \sigma_i$ and $i\in \mathbb{Z}$, all tuples of individual constants $\avec{c}$, and predicates $P$;

\item[--] $\Mmf, i \models \mathsf{even}$, for even $i \in \mathbb Z$;

\item[--] $\Mmf, i \models \mathsf{odd}$, for odd $i \in \mathbb Z$.
\end{itemize}
It is routine to check that $\Mmf,0 \models \varphi_{\Pi, \D}$,
%
taking into account that $\mathfrak C_{\Pi, \D}, t \models P_1(\avec{c}_1) \Si_\varrho  P_2(\avec{c}_2)$ for some (= all) $t \in \sigma_i$ iff the following conditions hold:
\begin{description}
\item[\rm Case {$\varrho = [r,r]$}:]
$\mathfrak C_{\Pi, \D}, t' \models P_2(\avec{c}_2)$, for some $t' \in \sigma_{i - 2r/d}$, and $\mathfrak C_{\Pi, \D}, s \models P_1(\avec{c}_1)$ for all $s \in
\sigma_j$ such that
\begin{align*}
& i -2r/d < j < i, \qquad  \text{if $i$ is even,}\\
& i -2r/d \leq j \leq i, \qquad \text{if $i$  is odd;}
\end{align*}
\item[\rm Case {$\varrho = (r_1,r_2)$}:] there exists $\sigma_k$ with $\mathfrak C_{\Pi, \D}, t' \models P_2(\avec{c}_2)$, for some $t' \in \sigma_{k}$, and $\mathfrak C_{\Pi, \D}, s \models P_1(\avec{c}_1)$ for all $s \in
\sigma_j$ such that
\begin{align*}
 &i -2r_2/d < k < i - 2r_1/d,\quad k < j < i, \quad  \text{if $i$  is even and } k \text{ is even},\\
 &i -2r_2/d < k < i - 2r_1/d,\quad k \leq j < i, \quad  \text{if $i$ is even and } k \text{ is odd},\\
 &i -2r_2/d \leq k \leq i - 2r_1/d,\quad k < j \leq i, \quad  \text{if $i$ is odd and } k \text{ is even},\\
 &i -2r_2/d \leq k \leq i - 2r_1/d,\quad k \leq j \leq i, \quad  \text{if $i$ is odd and } k \text{ is odd};
 \end{align*}
\end{description}
and similarly for the other temporal operators in $\varphi_{\Pi, \D}$.

$(\Leftarrow)$ Suppose now $\varphi_{\Pi, \D}$ is satisfiable. Take the canonical  model $\mathfrak M$ of $\varphi_{\Pi, \D}$ with $\mathfrak M,0 \models \varphi_{\Pi, \D}$; see the work by \citeA{DBLP:conf/lpar/ArtaleKRZ13} for details. Using the observations above, it is not hard to check that $\Mmf, i \models P^{\avec{c}}$ iff $\mathfrak C_{\Pi, \D}, t \models P(\avec{c})$, for all $t \in \sigma_i$ and $i\in \mathbb{Z}$, all tuples of individual constants $\avec{c}$ and predicates $P$. Details are left to the reader.
\end{proof}

We are now in a position to prove our first complexity result:

\begin{theorem}\label{thm:dMTL}
Consistency checking for $\dMTL$ programs is \textsc{ExpSpace}-complete. The lower bound holds even for propositional $\dMTL$.
\end{theorem}
\begin{proof}
 We first show the upper bound. By the two lemmas above, a $\dMTL$ program $\Pi$ is consistent with a data instance $\D$ iff the \LTL{} formula $\varphi_{\Pi,\D}$ is satisfiable. Thus, a consistency checking  \textsc{ExpSpace} algorithm can first construct $\varphi_{\Pi,\D}$, which requires exponential time in the size of $\Pi$ and $\D$.
Indeed, the greatest common divisor of the set $\num(\Pi, \D)$ can be computed in polynomial time. The $\LTL$ formula $\varphi_{\Pi,\D}$ contains exponentially many clauses (as there are exponentially many assignments $\nu$) of at most exponential size (as they contain $2t/d$ conjuncts or disjuncts, where $t$ is a number from $\Pi$ or $\D$).
After that we can run a standard {\sc PSpace} satisfiability checking algorithm for \LTL; see, e.g., the work by \citeA{Sistla&Clarke85}.

 We establish the matching lower bound  by reduction of the non-halting problem for deterministic Turing machines with an exponential tape.
Let $M$ a deterministic Turing machine  that requires $2^{f(m)}$ cells of the tape given an input of length $m$, for some polynomial $f$. Let $n = f(m)$. Without loss of generality, we can assume that $M$ never runs outside the first $2^n$ cells.
Suppose $M= (Q, \Gamma, \#, \Sigma, \delta, q_0, q_h)$, where $Q$ is a finite set of states, $\Gamma$ a tape alphabet, $\# \in \Gamma$ the blank symbol, $\Sigma \subseteq
\Gamma$ a set of input symbols, $\delta\colon
(Q\setminus\{q_h\}) \times \Gamma \to Q \times \Gamma \times
\{L,R\}$ a transition function, and $q_0,q_h \in Q$ are the
initial and halting states, respectively. Let $\vec{a}=a_1\dots a_{m}$ be an input for $M$. We construct a propositional $\dMTL$ program $\Pi$ and a data instance $\D$ such that they are \emph{not} consistent iff $M$ accepts $\vec{a}$.
In our encoding, we employ the following propositional variables, where $a \in \Gamma$, $q \in Q$:
\begin{itemize}
      %
\item[--] $H_{q,a}$ indicating that a cell is read by the head, the current state of the machine is $q$, and the cell contains $a$;
\item[--] $N_a$ indicating that a cell is not read by the head and contains $a$,

\item[--] $\mathsf{first}$ and $\mathsf{last}$ marking the first and last cells of a configuration, respectively.
\end{itemize}
The program $\Pi$ consists of the following rules, for $a,a', a''\in \Gamma$, $q,q'\in Q$:
\begin{align*}
& \boxplus_{2^n+1}
  H_{q', a''} \leftarrow H_{q,a} \land \boxplus_{1} N_{a''},\quad \boxplus_{2^n} N_{a'} \leftarrow H_{q,a} , &&  \text{if } \delta(q,a) = (q',a',R), \\
& \boxplus_{2^n-1} H_{q', a''} \leftarrow H_{q,a} \land \boxminus_{1} N_{a''},\quad \boxplus_{2^n} N_{a'} \leftarrow H_{q,a} , && \text{if } \delta(q,a) = (q',a',L),
  \\
&  \boxplus_{2^n} N_a  \leftarrow \boxminus_1 N_{a'} \land  N_{a} \land \boxplus_{1} N_{a''},&& \\
& \boxplus_{2^n} N_a  \leftarrow \boxminus_1 H_{q, a'} \land  N_{a} \land \boxplus_{1} N_{a''},&&
  \text{if } \delta(q,a') \neq (r,b,R) \text{ for all }r,b,\\
&  \boxplus_{2^n} N_a  \leftarrow \boxminus_1 N_{a'} \land N_{a} \land \boxplus_{1} H_{q, a''},&&
  \text{if } \delta(q,a'') \neq (r,b,L) \text{ for all }r,b,\\
&  \boxplus_{2^n} N_a  \leftarrow N_{a} \land \mathsf{first} \land  \boxplus_{1}N_{a'},&&
 \\
&  \boxplus_{2^n} N_a  \leftarrow N_{a} \land \mathsf{first} \land \boxplus_{1} H_{q, a'},&&
  \text{if } \delta(q,a') \neq (r,b,L) \text{ for all }r,b,\\
%
&  \boxplus_{2^n} N_a  \leftarrow \boxminus_1 N_{a'} \land  N_a \land \mathsf{last},&&\\
&  \boxplus_{2^n} N_a  \leftarrow \boxminus_1 N_{q, a'} \land  N_a \land \mathsf{last},&&\text{if } \delta(q,a') \neq (r,b,R) \text{ for all }r,b,\\
&  \boxplus_{2^n} \mathsf{first}  \leftarrow \mathsf{first}, \\
& \boxplus_{2^n} \mathsf{last} \leftarrow  \mathsf{last},&&\\
&   \bot  \leftarrow H_{q_h, a},&\\
& \boxplus_{1} N_\# \leftarrow N_\# \land \diamondplus_{(0, \infty)} N_\#^<,&
\end{align*}
where $\boxplus_r$ is an abbreviation for $\boxplus_{[r,r]}$ and similarly for $\boxminus_r$.
Let $\D$ contain the following facts:
\begin{multline*}
 N_{a_i}@[i,i], \text{ for } 1 < i \leq m, \ \ N_{\#}@[m+1,m+1],\ \ N_{\#}^<@[2^n,2^n],\\
 H_{q_0, a_1}@[1,1], \ \ \mathsf{first}@[1,1], \quad \mathsf{last}@[2^n, 2^n].
\end{multline*}
The program represents the computation of $M$ on $\vec{a}$ as a sequence of configurations. The initial one is spread over the time instants  $1, \dots, 2^n$, from which the first $m$ instants represent $\vec{a}$ and the remaining ones are $\#$. The second configuration uses the next $2^n$ instants (i.e., $2^n+1, \dots, 2^n + 2^n$), etc.
It is routine to check that $M$ halts on $\vec{a}$ iff $\Pi$ and $\D$ are inconsistent.
\end{proof}

Note that $\dMTL$ allows \emph{punctual intervals} of the form $[r,r]$ as ranges of temporal operators, and that full propositional $\MTL$ with such intervals is undecidable~\cite{DBLP:journals/iandc/AlurH93}.


Now we turn to the data complexity of $\dMTL$ and show the following result:
\begin{theorem}\label{thm:phard}
    Consistency checking and answering propositional $\dMTL$ queries is \textsc{P}-hard for data complexity \textup{(}under \textsc{LogSpace} reductions\textup{)}.
\end{theorem}
\begin{proof}
We establish this lower bound by reduction of the monotone circuit value problem, which is known to be P-complete~\cite{Arora&Barak09}. Let $\boldsymbol{C}$ be a monotone circuit with input gates having fan-in 1 and all other gates fan-in 2. We assume that the gates are enumerated by consecutive positive integers, so that if there is an edge from $n$ to $m$ then $n < m$. Let $N = 2^k$, for some $k \in \mathbb{N}$, be the minimal number that is greater than or equal to the maximal gate number. We encode the computation of $\boldsymbol{C}$ on an input $\avec{\alpha}$ by a data instance $\D_{\boldsymbol{C}}$  with the following punctual facts, where $[n]$ stands for $[n,n]$:
\begin{itemize}\itemsep=0pt
\item[--] $V[2n+ n/N]$, if $n$ is an input  gate and $\avec{\alpha}(n) = V \in \{T,F\}$;
\item[--] $D[2n+n/N]$, if $n$ is an OR gate;
\item[--] $C[2n+n/N]$, if $n$ is an AND gate;
\item[--] $I_0[2n + m/N], I_1[2n + k/N]$, if $n$ is a gate with input gates $m$ and $k$.
\end{itemize}
Let $\Pi_{\boldsymbol{C}}$ be a $\dMTL$ program with the rules
\begin{align*}
T &\leftarrow \diamondminus_{[2,2]} T,\quad &&F \leftarrow \diamondminus_{[2,2]} F,\\
T & \leftarrow \diamondminus_{[0,1]}(I_0 \land T) \land D, &&  F  \leftarrow \diamondminus_{[0,1]}(I_0 \land F) \land C,\\
T & \leftarrow \diamondminus_{[0,1]}(I_1 \land T) \land D, && F \leftarrow \diamondminus_{[0,1]}(I_1 \land F) \land C,\\
F & \leftarrow \diamondminus_{[0,1]}(I_0 \land F) \land \diamondminus_{[0,1]}(I_1 \land F) \land D,\quad &&T \leftarrow \diamondminus_{[0,1]}(I_0 \land T) \land \diamondminus_{[0,1]}(I_1 \land T) \land C.
\end{align*}
Suppose $n$ is the output gate. Then it is straightforward to check that the value of $\mathbf{C}$ on $\avec{\alpha}$ is $T$ iff $(\Pi, \D) \models T[2n + n/N]$. This immediately implies the required hardness for the query answering problem. %
An example of a circuit $\boldsymbol{C}$ with an assignment $\avec{\alpha}$, and an initial part of the canonical model of $(\Pi_{\boldsymbol{C}}, \D_{\boldsymbol{C}})$ are shown below, with the black symbols above the timestamps indicating what is given in $\D_{\boldsymbol{C}}$ and the grey ones what is implied by $\Pi_{\boldsymbol{C}}$:
\\[5pt]
\begin{tikzpicture}
\begin{scope}[scale = 1.3, xshift = 160, circuit logic US, every circuit symbol/.style={thick},thick]
\node[and gate,inputs={nn}, point right,fill=gray!20,label={[xshift=.3cm, yshift = -.1cm]above left:{\scriptsize \smash{4}}}] (n1) at (0,0) {$\land$};

\node[or gate,inputs={nn}, point right,fill=gray!20,label={[yshift = -.05cm]above left:{\scriptsize \smash{3}}}] (n2) at (0,1) {$\lor$};

\node[and gate,inputs={nn}, point right,fill=gray!20,label={[xshift=.25cm, yshift = -.05cm] above left:{\scriptsize \smash{5}}}] (n3) at (1.3,0.5) {$\land$};

\draw (n1.input 2) -- +(-0.65,0);
\draw (n1.input 1) -| (-0.6,0.5);
\draw (n2.input 2) -| (-0.6,0.5);
\draw (n2.input 1) -- +(-0.7,0);
\draw (-0.6,0.5) -- (-1,0.5);
\node[rectangle, draw, label=160:{\scriptsize \!1}] at (-1.2,0.5) {\footnotesize $F$};
\node[rectangle, draw, label=160:{\scriptsize \!2}] at (-1.2,-0.1) {\footnotesize $T$};
\node[rectangle, draw, label=160:{\scriptsize \!0}] at (-1.2,1.1) {\footnotesize $T$};
\draw (n1.output) -- +(0.25,0) |- (n3.input 2);
\draw (n2.output) -- +(0.2,0) |- (n3.input 1);
\draw (n3.output) -- +(0.2,0);
\end{scope}
\begin{scope}[xscale = .9, yshift = -50, xshift = -50, point/.style={draw, thick, circle, inner sep=1.5, outer
		sep=2}]
    \node[point,  label=below:{\scriptsize $0$}, label=above:{\scriptsize $T$}] at (0,0) {};
    \node[point,  label=below:{\scriptsize $\frac{1}{8}$}] at (0.3,0) {};
    \node[point,  label=below:{\scriptsize $\frac{2}{8}$}] at (0.6,0) {};
    \node[point,  label=below:{\scriptsize $\frac{3}{8}$}] at (0.9,0) {};
    \node[point,  label=below:{\scriptsize $\frac{4}{8}$}] at (1.2,0) {};
    \node[point,  label=below:{\scriptsize $\frac{5}{8}$}] at (1.5,0) {};
    \node[point,  label=below:{\scriptsize $1$}] at (2.4,0) {};
    \node[point,  label=below:{\scriptsize $\frac{16}{8}$}, label={[yshift = .4cm]above: {\scriptsize $\textcolor{gray}{T}$}}] at (4.8,0) {};
    \node[point, label=above:{\scriptsize $F$}] at (5.1,0) {};
    \node[point
    ] at (5.4,0) {};
    \node[point,  label=below:{
    $\dots$}] at (5.7,0) {};
    \node[point
    ] at (6,0) {};
    \node[point,  label=below:{\scriptsize $\frac{21}{8}$}] at (6.3,0) {};
    \node[point,  label=below:{\scriptsize $3$}] at (7.2,0) {};
    \node[point,  label=below:{\scriptsize $\frac{32}{8}$}, label={[yshift = .4cm]above: {\scriptsize $\textcolor{gray}{T}$}}] at (9.6,0) {};
    \node[point, label={[yshift = .4cm]above: {\scriptsize $\textcolor{gray}{F}$}}] at (9.9,0) {};
    \node[point, label=above:{\scriptsize $T$}] at (10.2,0) {};
    \node[point,  label=below:{$\dots$}] at (10.5,0) {};
    \node[point] at (10.8,0) {};
    \node[point,  label=below:{\scriptsize $\frac{37}{8}$}] at (11.1,0) {};
    \node[point,  label=below:{\scriptsize $5$}] at (12.0,0) {};
    \node[point,  label=below:{\scriptsize $\frac{48}{8}$}, label=above:{\scriptsize $I_0$}, label={[yshift = .4cm]above: {\scriptsize $\textcolor{gray}{T}$}}] at (14.4,0) {};
    \node[point, label=above:{\scriptsize $I_1$}, label={[yshift = .4cm]above: {\scriptsize $\textcolor{gray}{F}$}}] at (14.7,0) {};
    \node[point, label={[yshift = .4cm]above: {\scriptsize $\textcolor{gray}{T}$}}] at (15.0,0) {};
    \node[point,  label=below:{ $\dots$}, label=above:{\scriptsize $D$}, label={[yshift = .4cm]above: {\scriptsize $\textcolor{gray}{T}$}}] at (15.3,0) {};
    \node[point] at (15.6,0) {};
    \node[point,  label=below:{\scriptsize $\frac{53}{8}$}] at (15.9,0) {};
    \node at (16.4,0) {$..$};
\end{scope}
\end{tikzpicture}\\[3pt]
To show P-hardness of the consistency problem, it suffices to add the fact $P[2n + n/N]$ to $\D_{\boldsymbol{C}}$, for a fresh $P$, and the axiom $\bot \leftarrow P \land T$ to $\Pi_{\boldsymbol{C}}$.
\end{proof}

The exact data complexity of answering propositional $\dMTL$ queries remains open. It is worth noting that answering ontology-mediated queries with propositional \LTL{} ontologies is \text{NC}$^1$-complete for data complexity~\cite{DBLP:conf/ijcai/ArtaleKKRWZ15}, while answering propositional datalog queries with the Halpern-Shoham operators is P-complete for data complexity~\cite{DBLP:conf/ijcai/KontchakovPPRZ16}.

The diamond operators $\diamondplus_\range$ and $\diamondminus_\range$ are disallowed in the head of $\dMTL$ rules. Denote by $\dMTL^\Diamond$ the extension of $\dMTL$ that allows both box and diamond operators in the head of rules. We show now that this language has much more expressive power and can encode 2-counter Minsky machines, which gives the following theorem; cf.\ the work by~\citeA{DBLP:journals/corr/abs-1305-6137}:

\begin{theorem}\label{thm:undec}
Consistency checking for propositional $\dMTL^{\!\Diamond}$ programs is undecidable.
\end{theorem}
\begin{proof}
We use some ideas of~\citeA{DBLP:journals/corr/abs-1305-6137}, where  a non-Horn fragment of \MTL{} was shown to be undecidable. The proof is by reduction of the undecidable non-halting problem for Minksy machines: given a 2-counter Minsky machine, decide whether it \emph{does not halt} starting from $0$ in both counters.

Suppose we are given a Minsky machine with counters $C_1$ and $C_2$ that has $n-1$ instructions of the form
\begin{align*}
& \text{$i$: Increment($C_k$), goto $j$},\\
& \text{$i$: Decrement($C_k$), goto $j$}, \\
& \text{$i$: If $C_k = 0$ then $j_1$ else $j_2$},
\end{align*}
where $i$, $j$, $j_1$ and $j_2$ are instruction indexes, $k=1,2$, and the $n$-th instruction is
$$
\text{$n$: Halt}.
$$
We encode successive configurations of the machine using the sequence $[0,4),[4,8),[8,12),\dots$ of time  intervals. The current  instruction index is represented by a propositional variable $P_i$, for $1 \le i \le n$, that holds at the first point, say $4m$, of the  interval $[4m,4m+4)$. The current value, say $k_1$, of the counter $C_1$ is encoded by exactly $k_1$ moments of time in the interval $(4m +1,4m+2)$ where the propositional variable $C$ holds true. Similarly, the value $k_2$ of $C_2$ is encoded by exactly $k_2$ moments in the interval $(4m +3,4m+4)$ where the propositional variable $C$ holds true.

The initial configuration is encoded by the following data instance $\D$, where the variable $Z$ indicates that both counters are 0:
\begin{equation}
   P_1@[0,0], \quad Z@ (1,2), \quad Z@ (3,4).
\end{equation}
%
%
For every $i$ $(1 \leq i \leq n)$ we require the rules
\begin{align}
   %
   %
   \boxplus_{[0,1]} Z \leftarrow& P_i,\quad \boxplus_{[2,3]} Z \leftarrow P_i, \quad \bot \leftarrow Z \land C, \quad \bot \leftarrow Z \land N
\end{align}
saying, in particular, that $C$ cannot hold true outside the intended intervals (here $N$ is an auxiliary variable).
To simplify notation, we use the following abbreviations: $\Ci= \boxplus_{[1,1]}$, $\Cii = \boxplus_{[3,3]}$, and $\nxt = \boxplus_{[4,4]}$.
The machine instructions are encoded as follows (the instructions for $C_2$ are obtained by replacing $\Ci$ with  $\Cii$):
\begin{align*}
&\nxt P_{j_1} \leftarrow P_i \land \Ci \boxplus_{(0,1)} Z,&&\\
&\nxt P_{j_2} \leftarrow P_i \land \Ci \diamondplus_{(0,1)} C,&&\text{$i$: if $C_1 = 0$ then $j_1$} \\
& \Ci \boxplus_{(0,1)} \text{CP} \leftarrow P_i, \quad   \Cii \boxplus_{(0,1)} \text{CP} \leftarrow P_i,&& \hspace*{5em} \text{else $j_2$}\\
& \nxt P_j \leftarrow P_i, \quad \Ci \boxplus_{(0,1)} \text{IC} \leftarrow P_i,\quad \Cii \boxplus_{(0,1)} \text{CP} \leftarrow P_i, && \text{$i$: Inc($C_1$), goto $j$}\\
& \nxt P_j \leftarrow P_i, \quad \Ci \boxplus_{(0,1)} \text{DC} \leftarrow P_i,\quad \Cii \boxplus_{(0,1)} \text{CP} \leftarrow P_i, && \text{$i$: Dec($C_1$), goto $j$}.
\end{align*}
Here the variable CP means copying of the counter value, DC means  decrementing it by $1$, and IC incrementing it by $1$. To achieve this, we require the following rules:
\begin{align}
  \notag \nxt C &\leftarrow \text{CP} \land C,\quad
  \nxt Z  \leftarrow \text{CP} \land Z,\\
  \notag \nxt C &\leftarrow \text{DC} \land C \land \diamondplus_{(0,1)} C,\\
  \notag \nxt Z &\leftarrow \text{DC} \land Z \land \diamondplus_{(0,1)} C,\quad \nxt  \boxplus_{[0,1]} Z \leftarrow \text{DC} \land C \land \boxplus_{(0,1)} Z, \\
  %
\label{eq:undec-incr-1} \diamondplus_{(0,1)} N  & \leftarrow \boxplus_{(0,1)} \text{IC} \land \boxplus_{(0,1)} Z, \quad \diamondplus_{(0,1)} N   \leftarrow C \land \text{IC} \land \boxplus_{(0,1)} Z,\\
  \label{eq:undec-incr-2} \nxt C &\leftarrow \text{IC} \land C, \quad \nxt C \leftarrow \text{IC} \land N,\\
 \label{eq:undec-incr-3} \nxt Z &\leftarrow \text{IC} \land Z \land \diamondplus_{(0,1)} N, \quad \nxt  \boxplus_{(0,1)} Z \leftarrow \text{IC} \land N \land \boxplus_{(0,1)} Z,
\end{align}
We explain the intuition behind the most complex rules~\eqref{eq:undec-incr-1}--\eqref{eq:undec-incr-3} that are used to model the increment of the counters. The rules~\eqref{eq:undec-incr-1} mark a new time-point  with the variable $N$ in a block located after the last $C$-time-point in this block (or, according the first axiom, $N$ is placed anywhere in the block if the current value of a counter is $0$). The rules~\eqref{eq:undec-incr-2} insert   $C$ in the next block, where in the current block we have either $C$ or $N$. The rules~\eqref{eq:undec-incr-3} transfer $Z$ from the current block to the next one excluding the time-point where $N$ holds.
%
Finally, we add the rule
\begin{align*}
  &\bot \leftarrow P_n, && \text{n: Halt}.
\end{align*}
It is not hard to check that the program and data instance above are consistent iff the given 2-counter Minsky machine does not halt.
\end{proof}

The diamond operators in the head of rules can encode disjunction and thereby ruin `Horness'\!. Thus, the temporalised description logic $\mathcal{EL}$ with such rules is undecidable~\cite{DBLP:conf/time/LutzWZ08}; cf.\  also the work by~\citeA{GJK-IJCAI16}. The addition of diamonds in the heads to the Horn fragment of the propositional Halpern-Shoham logic $\mathcal{HS}$ can make a P-complete logic  undecidable~\cite{DBLP:journals/tocl/BresolinKMRSZ17}. A distinctive feature of these formalisms is their two-dimensionality~\cite{many_dimensional_modal_logics}, while propositional $\dMTL$ is one-dimensional.  Diamonds in the head of rules also ruin FO-rewritability  of answering ontology-mediated queries with temporalised \textsl{DL-Lite} ontologies by increasing their data complexity to \textsc{coNP}~\cite{ArtaleKWZ13}. The same construction actually shows that nonrecursive $\dMTL$ with binary predicates and diamonds in the heads is \textsc{coNP}-hard.


\section{Nonrecursive $\dMTL$}\label{nonrecursive}

As none of the $\dMTL$ programs required in our use cases is recursive, we now consider the class $\nrdMTL$ of nonrecursive $\dMTL$ programs. We first show that consistency checking (and so query answering) for $\nrdMTL$ programs is \textsc{PSpace}-complete for combined complexity. Then we regard a given $\nrdMTL$ program as fixed and reduce these problems to evaluating a (data-independent) FO$(<)$-formula over any given data, thereby establishing that $\nrdMTL$ is in $\text{AC}^0$ for data complexity.

More precisely, for a program $\Pi$, let $\lessdot$ be the dependence relation on the predicate symbols in $\Pi$: we have $P \lessdot Q$ iff $\Pi$ contains a clause with $P$ in the head and $Q$ in the body. $\Pi$ is called \emph{nonrecursive} if $P \lessdot^+ P$ does not hold for any predicate symbol $P$ in $\Pi$, where $\lessdot^+$ is the transitive closure of $\lessdot$. We denote by $\mathsf{depth}_\Pi(P)$ the maximal number $l$ such that  $P_0 \lessdot P_1 \lessdot \dots \lessdot P_l = P$. (Note that $\mathsf{depth}_\Pi(P) = 0$ iff either $P$ does not occur in $\Pi$ or $P$ occurs only in the body of some rules.) The maximal $\mathsf{depth}_\Pi(P)$ over all predicates $P$ is denoted by  $\mathsf{depth}(\Pi)$.
It should be clear that, for any nonrecursive $\Pi$ and any data instance $\D$, there exists some $n \in \mathbb{N}$ such that $\cl^{n+1}(\D) = \cl^n(\D) = \mathfrak C_{\Pi, \D}$. Therefore, $\mathfrak C_{\Pi, \D}$ is finite.

Denote by $\min \D$ and $\max \D$ the minimal and, respectively, maximal \emph{finite} numbers that occur in the intervals from $\D$.
Let $K$ be the largest number occurring in $\Pi$. We then set
\begin{equation*}
M_l = \min \D -  K \times \mathsf{depth}(\Pi)  \quad \text{ and } \quad  M_r = \max \D +  K \times \mathsf{depth}(\Pi).
\end{equation*}
Let $d = \gcd(\num(\Pi, \D))$. The next lemma will be required for our \textsc{PSpace} algorithm checking consistency of $\nrdMTL$ programs.
\begin{lemma}\label{l:zones-nr}
  Let $\Pi$ be a $\nrdMTL$ program. Then every interval $\iota$ such that  $P(\avec{c})@\iota \in \mathfrak C_{\Pi, \D}$ or $\bot(\avec{c})@\iota \in \mathfrak C_{\Pi, \D}$ takes one of the following forms: $(- \infty, \infty)$, $\langle dk, \infty)$, $(- \infty, dk \rangle$, $\langle dk, dk' \rangle$, where $k, k' \in \mathbb{Z}$ and  $M_l \leq dk \leq dk' \leq M_r$.
\end{lemma}
\begin{proof}
That every interval in $\mathfrak C_{\Pi, \D}$ is of the form $(- \infty, \infty)$, $\langle dk, \infty)$, $(- \infty, dk \rangle$, $\langle dk, dk' \rangle$, where $k, k' \in \mathbb{Z}$, was observed in the proof of Lemma~\ref{l:zones}. Thus, we only need to establish  the bounds on $dk$ and $dk'$. For each $P$, let $\mathsf{hi}(P)$ and $\mathsf{lo}(P)$ be the maximal and, respectively, minimal number $dk \in \mathbb{Q}$ such that $P(\avec{c})@\iota \in \mathfrak C_{\Pi, \D}$ and $dk$ is an end-point of $\iota$. Note that $\mathsf{hi}(P)$ and $\mathsf{lo}(P)$ can be undefined.  We are going to show that $\mathsf{hi}(P)$ is either undefined or $\mathsf{hi}(P) \leq \max \D + \mathsf{depth}_\Pi(P) K$. (That $\mathsf{lo}(P)$ is either undefined or $\mathsf{lo}(P) \geq \min \D - \mathsf{depth}_\Pi(P) K$ is left to the reader.) Clearly, this fact implies the required bounds on $dk$ and $dk'$.

The proof is by induction on the construction of $\mathfrak C_{\Pi, \D}$. Let $\mathsf{hi}^n(P)$ be the maximal $dk \in \mathbb{Q}_2$ such that $P(\avec{c})@\iota \in \cl^n(\D)$ and $dk$ is an end-point of $\iota$. We show by induction on $n$ that either $\mathsf{hi}^n(P)$ is  undefined or $\mathsf{hi}^n(P) \leq \max \D + K \mathsf{depth}_\Pi(P)$.

For the basis of induction, if $\mathsf{hi}^0(P)$ is defined and  $P(\avec{c})@\iota \in \cl^0(\D)$ is an atom mentioning $\mathsf{hi}^0(P)$, then $P(\avec{c})@\iota \in \D$ and $\mathsf{hi}^0(P) \leq \max \D$. Assume next that  $n = n' + 1$. Suppose $\mathsf{hi}^n(P)$ is defined and let $P(\avec{c})@\iota \in \cl^n(\D)$ be an atom mentioning $\mathsf{hi}^n(P)$. If $P(\avec{c})@\iota \in \cl^{n'}(\D)$, we are done by the induction hypothesis. Otherwise, we consider how $P(\avec{c})@\iota$ was obtained. Suppose it was obtained by {\rm (coal)} with $\iota = \bigcup_{i \in I} \iota_i$. By the induction hypothesis, $\mathsf{hi}^{n'}(P) \leq \max \D + K\mathsf{depth}_\Pi(P)$, and so every number mentioned in $\{ \iota_i \mid i \in I\}$ does not exceed $\max \D + K\mathsf{depth}_\Pi(P)$. Thus, we have $\mathsf{hi}^{n}(P) \leq \max \D + K\mathsf{depth}_\Pi(P)$. Now suppose that $P(\avec{c})@\iota$ was obtained by {\rm (horn)} from $P_i(\avec{c}_i)@\iota_i$, $i \in I$. Observe that $\mathsf{depth}_\Pi(P_i) < \mathsf{depth}_\Pi(P)$ and, by the induction hypothesis, $\mathsf{hi}^{n'}(P_i) \leq \max \D + K\mathsf{depth}_\Pi(P_i)$. Since $\iota = \bigcap_{i \in I} \iota_i$, the maximal number mentioned in $\iota$ cannot exceed $\max \D + K\mathsf{depth}_\Pi(P)$. Thus, $\mathsf{hi}^n(P) \leq \max \D + K\mathsf{depth}_\Pi(P)$.
Consider now the case when $P(\avec{c})@\iota$ was obtained by applying $(\Si_{\range})$ to $P_i(\avec{c}_i)@\iota_i$, $i \in \{1,2\}$. By the induction hypothesis, the largest number mentioned in $\iota_i$ does not exceed $\max \D + K\mathsf{depth}_\Pi(P_i)$. On the other hand, $\mathsf{depth}_\Pi(P_i) < \mathsf{depth}_\Pi(P)$ and the maximal number in $\iota$ cannot be larger that the maximal number in $\{\iota_i \mid i \in \{1,2\} \}$ plus $K$. Thus, the maximal number in $\iota$ does not exceed
$$
\max \D + K\mathsf{depth}_\Pi(P_i) + K \leq \max \D + K\mathsf{depth}_\Pi(P),
$$
and so $\mathsf{hi}^n(P) \leq \max \D + K \mathsf{depth}_\Pi(P)$. The remaining temporal rules are similar and left to the reader.
\end{proof}

Suppose we are given a $\nrdMTL$ program $\Pi$ and a data instance $\D$. If $\Pi$ and $\D$ are inconsistent then, by Lemmas~\ref{canonical} and~\ref{l:zones-nr}, we have $\bot@\iota \in \mathfrak C_{\Pi, \D}$, for some $\iota$ of the form $(- \infty, \infty)$, $\langle dk, \infty)$, $(- \infty, dk \rangle$, $\langle dk, dk' \rangle$, where $k, k' \in \mathbb{Z}$ and  $M_l \leq dk \leq dk' \leq M_r$. Thus, there is a \emph{derivation} of $\bot@\iota$ from $\Pi$ and $\D$, that is, a tree whose root is $\bot@\iota$, whose leaves are some atoms from $\D$, and whose every non-leaf vertex results from applying one of the rules (coal), (horn), $(\Si_{\range})$, $(\boxplus_{\range})$, $(\U_{\range})$, $(\boxminus_{\range})$ to the immediate predecessors of this vertex.

\begin{lemma}\label{derivation}
If $\bot@\iota \in \mathfrak C_{\Pi, \D}$ then there is a derivation  of $\bot@\iota$ from $\Pi$ and $\D$ such that
\begin{itemize}
\item[$(i)$] the length of any branch in the derivation does not exceed $2|\Pi|$;

\item[$(ii)$] for some polynomial $p$, every non-leaf vertex, corresponding to the application of \textup{(}coal\textup{)} in the derivation, has at most $2^{p(|\Pi|,|\D|)}$ immediate predecessors.
\end{itemize}
\end{lemma}
\begin{proof}
To show $(i)$, it suffices to recall that $\Pi$ is non-recursive (and so none of the rules in $\Pi$ can be applied twice in the same branch of the derivation) and observe that we can always replace multiple successive applications of the rule (coal) with a single application.

$(ii)$ follows from Lemma~\ref{l:zones-nr}.
\end{proof}

\begin{theorem}\label{thm:pspace}
Consistency checking for $\nrdMTL$ programs is \textsc{PSpace}-complete for combined complexity. The lower bound holds even for propositional $\nrdMTL$.
\end{theorem}
\begin{proof}
The upper bound is established by a standard algorithm~\cite{Ladner77thecomputational,tobiespspace} using Lemma~\ref{derivation} and Savitch's theorem according to which $\textsc{NPSpace} = \textsc{PSpace}$. In essence, the $\textsc{NPSpace}$ algorithm guesses branches of the derivation one by one and keeps only last two branches in memory. By Lemma~\ref{derivation}~$(i)$, each branch contains $\le 2|\Pi|$ atoms of the form $P(\avec{c})@\iota$, where $\iota$ is as in Lemma~\ref{l:zones-nr}, and so is stored in polynomial space. In addition, we store the axioms in $\Pi$ that created these atoms, or (coal) if the atom was obtained by coalescing. In the latter case, we also need to guess a number $k$ indicating how many distinct intervals are coalesced to obtain $\iota$. By Lemma~\ref{derivation}~$(ii)$, $k \le 2^{p(|\Pi|,|\D|)}$, and so it can be stored in polynomial space.


The lower bound is proved by reduction of the satisfiability problem for quantified Boolean formulas (QBFs), which is known to be \textsc{PSpace}-complete. Let $\varphi = Q_n p_n \dots Q_0 p_0 \varphi_0$ be a QBF, where each $Q_i$ is either $\forall$ or $\exists$, and $\varphi_0 = c_0 \land \dots \land c_m$ is a propositional formula in CNF  with $c_i = l_0 \lor \dots \lor l_k$, with each $l_i$ being either a variable $p_j$ or its negation $\neg p_j$, for $0 \leq j \leq n$. In our $\nrdMTL$ program, we use the following propositional variables:
 \begin{itemize}
 \item[--] $P_0, \dots, P_n$ (to represent $p_0, \dots, p_n$ from $\varphi$);

 \item[--] $\bar P_0, \dots, \bar P_n$ (to represent $\neg p_0, \dots, \neg p_n$);

 \item[--] $P_{0}^0, \dots, P_0^n$ for $p_0$; $P_{1}^1, \dots, P_1^n$ for $p_1$, etc.; $P_{n}^n$ for $p_n$, and similarly for $\neg p_i$;

 \item[--] $F_0, \dots, F_{n+1}$;

 \item[--] $C_0, \dots, C_m$ (to represent $c_0,\dots,c_m$).
 \end{itemize}
We first take a data instance $\D$ with the following facts:
 \begin{align*}
   P_{i}^i@[0, 2^i), \ \bar P_{i}^i@[2^i, 2^{i+1}),\quad \text{ for }0 \leq i \leq n.
 \end{align*}
Starting from this data, we can generate all the truth-assignments for the variables $p_0, \dots, p_n$ using the following rules, where $0 \leq i \leq n$:
 \begin{align*}
   P_i \leftarrow P_{i}^n,&\quad \bar P_i \leftarrow \bar P_{i}^n, \\
    P_i^{j+1} \leftarrow P_i^j, \quad \boxplus_{2^{j+1}} P_i^{j+1} \leftarrow P_i^j,& \quad \bar P_i^{j+1} \leftarrow \bar P_i^j, \quad \boxplus_{2^{j+1}} \bar P_i^{j+1} \leftarrow \bar P_i^j, \quad i \leq j < n.
 \end{align*}
The canonical model for $\D$ and the rules above for the variables $p_0, p_1, p_2$ (thus, $n=2$) is shown in Fig.~\ref{fig:can-mod}.

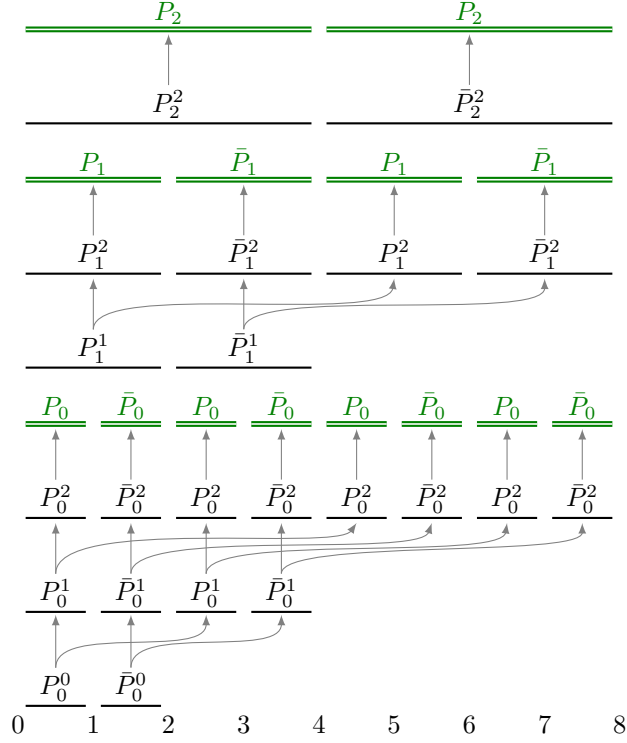
\begin{figure}[t]
\centering
\begin{tikzpicture}[yscale=2.5,
var/.style={minimum height=0.45cm,outer sep=-0.05cm},
>=latex,
]\small
\begin{scope}[yshift = -1cm]
 \node at (0,0.2) {$0$};%
 \node at (1,0.2) {$1$};%
 \node at (2,0.2) {$2$};%
 \node at (3,0.2) {$3$};%
 \node at (4,0.2) {$4$};%
 \node at (5,0.2) {$5$};%
 \node at (6,0.2) {$6$};%
 \node at (7,0.2) {$7$};%
 \node at (8,0.2) {$8$};%
\end{scope}

\begin{scope}[thick, yshift = -1cm]
 \draw (.1,.3) -- node[above,var] (p0) {$P_0^0$} ++(0.8,0);%
 \draw (1.1,.3) -- node[above,var] (p0b) {$\bar{P}_0^0$} ++(0.8,0);%

 \draw (.1,.8) -- node[above,var] (p11) {$P_0^1$} ++(0.8,0);%
 \draw (1.1,.8) -- node[above,var] (p12) {$\bar{P}_0^1$} ++(0.8,0);%
 \draw (2.1,.8) -- node[above,var] (p13) {$P_0^1$} ++(0.8,0);%
 \draw (3.1,.8) -- node[above,var] (p14) {$\bar{P}_0^1$} ++(0.8,0);%

 \draw (.1,1.3) -- node[above,var] (p21) {$P_0^2$} ++(0.8,0);%
 \draw (1.1,1.3) -- node[above,var] (p22) {$\bar{P}_0^2$} ++(0.8,0);%
 \draw (2.1,1.3) -- node[above,var] (p23) {$P_0^2$} ++(0.8,0);%
 \draw (3.1,1.3) -- node[above,var] (p24) {$\bar{P}_0^2$} ++(0.8,0);%
 \draw (4.1,1.3) -- node[above,var] (p25) {$P_0^2$} ++(0.8,0);%
 \draw (5.1,1.3) -- node[above,var] (p26) {$\bar{P}_0^2$} ++(0.8,0);%
 \draw (6.1,1.3) -- node[above,var] (p27) {$P_0^2$} ++(0.8,0);%
 \draw (7.1,1.3) -- node[above,var] (p28) {$\bar{P}_0^2$} ++(0.8,0);%

  \draw[double, darkgreen] (.1,1.8) -- node[above,var] (p1) {$P_0$} ++(0.8,0);%
 \draw[double, darkgreen] (1.1,1.8) -- node[above,var] (p2) {$\bar{P}_0$} ++(0.8,0);%
 \draw[double, darkgreen] (2.1,1.8) -- node[above,var] (p3) {$P_0$} ++(0.8,0);%
 \draw[double, darkgreen] (3.1,1.8) -- node[above,var] (p4) {$\bar{P}_0$} ++(0.8,0);%
 \draw[double, darkgreen] (4.1,1.8) -- node[above,var] (p5) {$P_0$} ++(0.8,0);%
 \draw[double, darkgreen] (5.1,1.8) -- node[above,var] (p6) {$\bar{P}_0$} ++(0.8,0);%
 \draw[double, darkgreen] (6.1,1.8) -- node[above,var] (p7) {$P_0$} ++(0.8,0);%
 \draw[double, darkgreen] (7.1,1.8) -- node[above,var] (p8) {$\bar{P}_0$} ++(0.8,0);%

\end{scope}

 \node[fit=(p11),inner sep=0.03cm] (pp11) {};
 \node[fit=(p12),inner sep=0.03cm] (pp12) {};
 \node[fit=(p13),inner sep=0.03cm] (pp13) {};
 \node[fit=(p14),inner sep=0.03cm] (pp14) {};
 \node[fit=(p21),inner sep=0.06cm] (pp21) {};
 \node[fit=(p22),inner sep=0.06cm] (pp22) {};
 \node[fit=(p23),inner sep=0.06cm] (pp23) {};
 \node[fit=(p24),inner sep=0.06cm] (pp24) {};
 \node[fit=(p25),inner sep=0.06cm] (pp25) {};
 \node[fit=(p26),inner sep=0.06cm] (pp26) {};
 \node[fit=(p27),inner sep=0.06cm] (pp27) {};
 \node[fit=(p28),inner sep=0.06cm] (pp28) {};
 \node[fit=(p1),inner sep=0.06cm] (pp1) {};
 \node[fit=(p2),inner sep=0.06cm] (pp2) {};
 \node[fit=(p3),inner sep=0.06cm] (pp3) {};
 \node[fit=(p4),inner sep=0.06cm] (pp4) {};
 \node[fit=(p5),inner sep=0.06cm] (pp5) {};
 \node[fit=(p6),inner sep=0.06cm] (pp6) {};
 \node[fit=(p7),inner sep=0.06cm] (pp7) {};
 \node[fit=(p8),inner sep=0.06cm] (pp8) {};

\begin{scope}[gray]
\tiny
 \draw[->] (p0) node[above left,yshift=0.1cm] {} -- (pp11);
 \draw[->] (p0) to[out=90,in=270,looseness=0.3] (pp13.south);
 \draw[->] (p0b) node[above left,yshift=0.1cm] {} -- (pp12);
 \draw[->] (p0b) to[out=90,in=270,looseness=0.3] (pp14.south);
 \draw[->] (p11) node[above left,yshift=0.1cm] {} -- (pp21);
 \draw[->] (p11) to[out=90,in=210,looseness=0.2] (pp25.south);
 \draw[->] (p12) node[above left,yshift=0.1cm] {} -- (pp22);
 \draw[->] (p12) to[out=90,in=230,looseness=0.16] (pp26.south);
 \draw[->] (p13) node[above left,yshift=0.1cm] {}-- (pp23);
 \draw[->] (p13) to[out=90,in=250,looseness=0.14] (pp27.south);
 \draw[->] (p14) node[above left,yshift=0.07cm] {} -- (pp24);
 \draw[->] (p14) to[out=90,in=270,looseness=0.12] (pp28.south);
 \draw[->] (p21) node[above left,yshift=0.1cm] {} -- (pp1);
 \draw[->] (p22) node[above left,yshift=0.1cm] {} -- (pp2);
 \draw[->] (p23) node[above left,yshift=0.1cm] {} -- (pp3);
 \draw[->] (p24) node[above left,yshift=0.1cm] {} -- (pp4);
 \draw[->] (p25) node[above left,yshift=0.1cm] {} -- (pp5);
 \draw[->] (p26) node[above left,yshift=0.1cm] {} -- (pp6);
 \draw[->] (p27) node[above left,yshift=0.1cm] {} -- (pp7);
 \draw[->] (p28) node[above left,yshift=0.1cm] {} -- (pp8);
\end{scope}

\begin{scope}[yshift = -.5cm, thick]
 \draw (.1,1.6) -- node[above,var] (q0) {$P_1^1$} ++(1.8,0);%
 \draw (2.1,1.6) -- node[above, var] (q0b) {$\bar{P}_1^1$} ++(1.8,0);%

 \draw (.1,2.1) -- node[above,var] (q21) {$P_1^2$} ++(1.8,0);%
 \draw (2.1,2.1) -- node[above,var] (q22) {$\bar{P}_1^2$} ++(1.8,0);%
 \draw (4.1,2.1) -- node[above,var] (q23) {$P_1^2$} ++(1.8,0);%
 \draw (6.1,2.1) -- node[above,var] (q24) {$\bar{P}_1^2$} ++(1.8,0);%

 \draw[double, darkgreen] (.1,2.6) -- node[above,var] (q1) {$P_1$} ++(1.8,0);%
 \draw[double, darkgreen] (2.1,2.6) -- node[above,var] (q2) {$\bar{P}_1$} ++(1.8,0);%
 \draw[double, darkgreen] (4.1,2.6) -- node[above,var] (q3) {$P_1$} ++(1.8,0);%
 \draw[double, darkgreen] (6.1,2.6) -- node[above,var] (q4) {$\bar{P}_1$} ++(1.8,0);%
\end{scope}

 \node[fit=(q21),inner sep=0.06cm] (qq21) {};
 \node[fit=(q22),inner sep=0.06cm] (qq22) {};
 \node[fit=(q23),inner sep=0.06cm] (qq23) {};
 \node[fit=(q24),inner sep=0.06cm] (qq24) {};
 \node[fit=(q1),inner sep=0.06cm] (qq1) {};
 \node[fit=(q2),inner sep=0.06cm] (qq2) {};
 \node[fit=(q3),inner sep=0.06cm] (qq3) {};
 \node[fit=(q4),inner sep=0.06cm] (qq4) {};

\begin{scope}[gray]
\tiny
 \draw[->] (q0) node[above left,yshift=0.1cm] {} -- (qq21.south);
 \draw[->] (q0) to[out=90,in=250,looseness=0.17] (qq23.south);
 \draw[->] (q0b) node[above left,yshift=0.1cm] {} -- (qq22);
 \draw[->] (q0b) to[out=90,in=270,looseness=0.15] (qq24.south);
  \draw[->] (q21) node[above left,yshift=0.1cm] {} -- (qq1);
  \draw[->] (q22) node[above left,yshift=0.1cm] {} -- (qq2);
  \draw[->] (q23) node[above left,yshift=0.1cm] {} -- (qq3);
  \draw[->] (q24) node[above left,yshift=0.1cm] {} -- (qq4);
\end{scope}

\begin{scope}[thick]
 \draw (.1,2.4) -- node[above,var] (r0) {$P_2^2$} ++(3.8,0);%
 \draw (4.1,2.4) -- node[above,var] (r0b) {$\bar{P}_2^2$} ++(3.8,0);%

 \draw[double, darkgreen] (.1,2.9) -- node[above,var] (r1) {$P_2$} ++(3.8,0);%
 \draw[double, darkgreen] (4.1,2.9) -- node[above,var] (r2) {$\bar{P}_2$} ++(3.8,0);%
\end{scope}

\node[fit=(r1),inner sep=0.06cm] (rr1) {};
 \node[fit=(r2),inner sep=0.06cm] (rr2) {};

\begin{scope}[gray]\tiny
 \draw[->] (r0) node[above left,yshift=0.1cm] {} -- (rr1.south);
\draw[->] (r0b) node[above left,yshift=0.1cm] {} -- (rr2.south);
\end{scope}
\end{tikzpicture}
\caption{The canonical model for the proof of Theorem~\ref{thm:pspace}.}
\label{fig:can-mod}
\end{figure}
We then need the rules:
\begin{align}
  C_i \leftarrow& P_j, \quad p_j \text{ occurs in } c_i,\label{eq:cnf-start}\\
  C_i \leftarrow& \bar {P_j}, \quad \neg p_j \text{ occurs in } c_i,\\
  F_0 \leftarrow& \bigwedge_{0 \leq i \leq m}C_i, \label{eq:cnf-end}
\end{align}
for $0 \leq i \leq m$, $0 \leq j \leq n$. Note that $F_0$ will hold at the moments of time corresponding to the assignments that make $\varphi_0$ true. Further, we consider the formula $\varphi_{i} = Q_{i-1} p_{i-1} \dots Q_0 p_0 \varphi_0$, for $1 \leq i \leq n+1$ (note that $\varphi_{n+1} = \varphi$), and provide rules that make $F_i$ true precisely at the moments of time corresponding to the assignments that make $\varphi_i$ true. We take
\begin{align}
  \boxplus_{[0,2^i]}F_{i+1} \leftarrow F_i \land P_i, \quad \boxminus_{[0,2^i]}F_{i+1} \leftarrow F_i \land \bar P_i, \quad & \text{if } Q_i = \exists,\label{eq:exists}\\
  \boxplus_{[0, 2^{i+1})} F_{i+1} \leftarrow \boxplus_{[0,2^i)}P_i \land \boxplus_{[0, 2^{i+1})} F_i, \quad & \text{if } Q_i = \forall, \label{eq:forall}
\end{align}
for $0 \leq i \leq n$, and, finally,
\begin{align*}
  \bot \leftarrow \boxplus_{[0, 2^{n+1})} F_{n+1}.
\end{align*}
All the rules above form the required $\nrdMTL$ program $\Pi$. We now prove that $\Pi$ is consistent with $\D$ iff $\varphi$ is not satisfiable. By Lemma~\ref{canonical}, it  suffices to show that $F_{n+1}@[0, 2^{n+1}) \in \mathfrak C_{\Pi, \D}$ iff $\varphi$ is satisfiable. For $(\Rightarrow)$, suppose $F_{n+1}@[0, 2^{n+1}) \in \mathfrak C_{\Pi, \D}$. If $Q_n = \exists$ then, in view of~\eqref{eq:exists}, either $F_{n}@[0, 2^{n}), P_{n}@[0, 2^{n}) \in \mathfrak C_{\Pi, \D}$  or $F_{n}@[2^n, 2^{n+1}), \bar P_{n}@[2^n, 2^{n+1}) \in \mathfrak C_{\Pi, \D}$. If the first option holds, we show that $\varphi_n$ is satisfiable when $p_n$ is true; if the second option holds, we show that $\varphi_n$ is satisfiable when $p_n$ is false. Similarly, if $Q_n = \forall$, then by~\eqref{eq:forall}, we have $F_{n}@[0, 2^{n}), P_{n}@[0, 2^{n}) \in \mathfrak C_{\Pi, \D}$ and  $F_{n}@[2^n, 2^{n+1}), \bar P_{n}@[2^n, 2^{n+1}) \in \mathfrak C_{\Pi, \D}$. In this case, we show that $\varphi_n$ is satisfiable when $p_n$ can be both false and true. To show that $F_{n}@[0, 2^{n}), P_{n}@[0, 2^{n}) \in \mathfrak C_{\Pi, \D}$ implies that $\varphi_n$ is satisfiable when $p_n$ is true (the other case is analogous and left to the reader), suppose $Q_{n-1} = \exists$. By~\eqref{eq:exists}, either $F_{n-1}@[0, 2^{n-1}), P_{n-1}@[0, 2^{n-1}) \in \mathfrak C_{\Pi, \D}$  or $F_{n}@[2^{n-1}, 2^{n}), \bar P_{n-1}@[2^{n-1}, 2^{n}) \in \mathfrak C_{\Pi, \D}$. (If $Q_{n-1} = \forall$, by~\eqref{eq:exists} both of these options hold.) Therefore, to show that $\varphi$ is satisfiable, it now suffices to show that (\emph{i}) $F_{n-1}@[0, 2^{n-1}), P_{n-1}@[0, 2^{n-1}) \in \mathfrak C_{\Pi, \D}$ implies that $\varphi_{n-1}$ is satisfiable when $p_n$ is true and $p_{n-1}$ is true; (\emph{ii}) $F_{n-1}@[2^{n-1}, 2^{n}), \bar P_{n-1}@[2^{n-1}, 2^{n}) \in \mathfrak C_{\Pi, \D}$ implies that $\varphi_{n-1}$ is satisfiable when $p_n$ is true and $p_{n-1}$ is false. We only consider (\emph{i}), leaving (\emph{ii})
to the reader, and after applying the argument above $n$ times, will need to show that (\emph{i}) $F_{0}@[0, 1), P_{0}@[0, 1) \in \mathfrak C_{\Pi, \D}$ implies that $\varphi_{0}$ is satisfiable when $p_n, \dots, p_{1}$ and $p_{0}$ are all true; (\emph{ii}) $F_0@[1, 2), \bar P_0@[1, 2) \in \mathfrak C_{\Pi, \D}$ implies that $\varphi_{0}$ is satisfiable when $p_n, \dots, p_{1}$ are true while $p_{0}$ is false. That (\emph{i}) holds follows from~\eqref{eq:cnf-start}--\eqref{eq:cnf-end}, and similarly for (\emph{ii}). This concludes the proof of $(\Rightarrow)$; the other direction is proved analogously.
\end{proof}

Using the techniques of~\citeA{DBLP:journals/tocl/ArtaleKRZ14}, it can be shown that nonrecursive Horn fragment of $\LTL$ is P-complete. The same complexity can be derived from the work by~\citeA{DBLP:journals/tocl/BresolinKMRSZ17} for the nonrecursive Horn fragment of the Halpern-Shoham logic $\mathcal{HS}$.

As we have just seen, the combined complexity of query answering  drops from \textsc{ExpSpace} for $\dMTL$ to \textsc{PSpace} for $\nrdMTL$. We now show that the data complexity drops to AC$^0$, which is important for practical query answering using standard database systems. Note that this result is non-trivial in view of Theorem~\ref{thm:phard}. The crux of the proof is encoding coalescing by FO-formulas with $\forall$ (which is typically not needed for rewriting atemporal ontology-mediated queries).

\begin{theorem}\label{thm:aczero}
 Consistency checking and answering $\nrdMTL$ queries is in $\text{AC}^0$ for data complexity.
\end{theorem}
\begin{proof}
We only consider a \emph{propositional} $\nrdMTL$ program $\Pi$. The proof can be straightforwardly adapted to the case of arity $\geq 1$ by adding more (object) variables to the predicates used below. Let $N$ be a set of comprising numbers or $\infty, -\infty$. We use $N + r$ as a shorthand for $\{ t + r \mid t\in N\}$ and similarly for $N - r$ (we assume that $t + \infty = \infty$ and $t - \infty = -\infty$).
For a propositional variable $P$ in $\Pi$, we define two sets $\mathsf{le}(P)$ and $\mathsf{ri}(P)$ as follows:
\begin{itemize}
\item[--] $\mathsf{le}(P) = \mathsf{ri}(P) = \{0\}$ if there is no $P'$ such that $P \lessdot P'$;

\item[--] otherwise, $\mathsf{le}(P)$ is the union of:
\begin{itemize}
\item $\bigcup_{i \in I} \mathsf{le}(P_i)$, for each $P \leftarrow \bigwedge_{i \in I} P_i$ in $\Pi$,

\item $\mathsf{le}(P_2) + r_1 \cup \mathsf{ri}(P_1)$, for each $P \leftarrow P_1 \Si_{\langle r_1, r_2 \rangle} P_2$ in $\Pi$,

\item $\mathsf{le}(P_2) - r_2 \cup \mathsf{le}(P_1)$, for each $P \leftarrow P_1 \U_{\langle r_1, r_2 \rangle} P_2$ in $\Pi$,

\item $\mathsf{le}(P_1) + r_2$, for each $P \leftarrow \boxminus_{\langle r_1, r_2 \rangle} P_1$ in $\Pi$,

\item $\mathsf{le}(P_1) - r_1$, for each $P \leftarrow \boxplus_{\langle r_1, r_2 \rangle} P_1$ in $\Pi$,
\end{itemize}
and $\mathsf{ri}(P)$ is the union of:
\begin{itemize}
\item $\bigcup_{i \in I} \mathsf{ri}(P_i)$, for each $P(\tau) \leftarrow \bigwedge_{i \in I} P_i$ in $\Pi$,

\item $\mathsf{ri}(P_2) + r_2 \cup \mathsf{ri}(P_1)$, for each $P \leftarrow P_1 \Si_{\langle r_1, r_2 \rangle} P_2$ in $\Pi$,

\item $\mathsf{ri}(P_2) - r_1 \cup \mathsf{le}(P_1)$, for each $P \leftarrow P_1 \U_{\langle r_1, r_2 \rangle} P_2$ in $\Pi$,

\item $\mathsf{ri}(P_1) + r_1$, for each $P \leftarrow \boxminus_{\langle r_1, r_2 \rangle} P_1$ in $\Pi$,

\item $\mathsf{ri}(P_1) - r_2$, for each $P \leftarrow \boxplus_{\langle r_1, r_2 \rangle} P_1$ in $\Pi$.
\end{itemize}
\end{itemize}
Using an argument that is similar to the proof of Lemma~\ref{l:zones-nr}, one can show the following:

\begin{lemma}\label{le-ri}
For any $\nrdMTL$ program $\Pi$, any data instance $\D$, and any  $P@\langle t_1, t_2 \rangle \in \mathfrak C_{\Pi, \D}$,
\begin{itemize}
\item[--] $t_1 = t_1' + n_1$, for some $n_1 \in \mathsf{le}(P)$ and some $t_1'$ such that $P'[t_1', t_1'] \in \D$ or $P'(t_1', s_2) \in \D$,

\item[--] $t_2 = t_2' + n_2$, for some $n_2 \in \mathsf{ri}(P)$ and some $t_2'$ such that $P'[t_2', t_2'] \in \D$ or $P'(s_1, t_2') \in \D$.
\end{itemize}
\end{lemma}

In view of Lemma~\ref{le-ri}, we can prove Theorem~\ref{thm:aczero} by constructing FO-formulas $\varphi_P^{\langle m, n \rangle}(x,y)$ with $m \in \mathsf{le}(P)$ and $n \in \mathsf{ri}(P)$ such that, for any data instance $\D$,
\begin{equation}\label{eq:sound-compl}
P@\langle t_1+m, t_2+n \rangle \in \mathfrak C_{\Pi, \D} \qquad \text{iff} \qquad \mathfrak A_\D \models \varphi_P^{\langle m, n \rangle}(t_1, t_2),
\end{equation}
where $\mathfrak A_\D$ is the FO-structure defined below.
To slightly simplify presentation (and without much loss of generality), we assume that all numbers in $\num(\D)$ are positive, and set
$$
\mathfrak A_\D = \bigl( \Delta, <, P_1^{[]}, P_1^{()}, \dots, P_l^{[]}, P_l^{()},  \mathsf{bit}^{\it in}, \mathsf{bit}^{\it fr}\bigr),
$$
where
\begin{itemize}
\item[--] $\Delta$ is a set of $(\ell+1)$-many elements strictly linearly ordered by $<$,   $\ell$ is the maximum of the number of distinct timestamps in $\D$ and the number of bits in the longest binary fraction in $\D$ (excluding the binary point); for simplicity, we assume that $\Delta = \{0,\dots,\ell\}$, $<$ is the natural order, and denote by $\bar n$ the $n$th fraction in $(\num(\D),<)$, counting from 0;

\item[--] $P_i^{[]} (n,n)$ holds in $\mathfrak A_\D$ iff $P_i@[\bar n,\bar n] \in \D$ and $P_i^{()} (n,m)$ holds in $\mathfrak A_\D$ iff $P_i@(\bar n,\bar m) \in \D$, for any $P_i$ occurring in $\D$;

\item[--] for $\bar n \ne \infty$, $\mathsf{bit}^{\it in}(n,i,0)$ ($\mathsf{bit}^{\it fr}(n,i,0)$) holds in $\mathfrak A_\D$ iff the $i$th bit of the integer (respectively, fractional) part of $\bar n$ is $0$, and $\mathsf{bit}^{\it in}(n,i,1)$ ($\mathsf{bit}^{\it fr}(n,i,1)$), for $i \in \Delta$, holds in $\mathfrak A_\D$ iff the $i$th bit of the integer (respectively, fractional) part of $\bar n$ is $1$ (as usual, we start counting bits from the least significant one);


\item[--] for $\bar n = \infty$, $\mathsf{bit}^{\it in}(n,i,1)$ and $\mathsf{bit}^{\it fr}(n,i,1)$ for all $i \in \Delta$.
\end{itemize}
For example, the data instance $\D = \{P[110.001, 110.001],P(10000, \infty)\}$ is given as the FO structure
$$
\mathfrak A_\D = \bigl( \Delta, <, P^{[]},P^{()}, \mathsf{bit}^{\it in}, \mathsf{bit}^{\it fr}\bigr),
$$
where $\Delta = \{0,\dots,6 \}$, $P^{[]} = \{(0, 0)\}$, $P^{()} = \{ (1,2) \}$, and
\begin{align*}
\mathsf{bit}^{\it in} = & \{(0, 0, 0), (0, 1, 1), (0, 2, 1)\} \cup \{(0, i, 0) \mid 3 \leq i \leq 6\} \cup{}\\
&\{(1, i, 0) \mid 0 \leq i \leq 3\} \cup \{(1,4,1)\} \cup \{(1,5,0)\} \cup \{(1,6,0)\} \cup{}\\
&\{(2, i, 1) \mid 0 \leq i \leq 6\}.\\
\mathsf{bit}^{\it fr} = & \{(0, 4, 1) \} \cup \{(0, i, 0) \mid 0 \leq i \leq 6,\, i \neq 4\} \cup{}\\
&\{(1, i, 0) \mid 0 \leq i \leq 6 \} \cup{}\\
&\{(2, i, 1) \mid 0 \leq i \leq 6 \}.
\end{align*}
To construct the required $\varphi_P^{\langle m, n \rangle}(x,y)$,
suppose that we have FO-formulas
\begin{itemize}
\item[--] $\mathsf{coal}_P^{\langle m, n \rangle} (x, y)$ saying that $P@\langle x+m, y+n \rangle$ is added to $\mathfrak C_{\Pi, \D}$ by an application of the rule (coal);

\item[--] $\psi_P^{\langle m, n \rangle}(x,y)$ saying that
\begin{itemize}
\item[] either $P@\langle x+m, y+n \rangle$ is added to $\mathfrak C_{\Pi, \D}$ because it belongs to the given data instance (in which case we can assume that $m=n=0$, and $\langle \rangle$ is either $()$ or $[]$),

\item[] or $P@\langle x+m, y+n \rangle$ is added to $\mathfrak C_{\Pi, \D}$ as a result of an application of one of the `logical' rules.
\end{itemize}
\end{itemize}
In this case we can set
\begin{equation*}
\varphi_P^{\langle m, n \rangle}(x,y) ~=~ \psi_P^{\langle m, n \rangle}(x,y) \lor \mathsf{coal}_P^{\langle m, n \rangle} (x, y).
\end{equation*}
Using the predicate $\mathsf{is}_{a,b}$, which is $\top$ if $a = b$ and $\bot$ otherwise, we can define $\psi_P^{\langle m, n \rangle}(x,y)$ as a disjunction of the following formulas:
\begin{itemize}
\item[--] $\mathsf{is}_{\langle, [} \land \mathsf{is}_{\rangle, ]}\land \mathsf{is}_{m, 0} \land \mathsf{is}_{n, 0} \land P^{[]}(x,y)$;

\item[--] $\mathsf{is}_{\langle, (} \land \mathsf{is}_{\rangle, )}\land \mathsf{is}_{m, 0} \land \mathsf{is}_{n, 0} \land P^{()}(x,y)$;

\item[--] for every $\displaystyle P \leftarrow  \bigwedge_{1 \le i \le k}P_i$ in $\Pi$,
\end{itemize}
\begin{multline*}
\exists x_1, y_1, \dots, x_k, y_k \hspace*{-5mm} \bigvee_{\begin{subarray}{c}
    m_1 \in \mathsf{le}(P_1)\\
    n_1 \in \mathsf{ri}(P_1)\\
    \lceil_1 \in \{ [, ( \},\
    \rceil_1 \in \{ ], ) \}
    \end{subarray}}
  \hspace*{-5mm}   \Big(\varphi_{P_1}^{\lceil_1 m_1, n_1 \rceil_1} (x_1, y_1) \land{}  \dots \land{} \hspace*{-5mm}
    \bigvee_{\begin{subarray}{c}
    m_k \in \mathsf{le}(P_k)\\
    n_k \in \mathsf{ri}(P_k)\\
    \lceil_k\in \{ [, ( \},\
    \rceil_k \in \{ ], ) \}
    \end{subarray}}
\hspace*{-5mm}     \big(\varphi_{P_k}^{\lceil_k m_k, n_k \rceil_k} (x_k, y_k) \land{}\\
     \mathsf{inter}^{\langle m, n \rangle}_{\lceil_1 m_1, n_1 \rceil_1, \dots, \lceil_k m_k, n_k \rceil_k}(x, y, x_1, y_1, \dots, x_k, y_k) \big) \dots \Big),
\end{multline*}
\begin{itemize}
\item[] where $\mathsf{inter}^{\langle m, n \rangle}_{\lceil_1 m_1, n_1 \rceil_1, \dots, \lceil_k m_k, n_k \rceil_k}(x, y, x_1, y_1, \dots, x_k, y_k)$ says that $\langle x+m, y+n \rangle$ is an intersection of $\lceil_1 x_1 + m_1, y_1 + n_1 \rceil_1 ,\dots, \lceil_k x_k + m_k, y_k + n_k \rceil_k$ (this formula can easily be defined in terms of the predicates $x+m = y+n$ and $x+m < y+n$ given below);

\item[--] for every $P \leftarrow P_1 \Si_{\range} P_2$ in $\Pi$, the formula $\sigma_{\range, P,P_1, P_2}^{\langle m, n \rangle}(x,y)$ saying that $\langle x+m, y+n \rangle$ is  $((\iota^c_1 \cap \iota_2) \mo \range) \cap \iota^c_1$ for some $\iota_1$ and $\iota_2$, where $P_1$ and $P_2$ hold, respectively
(we give a definition of $\sigma_{\range, P,P_1, P_2}^{\langle m, n \rangle}(x,y)$ in the appendix);

\item[--] analogous formulas encoding the relevant operations on intervals for the other temporal operators.

\end{itemize}
The formula $\mathsf{coal}_P^{\langle m, n \rangle} (x, y)$ is defined as follows:
\begin{equation}\label{eq:coal}
\mathsf{coal}_P^{\langle m, n \rangle} (x, y) ~=~ \forall z \hspace*{-3mm}\displaystyle\bigwedge_{\begin{subarray}{c}l \in \mathsf{le}(P) \cup \mathsf{ri}(P)
    \end{subarray}} \hspace*{-3mm} \bigl( (x + m \leq z + l) \land (z + l \leq y + m) \to \mathsf{nogap}^l_{P,\langle m, n \rangle} (z, x, y) \bigr),
\end{equation}
where $\mathsf{nogap}^l_{P,\langle m, n \rangle} (z, x, y)$ is the formula
\begin{align}
\notag &
\exists x_1, y_1, x_2, y_2, x_3, y_3    \hspace{-.3cm} \bigvee_{\begin{subarray}{c}
    m_1 \in \mathsf{le}(P)\\
    n_1 \in \mathsf{ri}(P)\\
    \lceil_1 \in \{ [, ( \},\
    \rceil_1 \in \{ ], ) \}
    \end{subarray}}
\hspace*{-3mm}
    \bigg( \psi_{P}^{\lceil_1 m_1, n_1 \rceil_1} (x_1, y_1) \land \mathsf{sub}^{\lceil_1 m_1, n_1 \rceil_1}_{\langle m, n \rangle} (x_1, y_1, x, y) \land{} \\
\notag & \hspace*{3mm}\bigvee_{\begin{subarray}{c}
    m_2 \in \mathsf{le}(P)\\
    n_2 \in \mathsf{ri}(P)\\
    \lceil_2 \in \{ [, ( \},\
    \rceil_2 \in \{ ], ) \}
    \end{subarray}} \hspace*{-3mm} \Big( \psi_{P}^{\lceil_2 m_2, n_2 \rceil_2} (x_2, y_2) \land \mathsf{sub}^{\lceil_2 m_2, n_2 \rceil_2}_{\langle m, n \rangle} (x_2, y_2, x, y) \land{} \\
\label{eq:coal-1}&  \hspace*{1.7cm}  \Big((x_1 + m_1 < z + l < y_1 + n_1) \lor{} \\
\label{eq:coal-2}&   \hspace*{1.9cm} (x_1 + m_1 < y_1 + l_1 = z + l = x_2 + m_2 < y_2 + n_2) \land \mathsf{is}_{\rceil_1, ]} \lor \mathsf{is}_{\lceil_2, [}) \lor{}\\
\notag & \hspace*{1.3cm}  \bigvee_{\begin{subarray}{c}
    m_3 \in \mathsf{le}(P)\\
    n_3 \in \mathsf{ri}(P)
    \end{subarray}} \big( \psi_{P}^{[ m_3, n_3 ]} (x_3, y_3) \land \mathsf{sub}^{[ m_3, n_3 ]}_{\langle m, n \rangle} (x_3, y_3, x, y) \land{}\\
\label{eq:coal-3} & \hspace*{3cm}    \big[(x_3 + m_3= y_3+n_3 = z + l = x+m = x_1 + m_1 < y_1 + n_1) \lor{} \\
\label{eq:coal-4} &  \hspace*{3.2cm}  (x_1 + m_1 < y_1 + n_1 = z+l = y+ n = x_3 + m_3= y_3+n_3) \lor{} \\
\label{eq:coal-5} &  \hspace*{1.2cm}  (x_1 + m_1 < y_1 + l_1 = z + l = x_3 + m_3= y_3+n_3 = x_2 + m_2 < y_2 + n_2)\big] \big) \Big)\bigg)
\end{align}
and $\mathsf{sub}^{\lceil m', n' \rceil}_{\langle m, n \rangle}(x', y', x, y)$ says that $\lceil x' + m', y' + n' \rceil$ is a subinterval of $\langle x+m, y+n \rangle$. Intuitively, $\mathsf{nogap}^l_{P,\langle m, n \rangle} (z, x, y)$ says that around the time instant $z+l$ (that is, to the left and right of it as well as at $z+l$ itself),  their is no subinterval of $\langle x+m, y+n \rangle$ that is not covered by $P$. The five cases considered in the formula $\mathsf{nogap}^l_{P,\langle m, n \rangle} (z, x, y)$ are illustrated in Fig.~\ref{fig:coal}.

\begin{figure}[h]
\centering
\begin{tikzpicture}[point/.style={draw, thick, circle, inner sep=1.5, outer
		sep=2}]

  \node at (-4, 0) {Case~\eqref{eq:coal-1}};

  \coordinate (xpm) at (0,0);
  \node[anchor = south] at (xpm) {$x+m$};

  \coordinate (ypn) at (8,0);
  \node[anchor = south] at (ypn) {$y+n$};

  \node at (xpm) [yshift=-9cm] {} edge[thin, dashed] (xpm);
  \node at (ypn) [yshift=-9cm] {} edge[thin, dashed] (ypn);

  \node[yshift=-.5cm] at (2, 0) {\color{darkgreen}$x_1+m_1$};
  \node[yshift=-.5cm] at (6, 0) {\color{darkgreen}$y_1+n_1$};

  \node at (2,0) {\color{darkgreen}$\langle$};
  \node at (6,0) {\color{darkgreen}$\rangle$};
  \draw[darkgreen] (2,0) -- (6,0);

  \node[point, yshift=.2cm, red, label=above:{\color{red}$z+l$}] at (4,0) {};


  \node at (-4, -2) {Case~\eqref{eq:coal-2}};

  \node[label = below:{\color{darkgreen}$x_1+m_1$}] at (2,-2) {\color{darkgreen}$\color{darkgreen}\langle$};
  \node[label = below:{\color{darkgreen}$y_1+n_1$}] at (3.8,-2) {\color{darkgreen}$)$};
  \draw[darkgreen] (2,-2) -- (3.8,-2);

  \node[label = {[yshift=-.5cm]below:{\color{blue}$x_2+m_2$}}] at (4.2,-2) {\color{blue}$[$};
  \node[label = {[yshift=-.5cm]below:{\color{blue}$y_2+n_2$}}] at (6,-2) {\color{blue}$\rangle$};
  \draw[blue] (4.2,-2) -- (6,-2);

  \node[point, yshift=.2cm, red, label=above:{\color{red}$z+l$}] at (4,-2) {};


    \node at (-4, -4) {Case~\eqref{eq:coal-3}};

  \node[label = below right:{\color{darkgreen}$x_1+m_1$}] at (0.2,-4) {\color{darkgreen}$\langle$};
  \node[label = {[yshift=.15cm]below:{\color{darkgreen}$y_1+n_1$}}] at (3.8,-4) {\color{darkgreen}$\rangle$};
  \draw[darkgreen] (0.2,-4) -- (3.8,-4);

  \node[point, yshift=.2cm, red, label=above left:{\color{red}$z+l$}] at (0,-4) {};

  \node[label = below:{\color{orange}$x_3+m_3 \ \ \ y_3+ n_3$}] at (0,-4.4) {\color{orange}$[\ ]$};


    \node at (-4, -6) {Case~\eqref{eq:coal-4}};

  \node[label = {[yshift=.15cm]below:{\color{darkgreen}$x_1+m_1$}}] at (4.2,-6) {\color{darkgreen}$\langle$};
  \node[label = below left:{\color{darkgreen}$y_1+n_1$}] at (7.8,-6) {\color{darkgreen}$\rangle$};
  \draw[darkgreen] (4.2,-6) -- (7.8,-6);

  \node[point, yshift=.2cm, red, label=above left:{\color{red}$z+l$}] at (8,-6) {};

  \node[label = below:{\color{orange}$x_3+m_3 \ \ \ y_3+ n_3$}] at (8,-6.4) {\color{orange}$[\ ]$};


    \node at (-4, -8) {Case~\eqref{eq:coal-5}};

   \node[label = {[yshift=-1.5cm]below:{\color{darkgreen}$x_1+m_1$}}] at (2,-8) {\color{darkgreen}$\langle$};
  \node[label = {[yshift=-1.5cm]below:{\color{darkgreen}$y_1+n_1$}}] at (3.8,-8) {\color{darkgreen}$)$};
  \draw[darkgreen] (2,-8) -- (3.8,-8);

  \node[label = {[yshift=-1cm]below:{\color{blue}$x_2+m_2$}}] at (4.2,-8) {\color{blue}$($};
  \node[label = {[yshift=-1cm]below:{\color{blue}$y_2+n_2$}}] at (6,-8) {\color{blue}$\rangle$};
  \draw[blue] (4.2,-8) -- (6,-8);

  \node[point, yshift=.2cm, red, label=above:{\color{red}$z+l$}] at (4,-8) {};

  \node[label = below:{\color{orange}$x_3+m_3 \ \ \ y_3+ n_3$}] at (4,-8.4) {\color{orange}$[\ ]$};
\end{tikzpicture}
\caption{Five cases of the formula $\mathsf{nogap}^l_{P,\langle m, n \rangle} (z, x, y)$.}
\label{fig:coal}
\end{figure}
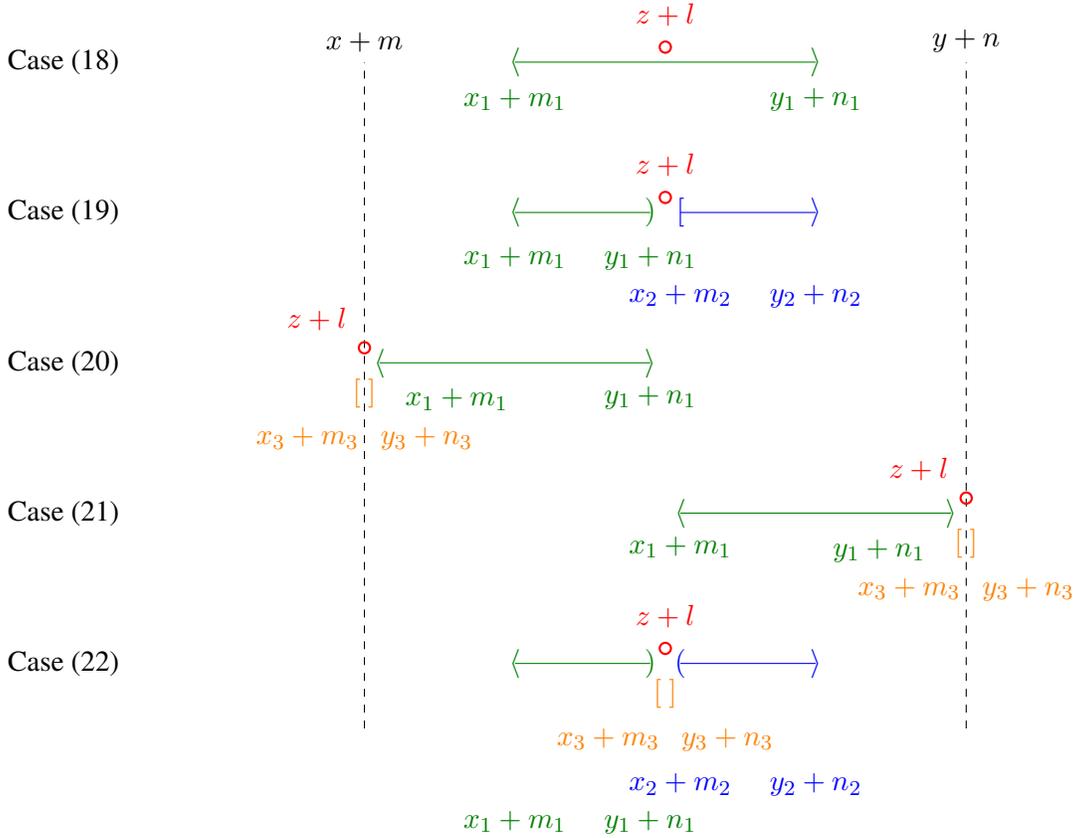

When evaluating $\varphi^{\langle m, n \rangle}(x,y)$ over $\mathfrak A_\D$, we need to compute the truth-values of $x+m = y+n$ and $x+m < y+n$ (for fixed $m$ and $n$). We regard the former as a formula with the predicates $\mathsf{bit}^{in}$, $\mathsf{bit}^{fr}$ and $<$ that is true  just in case $x = y + (n-m)$ if $n \geq m$, and $y = x+(m-n)$ otherwise. We provide a definition of $x = y + c$, for a positive $c$, in the appendix. A formula expressing $x+m < y+n$ is constructed similarly and left to the reader.

Finally, we show how the formulas $\varphi_P^{\langle m, n \rangle}(x,y)$ defined above can be used to check whether an interval $\iota = \langle \iota_b, \iota_e \rangle$ is a certain answer to $(\Pi, P@x)$ over $\D$. As follows from Lemma~\ref{le-ri}, if $\bot@ \lceil t_1, t_2 \rceil \in \mathfrak C_{\Pi, \D}$ then, for some $m \in \mathsf{le}(\bot)$, $n \in \mathsf{ri}(\bot)$ and some numbers $t_1', t_2' \in \mathsf{num}(\D)$ such that $t_1'$ ($t_2'$) occurs as the left (right) end of some interval, we have $t_1 = t_1' + m$ and $t_2 = t_2' + n$.
Take the structure $\mathfrak A_\D^{\iota}$ that extends $\mathfrak A_\D$ with the numbers $\iota_b$ and $\iota_e$. By~\eqref{eq:sound-compl}, $\iota$ is a certain answer to $(\Pi, P@x)$ over $\D$ iff the formula
\begin{multline}\label{rewr}
\exists x, y \bigvee_{\begin{subarray}{c}
    m \in \mathsf{le}(\bot)\\
    n \in \mathsf{ri}(\bot)
    \end{subarray}} \varphi_\bot^{\lceil m, n \rceil}(x,y) \lor{}\\
     \exists x, y, x_1, y_1 \hspace*{-7mm}\bigvee_{\begin{subarray}{c}
    m_1 \in \mathsf{le}(\bot)\\
    n_1 \in \mathsf{ri}(\bot)\\
    \lceil_1 \in \{ [, ( \},\
    \rceil_1 \in \{ ], )\}
    \end{subarray}} \hspace*{-7mm} \big(  \varphi_P^{\lceil_1 m, n \rceil_1}(x_1, y_1) \land \mathsf{sub}^{\langle 0, 0 \rangle}_{\lceil_1 m, n \rceil_1}(x,y, x_1, y_1) \land (x = \iota_b) \land (y = \iota_e) \big)
\end{multline}
holds true in $\mathfrak A_\D^{\iota}$.
\end{proof}

\section{Implementing $\nrdMTL$}\label{implementing}

Unfortunately, the (data independent) FO-rewriting~\eqref{rewr} turns
out to be impractical because of the universal quantifier used for
coalescing in~\eqref{eq:coal}.  It is well known that $\forall$ is implemented in SQL as $\neg \exists \neg$ resulting in suboptimal performance in general.
Having experimented with a few different approaches, we decided to use a materialisation (bottom-up) technique. In this section, we first present a bottom-up algorithm whose worst-case running time is linear in the number of intervals of an input data instance $\D$, under a practically motivated assumption that the order of occurrence of the intervals in $\D$ coincides with the natural temporal order on those intervals. Then we describe how our algorithm can be implemented in SQL (with views). In particular, we consider two alternative implementations of coalescing in SQL.

\subsection{Bottom-up algorithm}

We first introduce some notation and obtain a few results about \emph{temporal tables} $T$ with column names $\mathsf{attr}_1, \dots, \mathsf{attr}_m, \mathsf{lpar}, \mathsf{ledge}, \mathsf{redge}, \mathsf{rpar}$. A temporal table with $m=0$ will be called \emph{purely temporal}.
%
%
We refer to the $i$-th row of $T$ as $T[i]$, to the value of the column $\mathsf{attr}_j$ in the $i$-th row as $T[i, \mathsf{attr}_j]$, and set $T[i, \mathsf{attr}_j,\dots, \mathsf{attr}_k] = (T[i, \mathsf{attr}_j],\dots, T[i, \mathsf{attr}_k])$. We assume that the columns $\mathsf{ledge}$ and $\mathsf{redge}$ store timestamps or special values for $\infty, -\infty$, $\mathsf{lpar}$ stores $[$ or $($, and $\mathsf{rpar}$ stores $]$ or $)$. Define an order $\prec$ on intervals by taking $\langle t_1, t_2 \rangle \prec \lceil s_1, s_2 \rceil$ iff one of the following conditions holds:
\begin{itemize}
\item[--] $t_1 < s_1$;

\item[--] $t_1 = s_1$, $\langle$ is $[$, and $\lceil$ is $($;

\item[--] $t_1 = s_1$, $\langle$ and $\lceil$ are the same, and $t_2 < s_2$;

\item[--] $t_1 = s_1$, $\langle$ and $\lceil$ are the same, $t_2 = s_2$, $\rangle$ is $)$, and $\lceil$ is $]$.
\end{itemize}
It should be clear that $\prec$ is a strict linear order on the set of all intervals. For example, we have $[3, 8) \prec [4, 7) \prec (4,6) \prec (4,7) \prec (4,7]$. (In fact, the results of this section will work with any other linear order over intervals.) We write $T[i, \mathsf{lpar}, \mathsf{ledge}, \mathsf{redge}, \mathsf{rpar}] \prec T'[j, \mathsf{lpar}, \mathsf{ledge}, \mathsf{redge}, \mathsf{rpar}]$ to say that the interval defined by the $i$th row of a temporal table $T$ $\prec$-precedes the interval given by the $j$th row of a temporal table $T'$.

We make the following \emph{temporal ordering assumption} (or TOA), for any temporal table $T$ with $m$ attributes:
\begin{multline*}
\text{if $T[i, \mathsf{attr}_1, \dots, \mathsf{attr}_m] = T[j, \mathsf{attr}_1, \dots, \mathsf{attr}_m]$ and $i < j$},\\
\text{then $T[i, \mathsf{lpar}, \mathsf{ledge}, \mathsf{redge}, \mathsf{rpar}] \preceq T[j, \mathsf{lpar}, \mathsf{ledge}, \mathsf{redge}, \mathsf{rpar}]$}.
\end{multline*}
For a purely temporal table $T$, this assumption means that the rows of $T$ respect $\preceq$.


Let $T[\mathsf{attr}_j, \dots, \mathsf{attr}_k]$ be the  \emph{projection} of $T$ on the columns $\mathsf{attr}_j, \dots, \mathsf{attr}_k$ that keeps only distinct tuples. We define $|T|_{o}$ to be the cardinality of $T[\mathsf{attr}_1,\dots, \mathsf{attr}_m]$ and $|T|_t$ to be the cardinality of $T[\mathsf{lpar}, \mathsf{ledge},\mathsf{redge}, \mathsf{rpar}]$. The first measure estimates how large the table is in terms of individual constants, while the second measure concerns the number of timepoints. For the tables of extensional predicates in our use-cases, $|T|_{o}$ is much smaller than $|T|_t$.

We say that a table $T$ is \emph{coalesced} if it does not contain distinct tuples $(c_1, \dots, c_m, \langle, t_1, t_2, \rangle)$ and $(c_1, \dots, c_m, \lceil, t_1', t_2', \rceil)$ such that $\langle t_1, t_2 \rangle \cap \lceil t_1', t_2' \rceil \neq \emptyset$. For a tuple of individual constants $(c_1, \dots, c_m)$, let $T_{c_1, \dots, c_m}$ be the set of all intervals $\langle t_1, t_2 \rangle$ such that $(c_1, \dots, c_m, \langle, t_1, t_2, \rangle)$ occurs in $T$. For a set $\mathcal I$ of intervals, we then denote by $\mathsf{coalesce}(\mathcal I)$ the (minimal) set of intervals that results from coalescing $\mathcal I$. Finally, a \emph{coalescing} of $T$ is a minimal table, $T^*$, with the same columns as $T$ such that the following condition holds:
\begin{description}
\item[(coalesce)] for any $(c_1, \dots, c_m)$ in $T[\mathsf{attr}_1,\dots, \mathsf{attr}_m]$ and $\langle t_1, t_2 \rangle$ in $\mathsf{coalesce}(T_{c_1, \dots, c_m})$, there exists $(c_1, \dots, c_m, \langle t_1, t_2 \rangle)$ in $T^*$.
\end{description}
Clearly, $T^*$ is a coalesced table.

\begin{lemma}\label{lemma:coal}
Suppose a table $T$ satisfies TOA. Then its coalescing $T^*$ satisfying TOA and such that $|T^*|_o = |T|_o$ and $|T^*|_t \leq |T|_t$ can be computed in time $O(|T|_{o}^2 \times |T|_{t})$.
\end{lemma}
\begin{proof}
Consider first a purely temporal table $S$ that satisfies temporal ordering. There is a simple linear-time algorithm to produce a coalesced table $S^*$ that also satisfies temporal ordering. Indeed, initially we set $\langle b, e \rangle = S[0, \mathsf{lpar}, \mathsf{ledge}, \mathsf{redge}, \mathsf{rpar}]$. In a loop, we take each $\lceil t_1, t_2 \rceil = S[i, \mathsf{lpar}, \mathsf{ledge}, \mathsf{redge}, \mathsf{rpar}]$ (clearly, $\langle b, e \rangle \prec \lceil t_1, t_2 \rceil$). If $\lceil t_1, t_2 \rceil$ and $\langle b, e \rangle$ are disjoint, we add $\langle b, e \rangle$ to $S^*$ and set $\langle b, e \rangle = \lceil t_1, t_2 \rceil$. If they are not disjoint, we set $\langle b, e \rangle=\lceil t_1, t_2 \rceil \cup \langle b, e \rangle$ and move on. It is easily checked that the resulting table $S^*$ is as required. Below, we refer to this algorithm as an \emph{imperative} coalescing algorithm.

It only remains to explain how the algorithm above can be applied to
$T$ in order to obtain the required complexity. Note that
$|T| \leq |T|_{o} \times |T|_{t}$ and we can construct $|T|_{o}$-many
separate tables $T_{c_1, \dots, c_m}$, for each $(c_1, \dots, c_m)$,
in time $|T| \times |T|_{o}$. Then, we can apply the algorithm
described above to each $T_{c_1, \dots, c_m}$ in time $|T|_{t}$ and
merge the results. Therefore, the overall running time is
$|T| \times |T|_{o} + |T|_{t} \times |T|_o = O(|T|_{o}^2 \times
|T|_{t})$.
\end{proof}

Before presenting our query answering algorithm, we determine the
complexity of computing \emph{temporal joins}. Let $T$ be a table with
attributes $\mathsf{attr}_1, \dots, $
$\mathsf{attr}_m, \mathsf{lpar}, \mathsf{ledge}, \mathsf{redge},
\mathsf{rpar}$ and let $T'$ be a table with attributes $\mathsf{attr}_1', \dots, $
$\mathsf{attr}_n', \mathsf{lpar}, \mathsf{ledge}, \mathsf{redge},
\mathsf{rpar}$. A \emph{temporal join} of $T$ and $T'$ is a table
$T''$ with attributes
$\mathsf{attr}_1'', \dots, \mathsf{attr}_k'', \mathsf{ledge},
\mathsf{redge}, \mathsf{rpar}$ such that
$$
\{\mathsf{attr}_1'', \dots, \mathsf{attr}_k''\} = \{ \mathsf{attr}_1,
\dots, \mathsf{attr}_m\} \cup \{ \mathsf{attr}_1', \dots,
\mathsf{attr}_n' \}
$$
and
$(c_1'', \dots, c_k'', \langle, t_1'', t_2'', \rangle)$ is in $T''$
iff there exist two tuples $(c_1, \dots, c_m, \lceil, t_1, t_2, \rceil)$ from $T$ and $(c_1', \dots, c_n', \lfloor, t_1', t_2', \rfloor)$ from $T'$ satisfying
the following conditions:
\begin{itemize}
\item[--] $c_i'' = c_j$, for all $i,j$ such that $\mathsf{attr}_i'' =\mathsf{attr}_j$;
\item[--] $c_i'' = c_j'$, for all $i,j$ such that $\mathsf{attr}_i'' =\mathsf{attr}_j'$;
\item[--] $\lceil t_1, t_2 \rceil \cap \lfloor t_1', t_2' \rfloor \neq \emptyset$ and $\langle t_1'', t_2'' \rangle = \lceil t_1, t_2 \rceil \cap \lfloor t_1', t_2' \rfloor$.
\end{itemize}

\begin{lemma}
  If $T$, $T'$ satisfy TOA, then a temporal join $T''$ of $T$ and $T'$
  satisfying TOA and such that $|T''|_o \leq |T|_o \times |T'|_o$,
  $|T''|_t \leq |T|_t + |T'|_t$ can be computed in time
  $O(|T|_o^2 \times |T'|_o^2 \times (|T|_t + |T'|_t))$.
\end{lemma}
\begin{proof}
  We first give an algorithm for computing the temporal join of purely
  temporal tables $S$ and $S'$. We assume that these tables are
  coalesced (which can be done in time $O(|S|)$ and $O(|S'|)$). The
  algorithm works starting from the first tuples $S[i]$ and $S'[i']$ of the tables. If
  $S[i] \cap S'[i'] \neq \emptyset$, we write $S[i] \cap S'[i']$ to
  the output table $S''$. Then, if $S[i+1] \succeq S'[i'+1]$, we set
  $i' := i' + 1$ (without changing $i$); otherwise, $i := i+1$. We
  iterate until we have considered all the tuples in both
  tables. Clearly, computing the full $S''$ requires time
  $O(|S|+|S'|)$.

  The complete algorithm for the tables $T$ and $T'$ will first, similarly to the argument of Lemma~\ref{lemma:coal}, produce $|T|_o$-many
  purely temporal tables $T_{c_1, \dots, c_m}$, for each
  $(c_1, \dots, c_m)$ occurring in $T$. Note that
  $|T_{c_1, \dots, c_m}| \leq |T|_t$ for each of those
  tables. In the same way, we produce $|T'|_o$ purely temporal tables $T'_{c_1', \dots, c_n'}$, for
  each $(c_1', \dots, c'_n)$ occurring in $T'$. It remains to apply
  the temporal join algorithm described above to all pairs of tables
  $T_{c_1, \dots, c_m}$ and $T'_{c_1', \dots, c_n'}$, which can be
  done in the required time.
\end{proof}

Another operation on temporal tables we need is projection. Let $T$ be
a table with column names as above and let
$\{\mathsf{attr}_1', \dots, $
$\mathsf{attr}_n'\} \subseteq \{\mathsf{attr}_1, \dots, $
$\mathsf{attr}_m\}$. A \emph{projection} of $T$ on
$\mathsf{attr}_1', \dots, $ $\mathsf{attr}_n'$ is a table with columns
$\mathsf{attr}_1', \dots, $
$\mathsf{attr}_n', \mathsf{lpar}, \mathsf{ledge}, \mathsf{redge},
\mathsf{rpar}$ containing all
$(c_1', \dots, c_n', \langle t_1, t_2 \rangle)$ such that some
$(c_1, \dots, c_m, \langle t_1, t_2 \rangle)$ is in $T$ and $c_i'=c_j$
whenever $\mathsf{attr}_i' = \mathsf{attr}_j$.
As we have to preserve the temporal order, our algorithm for computing projections requires some attention. To show that a na\"ive projection does not preserve the temporal order, consider a table $T$ with two tuples $(a, [,1,1,])$ and $(b, [,0,0,])$, which satisfies our temporal order assumption. The projection of $T$ that removes the first column results is the table with two tuples $([,1,1,])$ and $([,0,0,])$, which is not ordered.
\begin{lemma}\label{lemma:proj}
If $T$ satisfies TOA, then a projection of $T$ satisfying TOA can be computed in time $O(|T|_o^2 \times |T|_t)$.
\end{lemma}

Now, consider the \emph{union} operation on pairs of tables $T$ and $T'$ with the same columns that returns a table with all the tuples from the set $T \cup T'$.
\begin{lemma}\label{lemma:union}
For any pair of tables $T$ and $T'$ satisfying TOA, their union table also satisfying TOA can be computed in time $O((|T|_o^2 + |T'|_o^2) \times (|T|_t + |T'|_t))$.
\end{lemma}

The proofs of Lemmas~\ref{lemma:proj} and~\ref{lemma:union} can be found in the  appendix.

We are now in a position to describe the bottom-up query answering algorithm. Suppose we are given a program $\Pi$ in normal form. Suppose also that each extensional predicate $P$ is given by a table $T_P$ satisfying TOA. (This assumption can be made in all of our use-cases. Indeed, both tables \TB{} and  $\mathsf{Weather}$ are naturally ordered by the timestamp, and our mappings (see Section~\ref{sec:uc}) can be easily written in a way to take advantage of this order and produce tables $T$ satisfying TOA.) Thus, we can assume that the given data instance $\D$ is represented by a set of $T_P$, where each $T_P$ contains all the tuples $(c_1, \dots, c_m, \langle, t_1, t_2, \rangle)$ such that $P(c_1, \dots, c_m)@ \langle t_1, t_2 \rangle \in \D$.

Consider a predicate $P$ and suppose that we have computed temporal tables $T_{P_i}$ satisfying TOA, for each $P_i$ with $P \lessdot P_i$ (see Section~\ref{nonrecursive}). We assume that the $T_{P_i}$ have (non-temporal) columns $(P_i, 1),  \dots, (P_i, m)$. For each rule $\alpha$ in $\Pi$ with $P$ in the head, we compute a table $T_P^\alpha$ satisfying TOA. If $\alpha$ is of the form~\eqref{eq:nf-horn}, we first compute the temporal join $T$ of $T_{P_1}, \dots, T_{P_I}$ (we change the names so that $T_{P_i}$ has columns $(P_i,\tau_1, 1), \dots, (P_i,\tau_m, m)$, where $\avec{\tau}_i = (\tau_1, \dots, \tau_m)$, and so all the tables $T_{P_i}$ have distinct column names). Then we select from $T$ only those tuples $(c_1, \dots, c_n, \langle, t_1, t_2, \rangle)$ for which $c_i = c_j$ in case the column names for $c_i$ and $c_j$ mention the same variable $x$, and the tuples for which $c_i = a$ in case the column name for $c_i$ mentions the constant $a$. These two steps can be done in time $O(\prod_{i}  |T_{P_i}|_o^2 \times \sum_{i} |T_{P_i}|_t)$, and the size of the resulting table does not exceed $\prod_{i}  |T_{P_i}|_o \times \sum_{i} |T_{P_i}|_t$. It remains to perform projection in the following way. Suppose  $P(\avec{\tau})$ with $\avec{\tau} = (x_1, \dots x_m)$ is the head of $\alpha$ (if $\avec{\tau}$ also contains constants, the procedure below can be easily modified). Then we keep only one column among all the columns named $(P_i,x_j, k)$, for each variable $x_j$. It remains to rename the remaining $(P_i,x_j, k)$ to $(P,j)$, for each $j$. The total time required to compute $T_P^\alpha$ is $O(\prod_{i}  |T_{P_i}|_o^2 \times \sum_{i} |T_{P_i}|_t)$.

If $\alpha$ is of the form~\eqref{eq:nf-from-box}, provided that $T_{P_1}$ is coalesced, computing $T_P^\alpha$ reduces to using arithmetic operations for $\iota \pc \range$, $\iota \mc \range$, and $\range \sqsubseteq \iota$ as in the rules $(\boxplus_{\range})$/$(\boxminus_{\range})$, and projection. Therefore, $T_P^\alpha$ satisfying TOA can be computed in time $|T_{P_1}|_o^2 \times |T_{P_1}|_t$. Computing $T_P^\alpha$ for rules of the form~\eqref{eq:nf-to-box} can be done in time $O(|T_{P_1}|_o \times |T_{P_2}|_o \times (|T_{P_1}|_t + T_{P_2}|_t))$. Indeed, to construct $T_P^\alpha$ for a rule $\alpha$ of the form $P(\avec{\tau}) \leftarrow P_1 (\avec{\tau}_1)\Si_\range P_2(\avec{\tau}_2))$, we follow the rule $(\Si_{\range})$ and first produce a table $T_{P_1}^c$ with the same columns as $T_{P_1}$, where for each tuple of $T_{P_1}$, we apply the operation $\cdot^c$ to its interval. We then compute the temporal join $T$ of $T_{P_1}^c$ and $T_{P_2}$ after applying the renaming described above. Then we compute $T^{+^o \varrho}$ by applying the operation $+^o \varrho$ to the interval columns of each tuple in $T$, after which we compute the temporal join of $T^{+^o \varrho}$ and $T_{P_1}^c$ (with renaming applied to the columns of $T_{P_1}^c$). To produce $T_P^\alpha$, it remains to perform projection and renaming as described above. Finally, to compute $T_P$, it is sufficient to compute the union of all $T_P^\alpha$ satisfying TOA. Thus, we obtain the following, where the \emph{degree} of the  rule~\eqref{eq:nf-horn} is $|I|$, of~\eqref{eq:nf-to-box} is $2$, and of~\eqref{eq:nf-from-box} is $1$:
\begin{lemma}
  Let $\Pi$ be a program and $P$ a predicate in it such that $K$-many  rules have $P$ in the head, with $R$ being the maximal degree of those rules, $m$ the maximum of $|T_{P'}|_t$ among $P'$ such that $P \lessdot P'$, and $n$ the maximum of $|T_{P'}|_o$ among those $P'$. Then $T_P$ is of size at most $n^{R}mRK$ and can be computed in time $O(n^{2R}mRK)$.
\end{lemma}
To compute the table for the goal $Q$, we iterate the described
procedure as many times as the length of the longest chain of
predicates in the dependence relation $\lessdot$ for $\Pi$. Thus, we obtain:
\begin{theorem}
  Let $m$ be the maximum of $|T_{P}|_t$ among the extensional predicates $P$, and $n$ the maximum of $|T_{P}|_o$ among those $P$. The overall time required to compute the goal predicate $Q$ of $\Pi$ is exponential in the size of $\Pi$, polynomial in $n$, and linear in $m$.
\end{theorem}
Note that if all $T_P$ are extracted from one table $\mathcal{R}$, as in our use-cases, then $n$ corresponds to the number of individual tuples in $\mathcal{R}$, whereas $m$ to the number of temporal intervals. It is to be emphasised that, in practice, programs tend to be small, and the number of individual constants is also small compared to the number of temporal intervals. The theorem above explains the linear patterns in
our experiments below, where the size of individual tuples is
fixed.


\subsection{Implementation in SQL}
\label{subsec:imp-in-sql}

Now, we show how to rewrite a given $\nrdMTL$ query $(\Pi, Q(\avec{\tau})@x)$ with $\Pi$ in normal form
\eqref{eq:nf-horn}--\eqref{eq:nf-from-box} to an SQL query computing
the certain answers $(\avec{c},\iota)$ to the query with
\emph{maximal} intervals $\iota$. We illustrate the idea by a (relatively) simple example.

Consider the $\nrdMTL$ query $(\Pi, \mathsf{HeatAffectedCounty}(\mathsf{county})@x)$, where
\begin{align*}
  \Pi = \{ &     \boxminus_{[0,24h]}  \mathsf{ExcessiveHeat}(v) \leftarrow  \boxminus_{[0,24h]} \mathsf{TempAbove24}(v) \land    \diamondminus_{[0,24h]} \mathsf{TempAbove41}(v),
\\
&   \mathsf{HeatAffectedCounty}(v) \leftarrow \mathsf{LocatedInCounty}(u,v) \land  \mathsf{ExcessiveHeat}(u) \}
\end{align*}
is part of the meteorological ontology from Section~\ref{sec:uc}. First, we transform $\Pi$ to normal form:
\begin{align*}
   \Pi = \{ & \mathsf{ExcessiveHeat}(v) \leftarrow \diamondplus_{[0,24h]} \mathsf{X} (v), \ \     \mathsf{X}(v) \leftarrow \mathsf{Y}(v) \land \mathsf{Z}(v),
   \\
    & \mathsf{Y}(v) \leftarrow \boxminus_{[0,24h]} \mathsf{TempAbove24}(v),\ \ \mathsf{Z}(v) \leftarrow \diamondminus_{[0,24h]} \mathsf{TempAbove41}(v),
\\
&   \mathsf{HeatAffectedCounty}(v) \leftarrow \mathsf{LocatedInCounty}(u,v) \land  \mathsf{ExcessiveHeat}(u) \}.
\end{align*}
%
%
We regard $\mathsf{TempAbove24}$, $\mathsf{TempAbove41}$, $\mathsf{LocatedInCounty}$ as extensional predicates given by the tables $T_{\mathsf{TempAbove24}}$, $T_{\mathsf{TempAbove41}}$, $T_{\mathsf{LocatedInCounty}}$. The first two of these tables have  columns $\mathsf{station\_id}, \mathsf{ledge}, \mathsf{redge}$, and the third one $\mathsf{station\_id}, \mathsf{county}, \mathsf{ledge}, \mathsf{redge}$. To simplify presentation, we omit the columns $\mathsf{lpar}$ and $\mathsf{rpar}$ used in the previous section and assume that all the temporal intervals take the form $(t,t']$; see Section~\ref{sec:uc}.

For each predicate $P$ in $\Pi$, we also create a view (temporary
table) $V_P^*$
with the same columns as $T_P$. We set
$V_P^* = \mathsf{coalesce}(T_P)$, where $\mathsf{coalesce}$ is a query
that implements coalescing in SQL\footnote{It should not be confused with the standard
  coalesce function in SQL that returns the first of its
  arguments that is not null, or null if all of the arguments are null.} We explain the idea behind this query for a temporal table $T$ (as
mentioned above, we omit columns $\mathtt{lpar}, \mathtt{rpar}$). For
a moment of time $t$ occurring in $T$, we denote by $b^\geq(T,t)$ the
number of $i$ such that $T[i, \mathtt{ledge}] \geq t$, and by
$e^\geq(T,t)$ the number of $i$ such that
$T[i, \mathtt{redge}] \geq t$; the numbers $b^\leq(T,t)$ and
$e^\leq(T,t)$ are defined analogously. It can be readily seen that
every $t$ in $T[\mathtt{ledge}]$ such that $b^\geq(T,t) = e^\geq(T,t)$
is the beginning of some interval in the coalesced table
$T^*$. Similarly, every $t'$ in $T[\mathtt{redge}]$ such that
$b^\leq(T,t') = e^\leq(T,t')$ is the end of some interval in
$T^*$. The coalesced intervals of $T^*$ can be then obtained as pairs
$(t,t'']$, where $t$ is as above and $t''$ is the minimum over those
$t'$ defined above that are $\ge t$. Thus, to coalesce
$T_{\mathsf{TempAbove24}}$ we first use the query
\begin{align*}
  V_l = ~& \mathtt{SELECT} \ T.\mathsf{station\_id} \ \mathtt{AS}\ \mathsf{station\_id},\ T.\mathsf{ledge} \ \mathtt{AS}\ \mathsf{ledge}\  \mathtt{FROM}\ T_{\mathsf{TempAbove24}}\ T\ \mathtt{WHERE}\\
& (\mathtt{SELECT\ COUNT(*)\ from} \ T_{\mathsf{TempAbove24}}\ S \ \mathtt{WHERE} \ S.\mathsf{ledge} \geq \ T.\mathsf{ledge} \ \mathtt{AND}\\
 & \quad \quad S.\mathsf{station\_id} = T.\mathsf{station\_id}) =\\
& (\mathtt{SELECT\ COUNT(*)\ from} \ T_{\mathsf{TempAbove24}}\ S \ \mathtt{WHERE} \ S.\mathsf{redge} \geq \ T.\mathsf{ledge} \ \mathtt{AND}\\
& \quad \quad S.\mathsf{station\_id} = T.\mathsf{station\_id}),
\end{align*}
which extracts the pairs $(n,t)$, where $t$ is as described above and $\mathsf{station\_id} = n$. An analogous query can be used to produce $V_r$, a table of pairs $(n,t')$, where $t'$ is as described above and $\mathsf{station\_id} = n$. Finally, we set
\begin{align*}
   V_{\mathsf{TempAbove24}}^* =~  &\mathtt{SELECT} \ V_l.\mathsf{station\_id} \ \mathtt{AS}\ \mathsf{station\_id},\ V_l.\mathsf{ledge} \ \mathtt{AS}\ \mathsf{ledge}, \\
   & \quad \quad (\mathtt{SELECT} \ \mathtt{MIN} \ (V_r.\mathsf{redge}) \ \mathtt{FROM}\  V_r \ \mathtt{WHERE} \ V_r.\mathsf{redge} \geq V_l.\mathsf{ledge}\ \mathtt{AND} \\
   & \quad \quad \quad \quad \ V_l.\mathsf{station\_id} = V_r.\mathsf{station\_id})\ \mathtt{AS}\ \mathsf{redge} \\
   &\mathtt{FROM}\ V_l.
\end{align*}
A more efficient variant of this algorithm that uses window functions with sorting and partitioning allows us to avoid joins used, e.g., in the query $V_l$~\cite{DBLP:conf/dexa/ZhouWZ06}. We will refer to this algorithm in Section~\ref{sec:eval} as a \emph{standard SQL} algorithm. In contrast to the imperative algorithm described in the proof of Lemma~\ref{lemma:coal}, this algorithm can be implemented using standard SQL operators.

  In addition, for each intensional predicate $P$ of $\Pi$, we create a view $V_P$ defined by an SQL query that reflects the definitions of $P$  in $\Pi$. For example, we set
\begin{align*}
  V_{\mathsf{Y}} = \ &\mathsf{SELECT}\ V_{\mathsf{TempAbove24}}^*.\mathsf{station\_id} \
  \mathsf{AS} \  \mathsf{station\_id} , \\
  &\qquad V_{\mathsf{TempAbove24}}^*.\mathsf{ledge}+24h \ \mathsf{AS} \ \mathsf{ledge},
   \ V_{\mathsf{TempAbove24}}^*.\mathsf{redge} \ \mathsf{AS} \ \mathsf{redge}, \\
  & \mathsf{FROM} \ V_{\mathsf{TempAbove24}}^* \ \mathsf{WHERE}\  V_{\mathsf{TempAbove24}}^*.\mathsf{redge} - V_{\mathsf{TempAbove24}}^*.\mathsf{ledge} \geq 24h.
\end{align*}
This query implements the $\iota +^c \range$ operation for $\range = [0, 24h]$, and  the $\mathsf{WHERE}$ clause checks whether $\range \sqsubseteq \iota$ holds, where $\iota = ( V_{\mathsf{TempAbove24}}^*.\mathsf{ledge},V_{\mathsf{TempAbove24}}^*.\mathsf{redge}]$. We then set  $V_{\mathsf{Y}}^* = \mathsf{coalesce}(V_{\mathsf{Y}})$ and note that the query
\begin{align} \label{eq:select-for-y}
  \mathsf{SELECT}\
  \mathsf{station\_id} , \
  \mathsf{ledge},  \mathsf{redge} \
    \mathsf{FROM} \ V_{\mathsf{Y}}^*,
\end{align}
when evaluated over the tables $T_{\mathsf{TempAbove24}}$, $T_{\mathsf{TempAbove41}}$  and  $T_{\mathsf{LocatedInCounty}}$, would produce the answers to the query $(\Pi, \mathsf{Y}(\mathsf{station\_id}, \mathsf{county})@ x)$ with \emph{maximal} intervals $\iota = (\iota_b, \iota_e]$, where $\iota_b$ corresponds to $\mathsf{ledge}$, and $\iota_e$ to $\mathsf{redge}$.

\begin{figure}
 \centering
\begin{lstlisting}[keywords={foreach,for,each,if,in,return,else,function, SELECT, FROM, WHERE, AND, AS, UNION},keywordstyle=\textbf,basicstyle=\ttfamily\small,frame=tb]
function ans($q(\avec{v}\underline{,x}) =Q(\avec{\tau})\underline{@x}$,  $\Pi$, $\mathcal{M}$, $\D$):
  $\mathcal{V}$ = views($\Pi$, $\mathcal{M}$, $Q$)
  $ans$ = SELECT projects($\avec{v},\avec{\tau}$), $\underline{\avec{\iota} \text{ AS } x}$ FROM $V_{Q}$
  return eval($ans \land \mathcal{V}$, $\D$)

function views($\Pi$, $\mathcal{M}$, $Q$):
  $\mathcal{V} = \emptyset$
  for each predicate $P$ defined by $\mathcal{M}$:
    $\mathcal{V}$ = $\mathcal{V}\cup \big\{V_P$=$\underline{coalesce}$(UNION($\{$SELECT projects$(\avec{\tau},\avec{\tau}),$  $\underline{T_1.\avec{\iota} \text{ AS } \avec{\iota} }$ FROM $sql $ AS $ T_1$
                                                    $\mid P(\avec{\tau})\underline{@\avec{\iota}}{\leftarrow}sql\in\mathcal{M}\}$))$\big\}$
  let $\lessdot$ be the dependence relation on the predicates $\text{in}$ $\Pi$
  for each intensional predicate $P$ with $Q\lessdot P$ or $Q=P$:
    $\mathcal{V}$ = $\mathcal{V}$ $ \cup$ $\{ V_P$=$\underline{coalesce}$(UNION($\{\text{view}(\rul,P) \mid \rul \in \Pi_P\}$))$\}$
  return $\mathcal{V}$

function view($\rul, P$):
  if $\rul =  P(\avec{\tau})\leftarrow \boxplus_{\range} P_1(\avec{\tau}_1)$:
    $V_P^{\rul}$ = SELECT projects($\avec{\tau}, \avec{\tau_1}$), $\underline{T_1.\avec{\iota} \mc \range \text{ AS } \avec{\iota}}$
          FROM $V_{P_1}$ AS $T_1$ WHERE join-cond($\avec{\tau_1})$  $\underline{\textbf{AND } \range \sqsubseteq {T_1}.\avec{\iota}}$
  else if $\rul =  P(\avec{\tau})\leftarrow \boxminus_{\range} P_1(\avec{\tau}_1)$:
    $V_P^{\rul}$ = SELECT projects($\avec{\tau}, \avec{\tau_1}$), $\underline{T_1.\avec{\iota} \pc \range \text{ AS }\avec{\iota}}$
          FROM $V_{P_1}$ AS $T_1$ WHERE join-cond($\avec{\tau_1})$  $\underline{\textbf{AND } \range \sqsubseteq {T_1}.\avec{\iota}}$
  else if $\rul = P(\avec{\tau}) \leftarrow P_1(\avec{\tau}_1) \Si_\range P_2(\avec{\tau}_2)$:
    $V_P^{\rul}$ = SELECT projects($\avec{\tau}$), $\underline{((T_1.\avec{\iota}^c \cap T_2.\avec{\iota}) \po \range) \cap T_1.\avec{\iota}^c \text{ AS } \avec{\iota}}$
         FROM $V_{P_1}$ AS $T_1$, $V_{P_2}$ AS $T_2$
         WHERE join-cond($\avec{\tau_1}, \avec{\tau_2})$) $\underline{\textbf{AND } (({T_1}.\avec{\iota}^c \cap {T_2}.\avec{\iota}) \po \range) \cap {T_1}.\avec{\iota}^c \neq \emptyset}$
  else if $\rul = P(\avec{\tau}) \leftarrow P_1(\avec{\tau}_1) \U_\range P_2(\avec{\tau}_2)$:
    $V_P^{\rul}$ = SELECT projects$(\avec{\tau},\avec{\tau_1},\avec{\tau_2})$, $\underline{((T_1.\avec{\iota}^c \cap T_2.\avec{\iota}) \mo \range) \cap T_1.\avec{\iota}^c \text{ AS } \avec{\iota}}$
         FROM $V_{P_1}$ AS $T_1$, $V_{P_2}$ AS $T_2$
         WHERE join-cond($\avec{\tau_1}, \avec{\tau_2}$) $\underline{\textbf{AND } ((T_1.\avec{\iota}^c \cap T_2.\avec{\iota}) \mo \range) \cap T_1.\avec{\iota}^c \neq \emptyset}$
  else if $\rul = P(\avec{\tau}) \leftarrow P_1(\avec{\tau_1}), \ldots, P_n(\avec{\tau_n})$:
    $V_P^{\rul}$ = SELECT projects($\avec{\tau}, \avec{\tau_1}, \ldots, \avec{\tau_n}$), $\underline{T_1.\avec{\iota}\cap \ldots \cap T_n.\avec{\iota} \text{ AS }\avec{\iota}}$
         FROM $V_{P_1}$ AS $T_1$, $\cdots$, $V_{P_n}$ AS $T_n$
         WHERE join-cond($\avec{\tau_1}, \ldots, \avec{\tau_n}$) $\underline{\textbf{AND } {T_1}.\avec{\iota}\cap \ldots T_n.\avec{\iota} \neq \emptyset}$
  return $V_P^{\rul}$

function projects($\avec{\tau}, \avec{\tau_1}, \ldots, \avec{\tau_n}$):
  columns = {}
  for each $v_i \in \avec{\tau} = v_1, ..., v_m$:
    let $k,j$ be a pair of integers such that $\avec{\tau_k}[j] = v_i$
    columns.add($T_k$.$attr_j$ AS $attr_i$)
  return columns

function join-cond($\avec{\tau_1}, \ldots, \avec{\tau_n}$):
  cond = true
  for each pair of different positions $\avec{\tau_l}[i]$ and $\avec{\tau_r}[j]$ such that $\avec{\tau_l}[i] = \avec{\tau_r}[j]$
    cond = cond AND $(T_{l}.attr_{i} = T_r.attr_{j})$
  return cond
\end{lstlisting}
 \caption{The algorithm for evaluating $\nrdMTL$ queries in SQL.}
 \label{fig:sql-nrdmtl}
\end{figure}

We now explain how to construct queries for the concepts whose  definitions involve $\land$ using the example of $\mathsf{HeatAffectedCounty}$:
\begin{align*}
V_{\mathtt{HeatAffectedCounty}} = \ &
\mathtt{SELECT} \  V_{\mathtt{LocatedInCounty}}^*.\mathtt{county} \ \mathtt{AS} \  \mathtt{county}, \\
  &\qquad \mathtt{MX}(V_{\mathtt{LocatedInCounty}}^*.\mathtt{ledge},V_{\mathtt{ExcessiveHeat}}^*.\mathtt{ledge}) \ \mathtt{AS} \ \mathtt{ledge},\\
  &\qquad \mathtt{MN}(V_{\mathtt{LocatedInCounty}}^*.\mathtt{redge},V_{\mathtt{ExcessiveHeat}}^*.\mathtt{redge})\ \mathtt{AS} \ \mathtt{redge}\\
& \mathtt{FROM} \ V_{\mathtt{LocatedInCounty}}^*, V_{\mathtt{ExcessiveHeat}}^*
\end{align*}
\begin{align*}
                                  & \mathtt{WHERE}\ \mathtt{MX}(V_{\mathtt{LocatedInCounty}}^*.\mathtt{ledge},V_{\mathtt{ExcessiveHeat}}^*.\mathtt{ledge}) < \\
                                  & \qquad \mathtt{MN}(V_{\mathtt{LocatedInCounty}}^*.\mathtt{redge},V_{\mathtt{ExcessiveHeat}}^*.\mathtt{redge}) \\
&\qquad \mathtt{AND} \  V_{\mathtt{LocatedInCounty}}^*.\mathtt{county} = V_{\mathtt{ExcessiveHeat}}^*.\mathtt{county},
\end{align*}
where $\mathtt{MN}$ ($\mathtt{MX}$) is the function that returns the earliest (latest) of any two given date/time values (it can be implemented in SQL as a user-defined function, or using the $\mathtt{CASE}$ operator). Finally, we use a query similar to~\eqref{eq:select-for-y} over $V^*_{\mathsf{HeatAffectedCounty}}$ to produce the answers to $(\Pi, \avec{q}(\mathsf{county}, x))$.

We are mostly interested in the scenario where the tables $T_P$ are not available immediately, but extracted from raw timestamped data tables $R$ by means of mappings. In this case, we use views $V_P$ instead of $T_P$ defined over $R$. For example, if the raw data is stored in the table $\mathsf{Weather}$, we define the view:
\begin{align*}
V_{\mathsf{TempAbove24}} =~ & \ \mathtt{SELECT\ sid,\ \mathtt{ledge},\ \mathtt{redge}}\\ 
&\hspace*{0.1cm} \mathtt{FROM\ } (
\mathtt{SELECT\ station\_id\ AS\ } \mathtt{sid}\mathtt{,}\\
& \hspace*{1.8cm} \mathtt{LAG(date\_time, 1) \ OVER\ (w) \ AS \ } \mathtt{ledge}\mathtt{,}\\
& \hspace*{1.8cm} \mathtt{date\_time \ AS \ } \mathtt{redge}\\
                            & \hspace*{1.2cm} \mathtt{FROM \ } \mathsf{Weather}  \\
                            & \hspace*{1.2cm} \mathtt{WINDOW\ w\ AS \ (PARTITION\ BY \ station\_id\ ORDER \ BY \ date\_time) }\ \\
                            & \hspace*{0.1cm} \mathtt{) }\ \mathtt{tmp}\\
                            &\mathtt{WHERE \ air\_temp\_set\_1 \ >= \ 24}.
\end{align*}

Our general rewriting algorithm is outlined in Fig.~\ref{fig:sql-nrdmtl},
where the function $\mathtt{ans}$ produces an SQL query that computes
the certain answers to $(\Pi, Q(\avec{\tau})@x)$ (with maximal
intervals) by evaluating the query over the input database
$\D$.
%
The algorithm is a variation of the standard translation of non-recursive Datalog to
relational algebra---see, e.g., the work by \citeA{Ullman88-dbkb-v1}---extended with the operations on temporal intervals described above (they are underlined in Fig.~\ref{fig:sql-nrdmtl}).

It is to be noted that the `views'
introduced by the algorithm do not require modifying the underlying
database. They can be implemented in different ways: for example, by
using subqueries, common table expressions (CTEs), or temporary
tables. For the experiments in Section~\ref{sec:eval}, we use the
last approach, where temporary tables are generated on the fly
and exist only within a transaction.

\section{Use Cases}
\label{sec:uc}

We test the feasibility of OBDA with $\nrdMTL$ by querying Siemens turbine log data and Meso\-West weather data. In this section, we briefly describe these use cases;  detailed results of our experiments will presented in Section~\ref{sec:eval}.

\subsection{Siemens}
Siemens service centres store aggregated turbine sensor data  in tables such as \TB. The data comes with (not necessarily regular) timestamps $t_1,t_2,\dots$, and it is deemed that the values remain constant in every interval $[t_i,t_{i+1})$.
Using a set of mappings, we extract from these tables a data instance containing ground facts such as
\vspace*{-1mm}
\begin{align*}
&  \mathsf{ActivePowerAbove1.5(tb0)@[12{:}20{:}48,12{:}20{:}49)},\\
&  \mathsf{ActivePowerAbove1.5(tb0)@[12{:}20{:}49,12{:}20{:}52)},\\
&  \mathsf{RotorSpeedAbove1500(tb0)@[12{:}20{:}48,12{:}20{:}49)},\\
& \mathsf{MainFlameBelow0.1(tb0)@[12{:}20{:}48,12{:}20{:}52)}.
\end{align*}
For example, the first two of them are obtained from the table \TB{}
using the following SQL mapping $\mathcal{M}$:
\begin{align*}
& \mathsf{ActivePowerAbove1.5(tbid)@[ \mathtt{ledge},\mathtt{redge})} \leftarrow\\
& \hspace*{0.3cm} \mathtt{SELECT\ tbid,\ \mathtt{ledge},\ \mathtt{redge}\ 
FROM\ } (\\
& \hspace*{0.5cm} \mathtt{SELECT\ turbineId\ AS\ } \mathtt{tbid}\mathtt{,}\\
& \hspace*{0.8cm} \mathtt{LAG(dateTime, 1) \ OVER\ (w) \ AS \ } \mathtt{ledge}\mathtt{,}\\
& \hspace*{0.8cm} \mathtt{LAG(activePower, 1) \ OVER\ (w) \ AS \  lag\_activePower\mathtt{,}}\\
& \hspace*{0.8cm} \mathtt{dateTime \ AS \ } \mathtt{redge}\\
& \hspace*{0.5cm} \mathtt{FROM \ } \TB  \\
 & \hspace*{0.5cm} \mathtt{WINDOW\ w\ AS \ (PARTITION\ BY \ turbineId\ ORDER \ BY \ dateTime) }\ \\
& \hspace*{0.5cm} \mathtt{) }\ \mathtt{tmp} \ \mathtt{WHERE \ lag\_activePower \ > \ 1.5}
\end{align*}
In terms of the basic predicates above, we define more complex ones
that are used in queries posed by the Siemens engineers:
%
\begin{align*}
&\mathsf{NormalRestart}(v) \leftarrow{} \mathsf{NormalStart}(v) \land \diamondminus_{(0,1h]}\mathsf{NormalStop}(v),\\
& \mathsf{NormalStop}(v) \leftarrow
 \mathsf{CoastDown1500to200}(v) \land{}
\diamondminus_{(0,9m]}\bigl[\mathsf{CoastDown6600to1500}(v) \land{}\\
& \hspace{50mm} \diamondminus_{(0,2m]} \bigl(\mathsf{MainFlameOff}(v) \land {}
\diamondminus_{(0,2m]} \mathsf{ActivePowerOff}(v) \bigr)\bigr],\\
&\mathsf{MainFlameOff}(v) \leftarrow \boxminus_{[0s,10s]} \mathsf{MainFlameBelow0.1}(v),\\
&\mathsf{ActivePowerOff}(v) \leftarrow \boxminus_{[0s,10s]} \mathsf{MainPowerBelow0.15}(v),\\
&\mathsf{CoastDown6600to1500}(v) \leftarrow{} 
\boxminus_{[0s,30s]} \mathsf{RotorSpeedBelow1500}(v) \land {} \\
& \hspace{75mm} \diamondminus_{(0, 2m]} \boxminus_{(0,30s]} \mathsf{RotorSpeedAbove6600}(v) ,\\
&\mathsf{CoastDown1500to200}(v) \leftarrow{} 
\boxminus_{[0s,30s]} \mathsf{RotorSpeedBelow200}(v)  \land{} \\
& \hspace{75mm} \diamondminus_{(0, 9m]} \boxminus_{(0,30s]} \mathsf{RotorSpeedAbove1500}(v),\\
& \mathsf{NormalStart}(v) \leftarrow \mathsf{STCtoRUCReached}(v) \land{}  \diamondminus_{(0,30s]} \bigl[\mathsf{RampChange1\text{-}2Reached}(v) \land{}\\
& \hspace{22mm} \diamondminus_{(0,5m]} \bigl(\mathsf{PurgingIsOver}(v) \land{} \diamondminus_{(0,11m]} \bigl(\mathsf{PurgeAndIgnitionSpeedReached}(v) \land{}\\
& \hspace{85mm} \diamondminus_{(0,15s]} \mathsf{FromStandStillTo180}(v) \bigr)\bigr)\bigr],\\
&\mathsf{STCtoRUCReached}(v) \leftarrow \boxminus_{(0,30s]} \mathsf{RotorSpeedAbove4800}(v) \land{}\\
& \hspace{75mm} \diamondminus_{(0,2m]} \boxminus_{(0,30s]} \mathsf{RotorSpeedBelow4400(v)},\\
&\mathsf{RampChange1\text{-}2Reached}(v) \leftarrow  \boxminus_{(0s,30s]} \mathsf{RotorSpeedAbove4400}(v) \land{}\\
&\hspace{73mm} \diamondminus_{(0,6.5m]} \boxminus_{(0,30s]} \mathsf{RotorSpeedBelow1500}(v),\\
%
&\mathsf{PurgingIsOver}(v) \leftarrow{} \boxminus_{[0s,10s]} \mathsf{MainFlameOn}(v) \land {} \\
& \hspace{1mm} \diamondminus_{(0, 10m]}\bigr[ \boxminus_{(0,30s]} \mathsf{RotorSpeedAbove1260}(v) \land{} \diamondminus_{(0,2m]} \boxminus_{(0,1m]} \mathsf{RotorSpeedBelow1000}(v) \bigl],\\
&\mathsf{PurgeAndIgnitionSpeedReached}(v) \leftarrow{} \boxminus_{[0s,30s]} \mathsf{RotorSpeedAbove1260}(v)  \land{} \\
& \hspace{77mm} \diamondminus_{(0, 2m]} \boxminus_{(0,30s]} \mathsf{RotorSpeedBelow200}(v),\\
%
&\mathsf{FromStandStillTo180}(v) \leftarrow{} \boxminus_{[0s,1m]} \mathsf{RotorSpeedAbove180}(v)  \land{} \\
& \hspace{76mm} \diamondminus_{(0, 1.5m]} \boxminus_{(0,1m]} \mathsf{RotorSpeedBelow60}(v).
\end{align*}
\subsection{MesoWest}
The
Meso\-West (\url{http://mesowest.utah.edu/}) project makes
publicly available historical records of the weather stations across
the US showing such parameters of meteorological conditions as
temperature, wind speed and direction, amount of precipitation,
etc. Each station outputs its measurements with some periodicity, with
the output at time $t_{i+1}$ containing the accumulative (e.g., for
precipitation) or averaged (e.g., for wind speed) value over the
interval $(t_i,t_{i+1}]$.
The data comes in a table $\mathsf{Weather}$, which looks as follows:
\begin{center}
	\footnotesize
\addtolength{\tabcolsep}{-3pt}
\begin{tabular}{|c|c|c|c|c|c|c|}
\hline
stationId & dateTime & airTemp & windSpeed & windDir & hourPrecip &  \dots\\\hline
 & & \dots &  & & & \\
KBVY & 2013-02-15;15:14 & 8 & 45 & 10 & 0.05 & \\
KMNI & 2013-02-15;15:21 & 6 & 123 & 240 & 0& \\
KBVY & 2013-02-15;15:24 & 8 & 47 & 10 & 0.08 & \\
KMNI & 2013-02-15;15:31 & 6.7 & 119 & 220 & 0 & \\
 & & \dots & & & & \\\hline
\end{tabular}
\end{center}

One more table, $\mathsf{Metadata}$, provides some atemporal meta information about the stations:
\begin{center}
		\footnotesize
\addtolength{\tabcolsep}{-3pt}
\begin{tabular}{|c|c|c|c|c|c|}
\hline
stationId & county & state & latitude & longitude &   \dots\\\hline
 & & \dots &  & &  \\
KBVY & Essex & Massachusetts & 42.58361 & -70.91639 &  \\
KMNI & Essex & Massachusetts & 33.58333  & -80.21667 &  \\
 & & \dots & & &  \\\hline
\end{tabular}
\end{center}
The monitoring and historical analysis of the weather involves answering  queries such as `find showery counties, where one station observes precipitation at the moment, while another one does not, but observed precipitation 30 minutes ago'\!.

We use SQL mappings over the $\mathsf{Weather}$ table similar to those in the Siemens case to obtain ground atoms such as
\begin{align*}
&   \mathsf{NorthWind(KBVY)@(15{:}14,15{:}24]},\\
&   \mathsf{HurricaneForceWind(KMNI)@(15{:}21,15{:}31]},\\
&   \mathsf{Precipitation(KBVY)@(15{:}14,15{:}24]},\\
&   \mathsf{TempAbove0(KBVY)@(15{:}14,15{:}24]},\\
&   \mathsf{TempAbove0(KMNI)@(15{:}21,15{:}31]}
\end{align*}
%
(according to the standard definition, the hurricane force wind is above 118 km/h). On the other hand, mappings to the $\mathsf{Metadata}$ table provide atoms such as
\begin{align*}
&   \mathsf{LocatedInCounty(KBVY,Essex)@(-\infty, \infty)},\\
&   \mathsf{LocatedInState(KBVY,Massachusetts)@(-\infty, \infty)}.
\end{align*}
Our ontology contains definitions of various meteorological terms:
%
\begin{align*}
& \mathsf{ShoweryCounty}(v) \leftarrow \mathsf{LocatedInCounty}(u_1, v) \land{} 
\mathsf{LocatedInCounty}(u_2, v) \land{}\\
& \hspace*{3.5cm} \mathsf{Precipitation}(u_1) \land{}  \mathsf{NoPrecipitation}(u_2) \land {}
\diamondminus_{(0,30m]} \mathsf{Precipitation}(u_2),\\
&   \boxminus_{[0,1h]}  \mathsf{Hurricane}(v) \leftarrow  \boxminus_{[0,1h]} \mathsf{HurricaneForceWind}(v),\\
&   \mathsf{HurricaneAffectedState}(v) \leftarrow \mathsf{LocatedInState}(u,v) \land{}  \mathsf{Hurricane}(u),
\end{align*}
\begin{align*}
&     \boxminus_{[0,24h]}  \mathsf{ExcessiveHeat}(v) \leftarrow  \boxminus_{[0,24h]} \mathsf{TempAbove24}(v) \land{}      \diamondminus_{[0,24h]} \mathsf{TempAbove41}(v),\\
&   \mathsf{HeatAffectedCounty}(v) \leftarrow \mathsf{LocatedInCounty}(u,v) \land   \mathsf{ExcessiveHeat}(u),\\
&  \mathsf{CyclonePatternState}(v) \leftarrow \mathsf{LocatedInState}(u_1, v) \land{} \mathsf{LocatedInState}(u_2, v) \land{} \\ & \hspace*{3.65cm} \mathsf{LocatedInState}(u_3, v) \land {}
 \mathsf{LocatedInState}(u_4, v) \land \mathsf{EastWind}(u_1) \land {} \\
& \hspace*{5.9cm} \mathsf{NorthWind}(u_2) \land \mathsf{WestWind}(u_3) \land{} \mathsf{SouthWind}(u_4).
\end{align*}

\section{Experiments}
\label{sec:eval}

To evaluate the performance of the SQL queries produced by the $\nrdMTL$ rewriting algorithm outlined in Section~\ref{subsec:imp-in-sql}, we developed two benchmarks for our use
cases.
We ran the experiments on an HP Proliant server with 2 Intel Xeon
X5690 Processors (with 12 logical cores at 3.47GHz each), 106GB of RAM
and five 1TB 15K RPM HD. We used both PostgreSQL 9.6 and the SQL
interface~\cite{DBLP:conf/sigmod/ArmbrustXLHLBMK15} of Apache Spark
2.1.0. Apache Spark is a cluster-computing framework that provides
distributed task dispatching, scheduling and data parallelisation. For
each of these two systems, we provided two different implementations,
imperative and standard SQL, which diverge in the computation of
maximal intervals; see
Section~\ref{implementing}.
%
We ran all the queries  with a timeout of 30 minutes.

%


\subsection{Siemens}
Siemens provided us with a sample of data for one running
turbine, which we denote by $\mathsf{tb0}$, over 4 days in the form of the
table $\TB$.  The data table was rather sparse, containing a lot of  nulls, because different sensors recorded data at
different frequencies. For example, $\mathsf{ActivePower}$ arrived  most frequently with average periodicity of $7$ seconds, whereas the values for the field $\mathsf{MainFlame}$ arrived most rarely, every $1$ minute on average.
We replicated this sample to imitate the data for one turbine over 10
different periods ranging from 32 to 320 months. The statistics of the
data sets are given in Tables~\ref{tab:data-stat-siemens} and~\ref{tab:data-stat-siemens-extra}.
%
We evaluated four queries
$\mathsf{ActivePowerTrip}(\mathsf{tb0})@x$,
$\mathsf{NormalStart}(\mathsf{tb0})@x$,
$\mathsf{NormalStop}(\mathsf{tb0})@x$ and
$\mathsf{NormalRestart}(\mathsf{tb0})@x$. The statistics of returned answers  is given in Table~\ref{tab:answer-stat-siemens}.

\begin{table}[h]
		\small
        \centering
        \addtolength{\tabcolsep}{-1pt}
\begin{tabular}{|c|c|c|c|c|c|c|c|c|c|c|}\hline
\# of months & 32 & 64  & 96 & 128 & 159 & 191 & 223 & 255 & 287 & 320 \\\hline
 \# of rows (approx.) &13 M&26 M&39 M&52 M&65 M&77 M&90 M&103 M&116 M&129 M\\\hline
size (GB) in CSV &0.57&1.2&1.7&2.3&2.9&3.4&4.0&4.5&5.1&5.7\\\hline
\end{tabular}
\caption{Siemens data for one turbine.}
\label{tab:data-stat-siemens}
\end{table}


The execution times for the Siemens use case are given in
Fig.~\ref{fig:siemens}.
\begin{figure}[ht]
\centering
\includegraphics[width=1\textwidth]{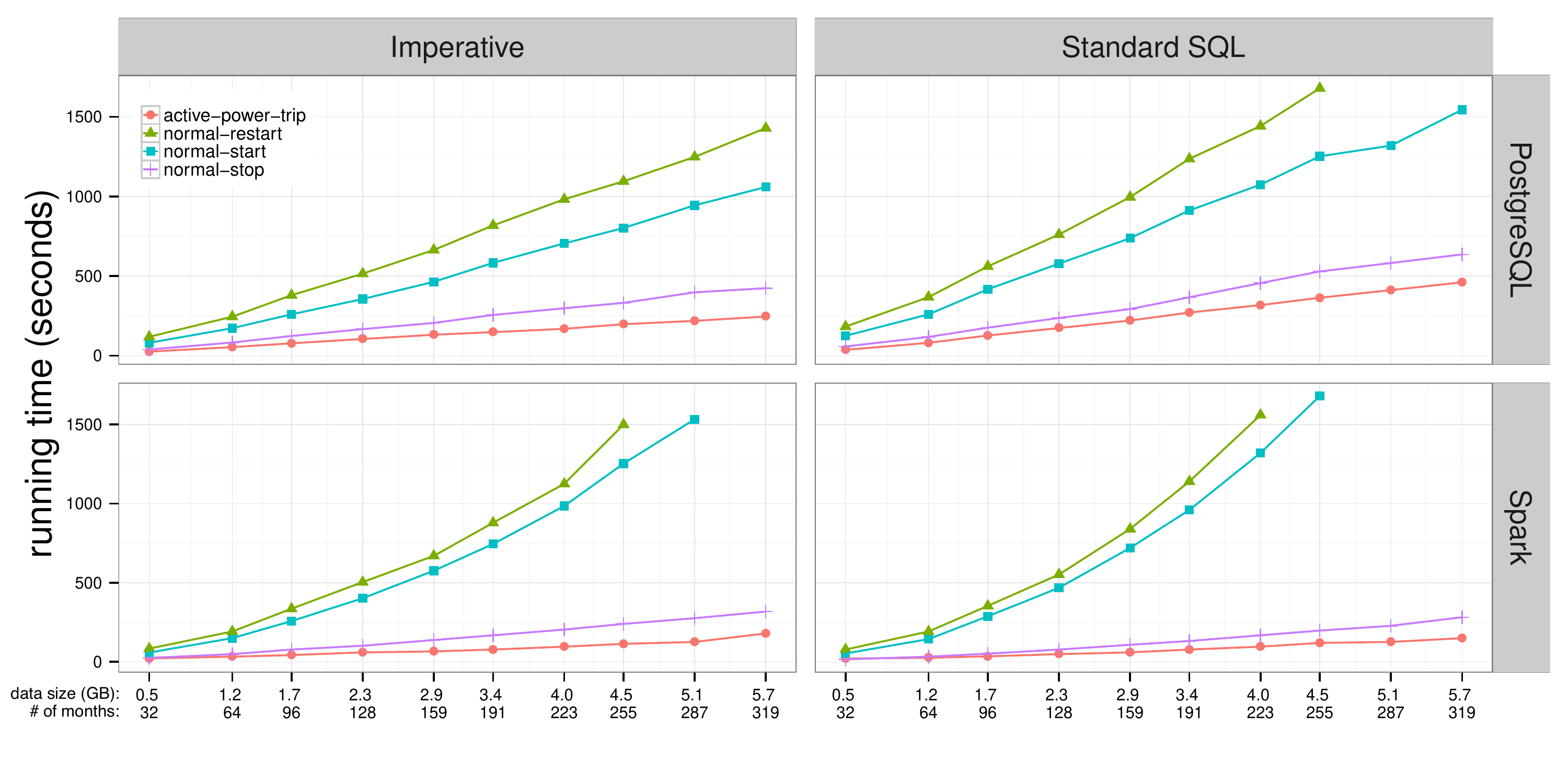}
\caption{Experiment results for the Siemens use case.}
\label{fig:siemens}
\end{figure}
Although Apache Spark was designed to perform
efficient parallel computations, it failed to take advantage of this
feature due to the fact that the Siemens data could not be partitioned by
mapping each part to a separate core.
%
PostgreSQL 9.6 also supports parallel query execution in some
cases. However, as many operators (e.g., scans of temporary tables) in our queries are classified either
`parallel unsafe' or `parallel restricted' in the parallel safety
documentation \cite{postgresql},
the query planner failed to produce any parallel execution strategy in
our case.
%
%
The reason why PostgreSQL outperformed Apache Spark is that the latter does not provide a convenient way to define proper indexes over temporary tables, which leads to quadratically growing running times.
%
On the other hand, PostgreSQL shows
linear growth
 in the size of
 data
(confirming theoretical results since we deal with a single turbine).

Note that the normal restart (start) query timeouts on the data for
more than 18 (respectively, 21) years, which is more than enough for
the monitoring and diagnostics tasks at Siemens, where the two most common
application scenarios for sensor data analytics are daily monitoring (that is, analytics of high-frequency data of the previous 24 hours)  and
fleet-level analytics of key-performance indicators over one year. In
both cases, the computation time of the results is far less a crucial
cost factor than the lead-time for data preparation.

\subsection{MesoWest}
In contrast to the Siemens case, the weather
tables contain very few nulls. Normally, the data values arrive with periodicity from 1 to 20 minutes. We tested the performance of our
algorithm by increasing $(i)$ the temporal span (with some necessary increase of the spatial spread) and $(ii)$ the
geographical spread of data.
For $(i)$, we took the New
York state data for the 10 continuous periods between 2005 and
2014; see Tables~\ref{tab:data-stat-ny} and~\ref{tab:data-stat-ny-extra}. As each year around 70 new weather stations were added, our 10 data samples increase more than linearly in size.
\begin{table}
		\small
        \centering
\begin{tabular}{|c|c|c|c|c|c|c|c|c|c|c|}\hline
\# of years & 1 & 2  & 3 & 4 & 5 & 6 & 7 & 8 & 9 & 10 \\
\# of stations & 229 & 306 & 370 & 441 & 484 & 542 & 595 & 643 & 807 & 874\\\hline
\# of rows (approx.) &4 M &11 M&19 M&27 M&36 M&49 M&63 M&79 M&99 M&124 M\\\hline
size (GB) in CSV &0.2&0.6&1.1&1.6&2.1&2.9&3.8&4.8&5.9&7.4\\\hline
\end{tabular}
\caption{NY weather stations from 2005 to 2014.}
\label{tab:data-stat-ny}
\end{table}
\begin{table}[h]
	\small
        \centering
        \addtolength{\tabcolsep}{-1.5pt}
\begin{tabular}{|c|c|c|c|c|c|c|c|c|c|c|}\hline
\multirow{2}{*}{states}  & DE, &  +NY &  +MD &  +NJ,&  +MA, &  +LA, &  +ME, &  +NH, &  +MS,SC, &  +KY, \\
&GA&&&RI&CT&VT&WV&NC&ND&SD\\\hline
\# of states & 2 & 3 & 4 & 6 & 8 & 10 & 12 & 14 & 17 & 19 \\
\# of stations & 408 & 659 & 1120 & 1476 & 1875 & 2305 & 2669 & 3019 & 3508 & 4037\\\hline
\# of rows (approx.) &17 M&32 M&41 M&52 M&67 M&81 M&93 M&106 M&121 M&141 M\\\hline
size (GB) in CSV&0.9&1.9&2.5&3.1&4.0&4.8&5.5&6.4&7.2&8.3\\\hline
\end{tabular}
\caption{Weather data for 1--19 states in 2012.}
\label{tab:data-stat-w2012}
\end{table}
For $(ii)$, we fixed the time
period of one year (2012) and linearly increased the data from 1 to 19
states (NY, NJ, MD, DE, GA, RI, MA, CT, LA, VT, ME, WV, NH, NC, MS, SC, ND,
KY, SD); see Table~\ref{tab:data-stat-w2012} and~\ref{tab:data-stat-w2012-extra}.
In both cases, 
 we executed four $\nrdMTL$ queries
$\mathsf{ShoweryCounty}(v)@x$,
$\mathsf{Hurricane}\mathsf{AffectedState}(\text{NY})@x$,
$\mathsf{Heat}\mathsf{AffectedCounty}(v)@x$,
$\mathsf{Cyclone}\mathsf{PatternState}(\text{NY})@x$.
The statistics of the returned answers is shown in Tables~\ref{tab:answer-stat-ny} and~\ref{tab:answer-stat-w2012}.

\begin{figure}
	\centering
	\includegraphics[width =\linewidth]{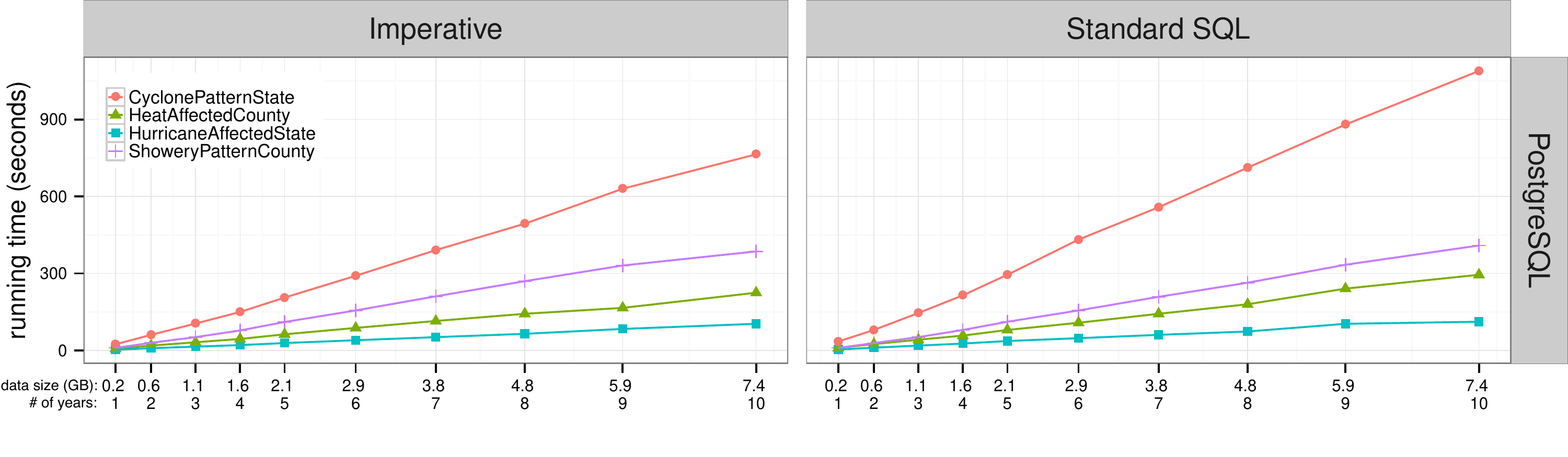}
	
	\includegraphics[width =\linewidth]{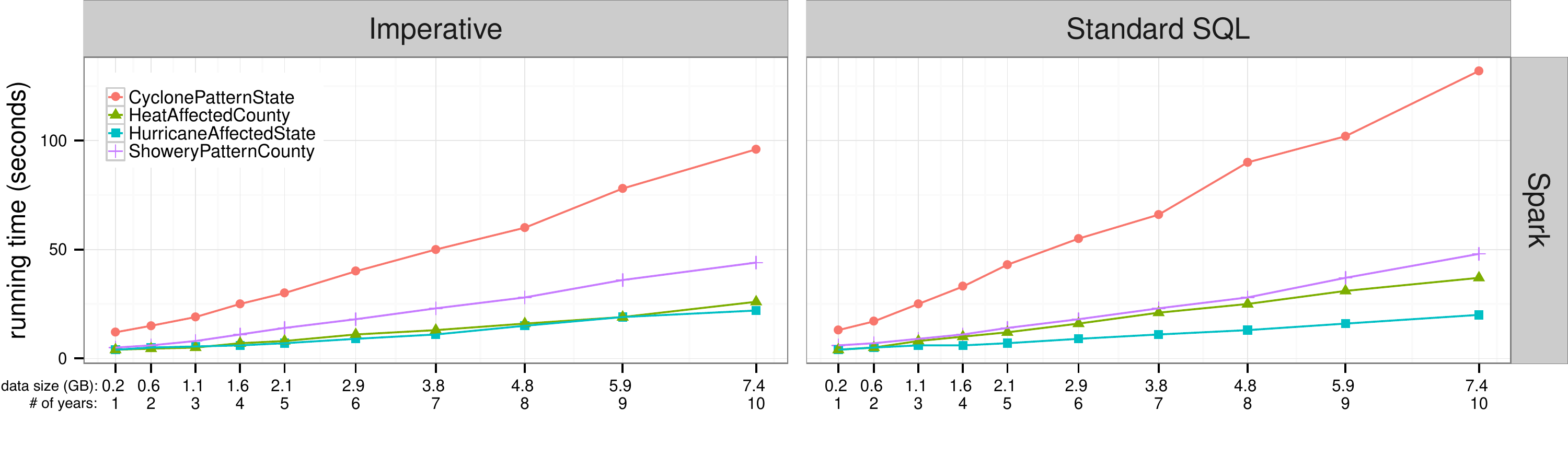}
\caption{Experiments over New York data of 2005--2014 with PostgreSQL and Spark}
\label{fig:weather-long}
\end{figure}

\begin{figure}
	\centering
\includegraphics[width =\linewidth]{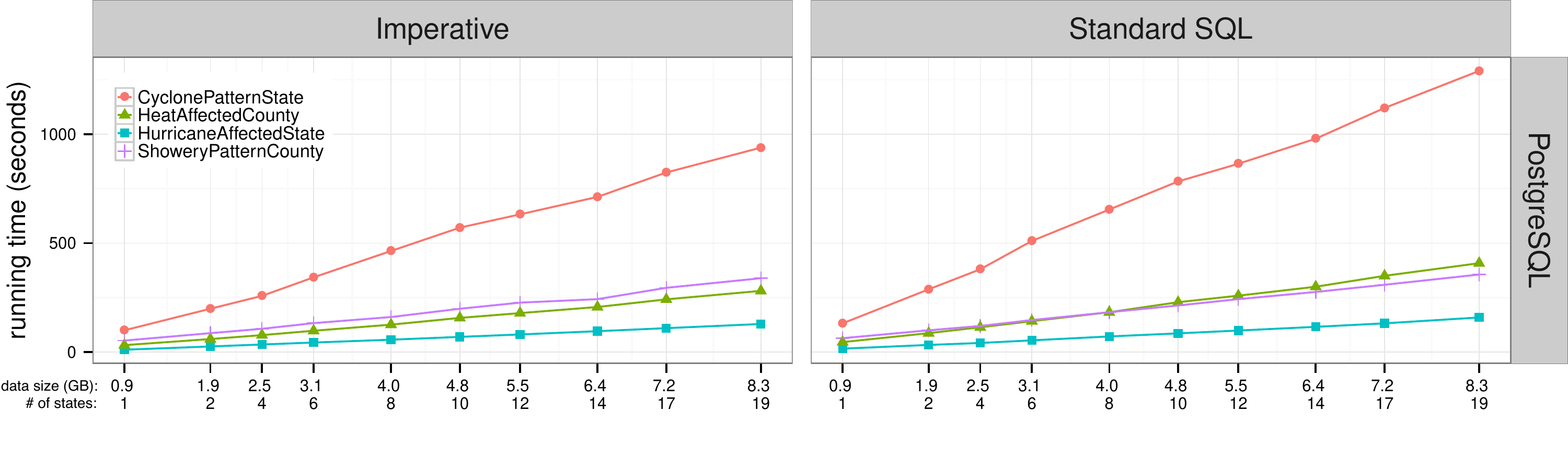}

	\includegraphics[width =\linewidth]{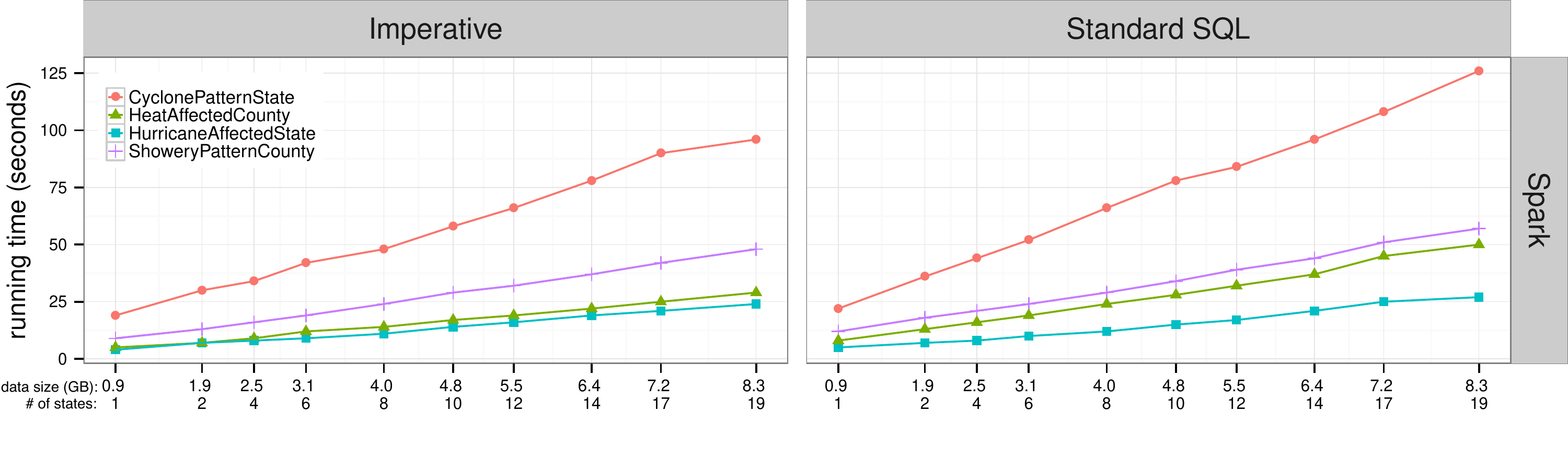}

\caption{Experiments over 1 year data from 1--19 states with PostgreSQL and Spark.}
\label{fig:weather-wide}
\end{figure}

The execution times are shown in Figures~\ref{fig:weather-long} and~\ref{fig:weather-wide}. All the four
queries can be answered within the time limit. The most expensive
one is the cyclone pattern state query because its definition
includes a join of four atoms for winds in four directions, each with
a large volume of instances.
All the graphs in Figures~\ref{fig:weather-wide} and~\ref{fig:weather-long} exhibit linear
behaviour with respect to the size of data.
%
The nearly tenfold better performance of Spark over PostgreSQL can be explained by the fact that, unlike the data in the Siemens case, the MesoWest data is highly
parallelisable. Since it was collected from hundreds of
different weather stations, it can be partitioned by station id,
state, county, etc.\ to perfectly fit the MapReduce programming model
extended with resilient distributed datasets
(RDDs)~\cite{DBLP:journals/cacm/ZahariaXWDADMRV16}. In this case,
Apache Spark is able to take advantage of the multi-core and large
memory hardware infrastructure,
to compute mappings and coalescing in parallel,
making it 10 times faster than
PostgreSQL; see Figures~\ref{fig:weather-long}
and~\ref{fig:weather-wide}.

Overall, the results of the experiments look very encouraging: our
$\nrdMTL$ query rewriting algorithm produces SQL queries that are
executable by a standard database engine PostgreSQL in acceptable
time, and by a cluster-computing framework Apache Spark in better than
acceptable time (in case data can be properly partitioned) over large
sets of real-world temporal data of up to 8.3GB in CSV format. The
relatively challenging queries such as $\mathsf{NormalRestart}$ and
$\mathsf{CyclonePatternState}$ require a large number of temporal
joins, which turn out to be rather expensive.

\section{Conclusions and Future Work}

To facilitate access to sensor temporal data with the aim of monitoring and diagnostics, we suggested the ontology
language $\dMTL$, an extension of datalog with the Horn fragment of
the metric temporal logic \MTL{} (under the continuous semantics). We showed that answering $\dMTL$
queries is {\sc ExpSpace}-complete for combined complexity, but
becomes undecidable if the diamond operators are allowed in the head
of rules. We also proved that answering nonrecursive $\dMTL$ queries is {\sc PSpace}-complete for combined complexity and in AC$^0$ for data complexity. We tested feasibility and efficiency of OBDA with $\nrdMTL$ on two real-world use cases by querying Siemens turbine data and MesoWest weather data. Namely, we designed $\nrdMTL$ ontologies defining typical concepts used by Siemens engineers and various meteorological terms, developed and implemented an algorithm rewriting $\nrdMTL$ queries into SQL queries, and then executed the SQL queries obtained by this algorithm from our ontologies over the Siemens and MesoWest data, showing their acceptable efficiency and scalability.
(To the best of our knowledge, this is the first work  on practical OBDA with temporal ontologies, and so no other systems with similar functionalities are available for comparison.)

Based on these encouraging results, we plan to include our temporal
OBDA framework into the Ontop platform~\cite{DBLP:conf/semweb/Rodriguez-MuroKZ13,DBLP:conf/semweb/KontchakovRRXZ14,DBLP:journals/semweb/CalvaneseCKKLRR17}; visit \url{http://ontop.inf.unibz.it/} for more information on Ontop. Note also that $\dMTL$ presented here has been recently used to develop an ontology of ballet moves (see Example~\ref{ex:dance}) that underlies a search engine of annotated sequences in ballet videos~\cite{DBLP:conf/esws/RahebMRPI17}. This is a third use case for our framework (and we are aware of a few more emerging use cases), which makes an efficient and user-friendly implementation of the framework a top priority.

We are also working on the streaming data setting, where the challenge
is to continuously evaluate queries over the incoming data.
A rule-based language with window operators for analysing streaming data has been suggested by~\citeA{DBLP:conf/aaai/BeckDEF15}. This language is very expressive as it uses an abstract semantics for window operators (which does not have to guarantee decidability) and allows negation and disjunction in the rules. It would be interesting to identify and adapt a suitable fragment of this language in our temporal OBDA framework.


\subsection*{Acknowledgements}
This work was supported by the UK EPSRC grant EP/M012670 `iTract: Islands of Tractability in Ontology-Based Data Access'\! and by the OBATS project at the Free University of Bozen-Bolzano.

Guohui Xiao is the corresponding author of this article.



\appendix

\section{}

\subsection*{Proof of Theorem~\ref{thm:aczero}}

The formula $\sigma^{\langle m, n \rangle}_{\range, P, P_1, P_2}(x,y)$ is defined as follows:
\begin{align*}
&\exists x_1, y_1, \dots, x_5, y_5 \
\hspace*{-5mm}\bigvee_{\begin{subarray}{c}
    m_1 \in \mathsf{le}(P_1)\\
    n_1 \in \mathsf{ri}(P_1)\\
    \lceil_1 \in \{ [, ( \},\
    \rceil_1 \in \{ ], ) \}
    \end{subarray}}
\hspace*{-3mm}
    \bigg( \varphi_{P_1}^{\lceil_1 m_1, n_1 \rceil_1} (x_1, y_1) \land \hspace*{-3mm}
   \bigvee_{\begin{subarray}{c}
    m_2 \in \mathsf{le}(P_2)\\
    n_2 \in \mathsf{ri}(P_2)\\
    \lceil_2 \in \{ [, ( \},\
    \rceil_2 \in \{ ], ) \}
    \end{subarray}} \bigg( \varphi_{P_2}^{\lceil_2 m_2, n_2 \rceil_2} (x_2, y_2) \land{} \\
&  \hspace*{5mm}  \bigvee_{\begin{subarray}{c}
    m_3 = m_1\\
    n_3 = n_1\\
    \lceil_3 \in \{ [, ( \},\
    \rceil_3 \in \{ ], ) \}
    \end{subarray}} \hspace*{-3mm} \Big(
    (x_3 = x_1) \land (y_3 = y_1) \land \mathsf{is}_{\lceil_3,[} \land \mathsf{is}_{\rceil_3,]} \land{} \\
&  \hspace*{10mm}   \bigvee_{\begin{subarray}{c}
    m_4 \in \mathsf{le}(P_1) \cup \mathsf{le}(P_2)\\
    n_4 \in \mathsf{ri}(P_1) \cup \mathsf{ri}(P_2)\\
    \lceil_4 \in \{ [, ( \},\
    \rceil_4 \in \{ ], ) \}
    \end{subarray}} \Big(
     \mathsf{inter}_{\lceil_2 m_2, n_2 \rceil_2, \lceil_3 m_3, n_3 \rceil_3}^{\lceil_4 m_4, n_4 \rceil_4}(x_4, y_4, x_2, y_2, x_3, y_3) \land{} \\
&   \hspace*{15mm}    \bigvee_{\begin{subarray}{c}
    m_4 \in \mathsf{le}(P)\\
    n_4 \in \mathsf{ri}(P)\\
    \lceil_5 \in \{ [, ( \},\
    \rceil_5 \in \{ ], ) \}
    \end{subarray}} \hspace*{-3mm} \big(
     \mathsf{pluso}_{\range, \lceil_4 m_4, n_4 \rceil_4}^{\lceil_5 m_5, n_5 \rceil_5}(x_5, y_5, x_4, y_4) \land{} \\
&    \hspace*{5cm} \mathsf{inter}_{\lceil_5 m_5, n_5 \rceil_5, \lceil_3 m_3, n_3 \rceil_3}^{\langle m, n \rangle}(x, y, x_5, y_5, x_3, y_3) \big)\Big)\Big)\bigg)\bigg),
\end{align*}
where $\mathsf{pluso}_{\range, \lceil_4 m_4, n_4 \rceil_4}^{\lceil_5 m_5, n_5 \rceil_5}(x_5, y_5, x_4, y_4)$ is an (obvious) formula saying that $\lceil_5 x_5 +m_5, y_5 + n_5 \rceil_5$ is the interval $\lceil_4 x_4+m_4, y_4+n_4 \rceil_4 +^o \range$.

The formula $x = y + c$, for a non-negative $c$, is defined as follows. For $c = \infty$, we take the formula
\begin{align*}
\forall j\, (\mathsf{bit}^{\it in}(x, j, 1) \land \mathsf{bit}^{\it fr}(x, j, 1)),
\end{align*}
whereas for a constant $c= h/2^k$, we can use
\begin{multline*}
\forall j\, \Big(\big(\mathsf{bit}^{\it in}(x, j, 0) \land \mathsf{bit}^{\it in}_{+ h/2^k}(y, j, 0)\big) \lor \big(\mathsf{bit}^{\it in}(x, j, 1) \land \mathsf{bit}^{\it in}_{+ h/2^k}(y, j, 1)\big)\Big) \land \\
\forall j\, \Big(\big(\mathsf{bit}^{\it fr}(x, j, 0) \land \mathsf{bit}^{\it fr}_{+ h/2^k}(y, j, 0)\big) \lor \big(\mathsf{bit}^{\it fr}(x, j, 1) \land \mathsf{bit}^{\it fr}_{+ h/2^k}(y, j, 1)\big)\Big),
\end{multline*}
where predicates $\mathsf{bit}^{\it in}_{+ h/2^k}(y, j, v)$, saying that $v$ is the $j$-th bit of the integer part of $y + h/2^k$, and $\mathsf{bit}^{\it fr}_{+ h/2^k}(y, j, v)$, saying that $v$ is the $j$-th bit of the fractional part of $y + h/2^k$, are defined inductively as follows:
\begin{align*}
    \mathsf{bit}^{\it fr}_{+0/2^k}(y, j, v) = \mathsf{bit}^{\it fr}&(y, j, v), \\
    \mathsf{bit}^{\it fr}_{+(d+1/2^k)}(y, j, v) = \exists u \Bigl( & (u = \ell - k) \land \Bigl( \bigl((j \leq u) \land \mathsf{bit}^{\it fr}_{+d}(y, j, v) \bigr) \lor{} \\
     &\bigl( (v = 0) \land \mathsf{bit}^{\it fr}_{+d}(y, j, 0) \land \exists j' ((u < j' < j) \land \mathsf{bit}^{\it fr}_{+d}(y, j', 0)) \bigr) \lor{} \\
  &  \bigl( (v = 0) \land \mathsf{bit}^{\it fr}_{+d}(y, j, 1) \land \forall j' ((u < j' < j) \to \mathsf{bit}^{\it fr}_{+d}(y, j', 1))\bigr) \lor{}\\
  & \bigl( (v = 1) \land \mathsf{bit}^{\it fr}_{+d}(y, j, 1) \land \exists j' ((u < j' < j) \land \mathsf{bit}^{\it fr}_{+d}(y, j', 0)) \bigr) \lor{} \\
  & \bigl( (v = 1) \land \mathsf{bit}^{\it fr}_{+d}(y, j, 0) \land \forall j' ((u < j' < j) \to \mathsf{bit}^{\it fr}_{+d}(y, j', 1)) \bigr) \Bigr) \Bigr),
\end{align*}
\begin{align*}
  \mathsf{bit}^{\it in}_{+0/2^k}(y, j, v) = \mathsf{bit}^{\it in}&(y, j, v), \\
  \mathsf{bit}^{\it in}_{+(d+1/2^k)}(y, j, v) = \exists u \Bigl( & (u = \ell - k) \land \Bigl(  \\
     &\bigl( (v = 0) \land \mathsf{bit}^{\it in}_{+d}(y, j, 0) \land \exists j' ( ((j' < j) \land \mathsf{bit}^{\it in}_{+d}(y, j', 0)) \lor{} \\
      &\hspace{5cm} ((u < j' < j) \land \mathsf{bit}^{\it fr}_{+d}(y, j', 0))) \bigr) \lor{} \\
  &  \bigl( (v = 0) \land \mathsf{bit}^{\it in}_{+d}(y, j, 1) \land \forall j' (((j' < j) \to \mathsf{bit}^{\it in}_{+d}(y, j', 1)) \land{} \\
  & \hspace{5cm} (u < j' < j) \to \mathsf{bit}^{\it fr}_{+d}(y, j', 1)) \bigr) \lor{}\\
  & \bigl( (v = 1) \land \mathsf{bit}^{\it in}_{+d}(y, j, 0) \land \exists j' ( ((j' < j) \land \mathsf{bit}^{\it in}_{+d}(y, j', 0)) \lor{} \\
      &\hspace{5cm} ((u < j' < j) \land \mathsf{bit}^{\it fr}_{+d}(y, j', 0))) \bigr) \lor{} \\
  &  \bigl( (v = 1) \land \mathsf{bit}^{\it in}_{+d}(y, j, 1) \land \forall j' (((j' < j) \to \mathsf{bit}^{\it in}_{+d}(y, j', 1)) \land{} \\
  & \hspace{5cm}  ((u < j' < j) \to \mathsf{bit}^{\it fr}_{+d}(y, j', 1))) \bigr) \Bigr) \Bigr).
\end{align*}
Here, $u = \ell - k$ can be easily defined using $<$ and $k$.

\subsection*{Proofs of Lemmas~\ref{lemma:proj} and~\ref{lemma:union}}

\begin{lemma*}
If $T$ satisfies TOA, then a projection of $T$ satisfying TOA can be computed in time $O(|T|_o^2 \times |T|_t)$.
\end{lemma*}
\begin{proof}
  We first partition $T$ into a set of purely temporal tables
  $T_{c_1, \dots, c_m}$ and compute the set of all individual tuples
  $(c_1', \dots, c_n')$ that will appear in the projection $T'$. Let
  $(c_1', \dots, c_n')$ be one such tuple, and consider the tables
  $T_{c_1^1, \dots, c_m^1}, \dots,$ $T_{c_1^k, \dots, c_m^k}$ such
  that the projection of each $(c_1^i, \dots, c_m^i)$ is precisely
  $(c_1', \dots, c_n')$. Clearly, we have at most $|T|_o$ such
  tables. It is well-known that, for a pair of ordered tables $S$ and
  $S'$, we can construct an ordered table that contains all the tuples
  $S \cup S'$ in time $|S|+|S'|$. We use this algorithm $k$ times to
  obtain an ordered table containing all the tuples of
  $T_{c_1^1, \dots, c_m^1} \cup \dots \cup$ $T_{c_1^k, \dots, c_m^k}$
  in time $O(k|T|_o)$. We then write the tuples of the form
  $(c_1', \dots, c_n', \langle, t_1, t_2, \rangle)$, where
  $(\langle, t_1, t_2, \rangle)$ is a tuple from the united table,
  into the output table. It can be readily checked that the complete
  output table can be produced in the required time.
\end{proof}

\begin{lemma*} For any pair of tables $T$ and $T'$ satisfying TOA, their union table also satisfying TOA can be computed in time $O((|T|_o^2 + |T'|_o^2) \times (|T|_t + |T'|_t))$.
\end{lemma*}
\begin{proof}
We first partition $T$ and $T'$ into sets of purely temporal tables $T_{c_1, \dots, c_m}$ and, respectively, $T'_{c_1, \dots, c_m}$. While doing this partition, we make sure that the tables $T_{c_1, \dots, c_m}$ are stored sequentially with respect to some order on the tuples $(c_1, \dots, c_m)$ (it can be done in time $|T|_o^2 \times |T|_t$). We do the same for the tables $T_{c_1, \dots, c_m}'$. It remains to go through all the tuples $\langle,  t_1, t_2, \rangle$ and $\lceil,  t_1', t_2', \rceil$ in all the tables $T_{c_1, \dots, c_m}$ and $T'_{c_1, \dots, c_m}$ to produce the union table by an algorithm similar to the one applied to the tables $S$ and $S'$ in the proof of Lemma~\ref{lemma:proj}.
\end{proof}

\subsection*{Experimental Results}

\begin{table*}[h]\centering
	\addtolength{\tabcolsep}{-3pt}
\begin{tabular}{|l|*{10}r|}
			\hline
			\diagbox{queries}{\# of months}  & 32 & 64  & 96 & 128 & 159 & 191 & 223 & 255 & 287 & 320 \\\hline
			ActivePowerTrip&324&648&970&1294&1618&1940&2264&2588&2912&3236\\
			NormalStop&648&1296&1940&2588&3236&3880&4528&5176&5824&6472\\
			NormalStart&162&324&485&647&809&970&1132&1294&1456&1618\\
			NormalRestart&0&0&0&0&0&0&0&0&0&0\\\hline
		\end{tabular}%
\caption{Number of the results returned from the Siemens queries.}
		\label{tab:answer-stat-siemens}
\end{table*}
\begin{table*}[h]\centering
\addtolength{\tabcolsep}{-2pt}

\begin{tabular}{|l|*{10}r|}
			\hline
			\diagbox{queries}{\# of years}  & 1 & 2& 3 & 4 & 5 & 6 & 7 & 8 & 9 & 10 \\\hline
			ShoweryPatternCounty&530&1221&1802&2647&3609&4349&5204&5912&6639&7655\\
			HurricaneAffectedState&2&4&5&5&5&8&9&801&1523&1533\\
			HeatAffectedCounty&0&5&7&14&21&33&39&51&57&59\\
			CyclonePatternState&914&1574&1617&1851&1936&2139&2246&2307&2333&2359\\\hline
		\end{tabular}%

\caption{Number of the results returned from the NY weather stations from 2005 to 2014.}
		\label{tab:answer-stat-ny}
\end{table*}

\newpage

\begin{table*}[h]\centering\footnotesize
\addtolength{\tabcolsep}{-3pt}
\begin{tabular}{|l|*{10}r|}
		\hline
		\diagbox{queries}{\# of states}  & 1 & 2  & 4 & 6 & 8 & 10 & 12 & 14 & 17 & 19 \\\hline
	ShoweryPatternCounty&3769&4481&4928&10349&12709&13681&14470&14933&16381&16883\\
	HurricaneAffectedState&2&784&789&789&790&790&798&811&813&813\\
	HeatAffectedCounty&53&65&81&84&88&98&100&117&142&224\\
	CyclonePatternState&9109&9179&9593&17577&30203&38421&40769&43662&54199&56303\\\hline
	\end{tabular}%

\caption{Number of the results returned from the Weather data for 1--19 states in 2012.}
		\label{tab:answer-stat-w2012}
\end{table*}
\begin{table*}[h] \centering\footnotesize
\addtolength{\tabcolsep}{-4pt}
\begin{tabular}{|cl|c|c|c|c|c|c|c|c|c|c|}\hline
&\# of months & 32 & 64  & 96 & 128 & 159 & 191 & 223 & 255 & 287 & 320 \\\hline
& \# of rows &12,935,&25,871,&38,726,&51,662,&64,597,& 77,453,&90,389,&103,324,&116,260,&129,195,\\
& & 538 & 076 & 765 & 303 & 841 & 530 & 068 & 606 & 144 & 682 \\\hline
CSV&size (GB)&0.57&1.2&1.7&2.3&2.9&3.4&4.0&4.5&5.1&5.7\\\hline
\multirow{2}{*}{PostgreSQL} &raw size (GB) & 0.7 & 1.4 & 2.2 & 2.9 & 3.7 & 4.4 & 5.2 & 5.9 & 6.7 & 7.4\\
&total size (GB) & 1.0 & 2.0 & 3.0 & 4.0 & 5.0 & 6.0 & 7.0 & 8.0 & 9.0 & 10.0\\\hline
Parquet&size (GB)&0.1&0.2&0.3&0.4&0.5&0.6&0.7&0.8&0.9&1.0\\\hline
\end{tabular}
\caption{Siemens data for one turbine.}
        \label{tab:data-stat-siemens-extra}
\end{table*}
\begin{table*}[h] \centering\footnotesize
\addtolength{\tabcolsep}{-4pt}
\begin{tabular}{|cl|c|c|c|c|c|c|c|c|c|c|}\hline
&\# of years & 1 & 2  & 3 & 4 & 5 & 6 & 7 & 8 & 9 & 10 \\
&\# of stations & 229 & 306 & 370 & 441 & 484 & 542 & 595 & 643 & 807 & 874\\\hline
&\# of rows &3,969,&10,959,&18,614,&26,622,&35,862,& 49,115,&63,469,&79,032,&99,221,&124,001,\\
&&455&978&686&218&560&307&733&846&419&260\\\hline
CSV&size (GB)&0.2&0.6&1.1&1.6&2.1&2.9&3.8&4.8&5.9&7.4\\\hline
\multirow{2}{*}{PostgreSQL} &raw size (GB) & 0.3 & 0.8 & 1.4 & 2.0 & 2.7 & 3.7 & 4.9 & 6.1 & 7.7 & 11.0\\
&total size (GB) & 0.4 & 1.1 & 2.0 & 2.9 & 3.9 & 5.4 & 7.1 & 8.9 & 11.0 & 14.0\\\hline
Parquet&size (GB)&0.03&0.08&0.15&0.2&0.3&0.4&0.5&0.6&0.8&0.9\\\hline
\end{tabular}
\caption{NY weather stations from 2005 to 2014.}
        \label{tab:data-stat-ny-extra}
\end{table*}
\begin{table*}[!h] \centering\footnotesize
\addtolength{\tabcolsep}{-4.5pt}
\begin{tabular}{|cl|c|c|c|c|c|c|c|c|c|c|}\hline
&\multirow{2}{*}{states}  & DE, &  +NY &  +MD &  +NJ,&  +MA, &  +LA, &  +ME, &  +NH, &  +MS,SC, &  +KY, \\
&&GA&&&RI&CT&VT&WV&NC&ND&SD\\\hline
&\# of states & 2 & 3 & 4 & 6 & 8 & 10 & 12 & 14 & 17 & 19 \\
&\# of stations & 408 & 659 & 1120 & 1476 & 1875 & 2305 & 2669 & 3019 & 3508 & 4037\\\hline
&\# of rows&16,760,&32,470,&41,346,&51,610,& 66,842,&80,561,&92,550,&106,415,&121,216, &140,517,\\
&& 333 & 116 & 986 & 908 & 618 & 273 & 905 & 139 & 837 & 500\\\hline
CSV&size (GB)&0.9&1.9&2.5&3.1&4.0&4.8&5.5&6.4&7.2&8.3\\\hline
\multirow{2}{*}{PostgreSQL} &raw size (GB) & 1.2 & 2.4 & 3.1 & 3.9 & 5.1 & 6.1 & 7.1 & 8.1 & 9.2 & 10.0\\
&total size (GB) & 2.0 & 4.1 & 5.3 & 6.5 & 8.6 & 10.0 & 12.0 & 14.0 & 16.0 & 18.0\\\hline
Parquet&size (GB)&0.1&0.2&0.3&0.4&0.5&0.6&0.7&0.8&0.9&1.1\\\hline
\end{tabular}
\caption{Weather data for 1--19 states in 2012.}
        \label{tab:data-stat-w2012-extra}
\end{table*}

Here, CSV is the size of the data in CSV format; PostgreSQL (raw size) is the size of the data itself stored in PostgreSQL reported by the \verb|pg_relation_size| function; %
PostgreSQL (total size) is the size of the total data (including the index)
    stored in PostgreSQL reported by the \verb|pg_total_relation_size|
    function; and Parquet is the size of the data in the Apache Parquet format, used by Apache Spark.

\newpage

\vskip 0.2in
\bibliography{temp}
\bibliographystyle{theapa}

\end{document}